\documentclass[11pt]{article}

\usepackage{amsthm,amsmath,amssymb}
\usepackage[bookmarks=true,hypertexnames=false,pagebackref]{hyperref}
\usepackage{thmtools}
\usepackage{thm-restate}

\hypersetup{colorlinks=true, citecolor=blue, linkcolor=red,
  urlcolor=blue}
\usepackage{newpxtext}
\usepackage[libertine,vvarbb]{newtxmath}

\usepackage{fullpage}
\usepackage{hyperref}
\usepackage{cleveref}
\usepackage[utf8]{inputenc}
\usepackage{framed}
\usepackage{enumerate}
\usepackage{graphicx}
\usepackage{float} 
\usepackage{subfigure}

\usepackage[ruled,vlined]{algorithm2e}
\usepackage{algorithmic}

\usepackage{tikz}
\usepackage{complexity}
\usepackage{ulem}

\newtheorem{theorem}{Theorem}[section]
\newtheorem{claim}[theorem]{Claim}

\newtheorem{corollary}[theorem]{Corollary}
\newtheorem{lemma}[theorem]{Lemma}
\newtheorem{definition}[theorem]{Definition}

\newtheorem{proposition}[theorem]{Proposition}

\newcommand{\mutualent}{\mathrm I}

\newcommand{\X}{\mathcal{X}}

\renewcommand{\hat}{\widehat}

\newcommand{\calD}{\mathcal{D}}

\newcommand{\calR}{\mathcal{R}}
\newcommand{\calA}{\sfM}

\newcommand{\Adv}{\mathrm{Adv}}

\newcommand{\zo}{\{-1,1\}}

\newcommand{\ve}{\varepsilon}
\newcommand{\epsl}{\epsilon}
\newcommand{\de}{\delta}
\newcommand{\lam}{\lambda}

\newcommand{\bbE}{\mathbb E}
\newcommand{\Var}{\mathbb Var}
\newcommand{\Ent}{\mathrm H}
\newcommand{\Ber}{\text{Ber}}
\newcommand{\sfM}{\mathsf{M}}
\newcommand{\Imi}{\mathsf{Im}}
\newcommand{\calG}{\mathcal{G}}
\newcommand{\B}{B}
\newcommand{\e}{a}
\newcommand{\Apr}{\mathsf{Apr}}
\newcommand{\imi}{im}
\newcommand{\apr}{apr}

\newcommand{\Deltasm}{\Delta}
\newcommand{\Deltalg}{\Gamma}
\newcommand{\Cb}{\zeta}
\newcommand{\I}{\mathrm I}
\newcommand{\In}{\mathcal I}
\newcommand{\Ra}{\mathcal R}
\newcommand{\one}{\mathbb 1}
\newcommand{\an}{o}
\newcommand{\sfO}{\mathsf{O}}

\newcommand{\Ee}{A}
\newcommand{\indep}{\perp\!\!\! \perp}
\newcommand{\tcoins}{t\text{-Coins Problem}}
\newcommand{\simtcoins}{\text{Simultaneous }t\text{-Coins Problem}}
\newcommand{\domin}{t}

\newcommand{\mic}{MIC}
\newcommand{\pos}{p}
\newcommand{\miccoin}{MIC_{cond}}

\newcommand{\warn}[1]{{\color{red}{#1}}}

\newcommand{\Jnote}[1]{[{\color{black}Jiapeng:}{ \color{red}#1}] }

\definecolor{mypink}{RGB}{219, 48, 122}
\title{A New Information Complexity Measure for Multi-pass Streaming with Applications}
\author{Mark Braverman\thanks{Princeton University. Email: \texttt{mbraverm@cs.princeton.edu}}
\and
Sumegha Garg\thanks{Rutgers University. Email: \texttt{sumegha.garg@rutgers.edu}}
\and
Qian Li\thanks{Shenzhen lnternational Center For
Industrial and Applied Mathematics,
Shenzhen Research Institute of Big
Data. Email: \texttt{liqian.ict@gmail.com}}
\and
Shuo Wang\thanks{Shanghai Jiao Tong University. Email: \texttt{s.wangg2002@gmail.com}}
\and
David P. Woodruff\thanks{Carnegie Mellon University. Email: \texttt{dwoodruf@andrew.cmu.edu}}
\and
Jiapeng Zhang\thanks{University of Southern California. Email: \texttt{jiapengz@usc.edu}}
}
\date{}
\allowdisplaybreaks
\begin{document}
\maketitle

\thispagestyle{empty}

\begin{abstract}
We introduce a new notion of information complexity for multi-pass streaming problems and use it to resolve several important questions in data streams:
\begin{enumerate}
\item In the coin problem, one sees a stream of $n$ i.i.d. uniform bits and one would like to compute the majority with constant advantage. We show that any constant pass algorithm must use $\Omega(\log n)$ bits of memory, significantly extending an earlier $\Omega(\log n)$ bit lower bound for single-pass algorithms of Braverman-Garg-Woodruff (FOCS, 2020). 
This also gives the first $\Omega(\log n)$ bit lower bound for the problem of approximating a counter up to a constant factor in worst-case turnstile streams for more than one pass.

\item In the needle problem, one either sees a stream of $n$ i.i.d. uniform samples from a domain $[t]$,
or there is a randomly chosen ``needle" $\alpha \in[\domin]$ for which each item independently is chosen to equal $\alpha$ with probability $p$, and is otherwise uniformly random in $[\domin]$. The problem of distinguishing these two cases is central to understanding the space complexity of the frequency moment estimation problem in random order streams. We show tight multi-pass space bounds for this problem for every $p < 1/\sqrt{n \log^3 n}$, resolving an open question of Lovett and Zhang (FOCS, 2023); even for $1$-pass our bounds are new. To show optimality, we improve both lower and upper bounds from existing results.
\end{enumerate}
Our information complexity framework significantly extends the toolkit for proving multi-pass streaming lower bounds, and we give a wide number of additional streaming applications of our lower bound techniques, including multi-pass lower bounds for $\ell_p$-norm estimation, $\ell_p$-point query and heavy hitters, and compressed sensing problems. 
\end{abstract}

\clearpage
\setcounter{page}{1}

\section{Introduction}\label{sec:intro}
Streaming problems with stochastic inputs have been popularly studied in the streaming community \cite{GM07,ToC16,Andoni08,CJP08,GM09,DBLP:conf/esa/CrouchMVW16,SODA2020randomordermatching,braverman2020coin,lovett2023streaming}, which have applications to diverse areas including learning theory \cite{DBLP:conf/focs/Raz16,Sharan19,diakonikolas2019communication,Colt22} and cryptography \cite{dinur2016memory,tcc-2018-28986,DBLP:conf/eurocrypt/JaegerT19,euro2019}. In this setting, one sees a stream of i.i.d. samples from some underlying distribution. As the samples are i.i.d., this is a special case of the well-studied random order streaming model. In this paper, we will consider streaming problems (with stochastic inputs) that are allowed multiple passes over their input. Surprisingly, even the most basic problems in data streams are not resolved in the stochastic setting. We discuss two such problems below.

\paragraph{Coin Problem.} If the stream is $X_1, X_2, \ldots, X_n$, with each $X_i$ being independently and uniformly drawn from $\{-1,1\}$, then the {\it coin problem} is to compute the sum of these bits up to additive error $O(\sqrt{n})$, which gives a non-trivial advantage for estimating the majority. This can be solved trivially in $O(\log n)$ bits of memory by storing the sum of input bits, and was recently shown to require $\Omega(\log n)$ bits of memory in \cite{braverman2020coin}. However, if we allow two or more passes over the stream, the only known lower bound is a trivial $\Omega(1)$ bits. 

Naturally, the coin problem is closely related to the fundamental question of counting the number of elements in a data stream --  maintain a counter $C$ in a stream under a sequence of updates of the form $C \leftarrow C + 1$ or $C \leftarrow C-1$, which is arguably the most basic question you could ask. More generally, one would like to approximate $C$ up to a constant multiplicative factor. One can solve this exactly using $\lceil \log_2 n \rceil$ bits of memory, where $n$ is the length of the stream.  This bound is tight for deterministic algorithms even if you allow a large poly$(n)$ approximation factor~\cite{ABJSSWZ22}. It is also tight for constant factor approximation randomized $1$-pass algorithms, via a reduction from the augmented indexing problem \cite{KNW10}. In fact, the lower bound of \cite{braverman2020coin} for the coin problem implies that the randomized $1$-pass lower bound holds even when the algorithm succeeds with high probability over an input of $i.i.d.$ coin flips. Despite our understanding of $1$-pass algorithms, we are not aware of an $\Omega(\log n)$ lower bound for this basic problem of approximate counting for more than one pass, even for worst-case streams, once two or more passes are allowed. For $k$-pass streaming algorithms, a weaker lower bound of $\Omega((\log n)^{1/k})$ can be derived using communication complexity lower bounds for the Greater-Than function\footnote{Briefly, given inputs $x,y\in [n^{0.1}]$ to the two-player CC problem for Greater-than function, Alice adds $x\cdot  n^{0.9}$ 1s to the stream and Bob adds $y\cdot n^{0.9}$ number of $-1$s. Determining the sign of $x-y$, or estimating $x-y$ up to additive error $n^{0.9}$ requires $\Omega((\log n)^{1/k})$ randomized $k$-round communication since it solves Greater-Than.}~\cite{miltersen1995data,viola2015communication}. 
There is work on approximate counting lower bounds of \cite{AHL16}, which analyzes a single linear sketch and therefore can only apply to single pass streaming algorithms. It also requires at least exponentially long streams in the value $n$ of the final count, despite the count being $O(n)$ in absolute value at any time during the stream.

\paragraph{Needle problem.} Here the goal is to distinguish between streams $X_1,\cdots,X_n$ sampled from two possible underlying distributions: let $\domin=\Omega(n)$,
\begin{itemize}
\item \textbf{Uniform distribution $ \boldsymbol{D_0}$:} each $X_i$ is picked independently and uniformly at random from the domain $[\domin]$, and
%Sample $t$ uniform elements from $[n]$.
\item \textbf{Needle distribution $ \boldsymbol{D_1}$:}  the distribution first uniformly samples an element $\alpha$ from $[\domin]$ (we call it the \textit{needle}). Then, each item $X_i$ independently with probability $p$ equals $\alpha$, and otherwise is sampled uniformly from $[\domin]$. 
\end{itemize}
%The central concern in the context of the needle problem is to determine the optimal sample-space trade-off:
%In the needle context, we care about this question: 
%\begin{center}
%\textbf{Question:} 
%Given a stream of $n$ samples, how many bits of memory are needed to distinguish between $ \boldsymbol{D_0}$ and $ \boldsymbol{D_1}$? 
Lovett and Zhang \cite{lovett2023streaming} shows a lower bound of $\Omega \left (\frac{1}{p^2 n \log n} \right )$ bits for any constant number of passes to distinguish $\boldsymbol{D_0}$ from $\boldsymbol{D_1}$, and left it as an open problem to improve this bound. Using an algorithm for frequency moment estimation, the work of Braverman et al. \cite{braverman2014optimal} shows that for $p = o \left (\frac{1}{n^{2/3}} \right )$, there is a single-pass upper bound of $O \left (\frac{1}{p^2n} \right )$ bits. For larger $p$, the best known algorithm is to store the previous $O \left (\frac{1}{p^2n}  \right )$ items in a stream, using $O \left (\frac{\log n}{p^2n} \right )$ bits, and check for a collision. Thus, for every $p$ there is a gap of at least $\log n$ in known upper and lower bounds, and for the important case of $p> \Omega \left (\frac{1}{n^{2/3}} \right )$, the gap is $\Theta(\log^2 n)$. We note that the work of Lovett and Zhang also requires $t = \Omega(n^2)$ in its lower bounds\footnote{Here we choose the domaint to be $[t]$ and $n$ the number of samples to be consistent with our notation for the coin problem. Unfortunately, this is exactly the opposite notation used in \cite{lovett2023streaming}}, which limits its applications to frequency moment estimation. 

The needle and coin problems are related to each other if $p = \Theta \left (\frac{1}{\sqrt{n}} \right )$ and when the algorithm has access to a random string that is not counted towards the space complexity\footnote{This is referred to as the public coin model in communication complexity, and we may naturally view this problem as an $n$-player communication game.}.
%We note that if one charges to store randomness, there is an $\Omega(\log n)$ bit lower bound from the communication complexity of the Equality problem \cite{alon1996space}.
One can randomly hash each universe element in a stream in the needle problem to $\{-1,1\}$ -- under the needle distribution $ \boldsymbol{D_1}$, the absolute value of the sum of bits is likely to be an additive $\Theta(\sqrt{n})$ larger than in the uniform distribution $ \boldsymbol{D_0}$, and thus the coin problem is at least as hard as the needle problem. However, the needle problem for $p = \Theta \left (\frac{1}{\sqrt{n}} \right )$ could be strictly easier than the coin problem.  

We stress that the coin and needle problems are arguably two of the most fundamental problems in data streams. Indeed, a vast body of work has considered estimating the $q$-th frequency moment $F_q = \sum_{i=1}^\domin |x_i|^q$, for an underlying $\domin$-dimensional vector $x$ undergoing positive and negative updates to its coordinates. If $\domin = 1$, this problem is at least as hard as the coin problem. The lower bound of \cite{braverman2020coin} thus gave the first $\Omega(\log n)$ lower bound for single pass $F_2$-estimation in the bounded deletion data stream model \cite{JW18}, even when one does not charge the streaming algorithm for storing its randomness. As the lower bound of \cite{braverman2020coin} was actually an information cost lower bound, it gave rise to the first direct sum theorem for solving multiple copies of the $F_2$-estimation problem, showing that solving $r$ copies each with constant probability requires $\Omega(r \log n)$ bits in this data stream model. Lower bounds for the coin problem were shown to imply additional lower bounds in random order and bounded deletion models for problems such as point query, and heavy hitters; see  \cite{braverman2020coin} for details. For $q > 2$, if we set $p = \Theta \left (\frac{1}{n^{1-1/q}} \right )$ in the needle problem, then it is not hard to see that $F_q$ differs by a constant factor for distributions {\bf $\boldsymbol{D}_0$} and {\bf $\boldsymbol{D}_1$} with large probability. In this case, the lower bound of \cite{lovett2023streaming} gives an $\Omega \left (\frac{n^{1-2/q}}{\log n} \right )$ lower bound for frequency moment estimation in the random order insertion-only data stream model for any constant number of passes, provided $t = \Omega(n^2)$. There is also a long sequence of prior work on this problem \cite{Andoni08,ToC16, DBLP:conf/esa/CrouchMVW16, DBLP:journals/algorithmica/McGregorPTW16}, which obtains polynomially worse lower bounds (though does not require $t = \Omega(n^2)$). The single-pass arbitrary order stream $O \left (n^{1-2/q} \right )$ upper bound of \cite{braverman2014optimal} for $q > 3$ matches this up to a logarithmic factor, and a central question in data streams is to remove this logarithmic factor, as well as the requirement that $t = \Omega(n^2)$. 

\subsection{Our Contributions}
\label{sec: contribution}
We give a new multi-pass notion of information complexity that gives a unified approach to obtain lower bounds for the coin and the needle problem for any number $k$ of passes. We note that the measure of information we use is a generalization of the notion of information complexity for $k = 1$ pass given in \cite{braverman2020coin}. Namely, we define the $k$-pass information complexity notion by 
%\qian{I use the needle ic notion rather than the coin ic notion, since the needle ic notion is larger than the coin ic notion, and we may not be able to use the coin ic notion to lower bound the needle problem.}
\begin{align*}
\mic(\sfM,\mu)&:=\sum_{i=1}^{k}\sum_{j=1}^n\sum_{\ell=1}^i \mutualent\left(\sfM_{(i,j)};X_{\ell}\mid \sfM_{(\leq i,\ell-1)}, \sfM_{(\leq i-1,j)}\right)\\ &+\sum_{i=1}^{k}\sum_{j=1}^n\sum_{\ell=j+1}^n \mutualent\left(\sfM_{(i,j)};X_{\ell}\mid \sfM_{(\leq i-1,\ell-1)}, \sfM_{(\leq i-1,j)}\right),
\end{align*}
where $(X_1,\cdots,X_n)\sim \mu$ and $\sfM_{(i,j)}$ represents the $j$-th memory state in the $i$-th pass. We will set $\mu$ to be the uniform distribution over $\zo^n$ in the coin problem and the uniform distribution $\boldsymbol{D}_0$ in the needle problem. When $\mu$ is clear from the context, we will drop it from the notation and write $\mic(\sfM)$. The primary challenge in establishing lower bounds for multi-pass streaming algorithms arises from the fact that the streaming data loses its independence when multiple passes are used. To mitigate this, the idea of $\mic(\sfM)$ is to capture some residual 
 independence by carefully fixing some memory states in the previous pass (see Section \ref{sec:upper}). To make this notion useful, we show that $\mic(\sfM)$ is upper bounded by $2ksn$ as the following lemma (see proof in Section \ref{sec:upper}): 
\begin{restatable}[]{lemma}{multiupper}
\label{lem:Multi-passupperbound}
Assuming that $(X_1,X_2,\cdots,X_n)$ are drawn from a product distribution $\mu$. Then, for any $k$-pass streaming algorithm $\calA$ with memory size $s$ running on input stream $X_1,\cdots,X_n$, it holds that: 
\[
\mic(\calA,\mu)\leq 2 k sn.
\]
\end{restatable}

Note that this multi-pass information complexity notion applies to any streaming problem as long as it is defined on a product distribution. As we will see in the technical sections, this notion is useful for proving multi-pass streaming lower bounds via various approaches, such as round elimination as well as randomized communication complexity.

 Just as the measure of information in \cite{braverman2020coin} was crucial for $1$-pass applications, such as the amortized complexity of approximate counting in insertion streams \cite{ahny22}, we will show our notions have a number of important applications and can be used to obtain multi-pass lower bounds for both the coin and needle problems. We will define and motivate our information complexity notion more below, but we first describe its applications.

\subsubsection{The Coin Problem}
We give tight lower bounds on the information complexity of the coin problem, significantly extending the results of \cite{braverman2020coin} for the $1$-pass setting to the multi-pass setting. We then give a new multi-pass direct sum theorem for solving multiple copies, and use it for streaming applications. 

\paragraph{Multi-Pass Coin Problem.}
We give the first multi-pass lower bound for the coin problem.
\begin{theorem}\label{thm:coin}(Multi-Pass Coin Problem)
    Given a stream of $n$ i.i.d. uniformly random bits, any $k$-pass streaming algorithm which outputs the majority of these bits with probability $1-\gamma$ for a small enough constant $\gamma > 0$, requires $\Omega(\frac{\log n}{k})$ bits of memory. 
\end{theorem}
Theorem \ref{thm:coin} is a significant strengthening of \cite{braverman2020coin} which held only for $k = 1$ pass. Although the work of \cite{BGZ21} allows for a larger bias on its coins, it also held only for $k = 1$ pass. 

As discussed before, we can interpret the coins as updates in $\{-1,1\}$ to a counter $C$ initialized to $0$. Adjoining a prefix of $\alpha \sqrt{n}$ $1$s to the stream for a large enough constant $\alpha > 0$, we have that by bounds on the maximum deviation for a $1$-dimensional random walk that $C$ will be non-negative at all points during the stream, which corresponds to the {\it strict turnstile streaming model} -- where one can only delete previously inserted items. We also have that the final value of $C$ will deviate from its expectation $\alpha \sqrt{n}$ by an additive $\Omega(\sqrt{n})$ with constant probability, that is, it is anti-concentrated. Consequently, Theorem \ref{thm:coin} implies the following.

\begin{theorem}\label{thm:counter}(Multi-Pass Counter in Strict Turnstile Streams) 
    Any $k$-pass strict turnstile streaming algorithm which counts the number of insertions minus deletions in a stream of length $n$ up to a small enough constant multiplicative factor and with probability at least $1-\gamma$ for a small enough constant $\gamma > 0$, requires $\Omega(\frac{\log n}{k})$ bits of memory.  
\end{theorem}
By constant multiplicative factor, we mean to output a number $C'$ for which $(1-\epsilon)C \leq C' \leq (1+\epsilon)C$ for a constant $\epsilon > 0$. For insertion-only streams where no deletions of items are allowed, non-trivial algorithms based on Morris counters achieve $O(\log \log n)$ bits \cite{M,F,G10,NY22}. We rule out any non-trivial algorithm for strict turnstile streams for any constant number of passes. %Estimating $F_2$ up to a constant factor in a turnstile stream is at least as hard as approximating a counter, so we obtain the first $\Omega(\log n)$ constant pass lower bound for $F_2$-estimation when not charging the algorithm for its randomness\footnote{If you do charge for the randomness, an $\Omega(\log n)$ lower bound also holds via a reduction from Equality \cite{alon1996space}. We note that that is not an information lower bound, so does not imply an $\Omega(r \log n)$ bound for solving $r$ copies, unlike our lower bound. }

%Here, the probability is over the randomness used by the algorithm. 
To obtain further applications, we first show a direct sum theorem for multi-pass coins.

\begin{theorem}\label{thm:directSum} (Direct Sum for Multi-Pass Counter) 
Suppose $k < \log n$ and $t < n^c$ for a sufficiently small constant $c > 0$.
Given $t$ independent streams each of $n$ i.i.d. uniformly random bits, any $k$-pass streaming algorithm which outputs a data structure such that, with probability $1-\gamma$ for a small enough constant $\gamma > 0$ over the input, the algorithm's randomness, and over a uniformly random $j \in [t]$, outputs the majority bit of the $j$-th stream, requires $\Omega \left (\frac{t \log n}{k} \right )$ bits of memory. 
\end{theorem}

As an example application of this theorem, in \cite{MWY13} the following problem was studied for a real number $p \in [0,2]$: given vectors $v_1, \ldots, v_t \in \{-\textrm{poly}(n), \ldots, \textrm{poly}(n)\}^d$, estimate a constant fraction  of the $\|v_1\|_p, \ldots, \|v_t\|_p$ up to a sufficiently small constant multiplicative factor with constant probability, where for a $d$-dimensional vector $y$, the $p$-norm\footnote{For $p < 1$ the quantity $\|v\|_p$ is not a norm, but it is still a well-defined quantity. With a standard abuse of notation, we will refer to it as a $p$-norm.} $\|y\|_p = \left (\sum_{i=1}^d |y_i|^p \right )^{1/p}.$ We will refer to this problem as {\sf Multi-$\ell_p$-Estimation}. 
The best upper bound is $O(t \cdot \log n)$, which follows just by solving each instance independently with constant probability and using $O(\log n)$ bits \cite{KNW10}. An $\Omega(t \log \log n + \log n)$ randomized lower bound follows for any $O(1)$-pass streaming algorithm by standard arguments\footnote{The $\Omega(t \log \log n)$ bound follows just to record the output, while the $\Omega(\log n)$ lower bound follows via a reduction from the Equality problem, as in \cite{alon1996space}.}. 
We note that if we do not charge the streaming algorithm for its randomness, then the $O(1)$ pass lower bound for {\sf Multi-$\ell_p$-Estimation} is an even weaker $\Omega(t \log \log n)$.

By using Theorem \ref{thm:directSum} and having each vector $v_i$ in the {\sf Multi-$\ell_p$-Estimation} problem correspond to a single counter, we can show the following: 

\begin{theorem}[Multi-Pass Multi-$\ell_p$-Estimation]\label{thm:multiLpEstimation}
Suppose $k < \log n$ and $t < n^c$ for a sufficiently small constant $c > 0$.
Any $k$-pass streaming algorithm which solves the {\sf Multi-$\ell_p$-Estimation} Problem on $t$ instances of a stream of $n$ updates for each vector, solving each $\ell_p$-norm estimation problem up to a small enough constant factor with probability $1-\gamma$ for a sufficiently small constant $\gamma$ , requires $\Omega \left (\frac{t \log n}{k} \right )$ bits of memory. 
\end{theorem}

Another important streaming question is the {\sf $\ell_2$-Point Query Problem}: given an underlying $d$-dimensional vector $x \in \{-\textrm{poly}(n), \ldots, \textrm{poly}(n)\}^d$ that undergoes a sequence of positive and negative additive updates to its coordinates, for each $j \in \{1, \ldots, d\}$, one should output $x_j$ up to an additive error $\epsilon \|x\|_2$ with constant probability. Related to this question is the {\sf $\ell_2$-Heavy Hitters Problem} which asks to output a set $S$ which (1) contains all indices $i$ for which $x_i^2 \geq \epsilon \|x\|_2^2$, and (2) does not contain any index $i$ for which $x_i^2 \leq \frac{\epsilon^2}{2} \|x\|_2^2$. Further, for all $i \in S$, one should output an estimate $\hat{x}_i$ with $|\hat{x}_i - x_i| \leq \epsilon \|x\|_2$.
In \cite{braverman2020coin}, for both of these problems an $\Omega(\epsilon^{-2} \log n)$ memory lower bound was shown for single-pass algorithms on length-$n$ streams, which improved the previous best known $\Omega(\epsilon^{-2} \log d)$ lower bounds when the stream length $n$ is much larger than the dimension $d$ of the input vectors. Notably, the lower bounds in \cite{braverman2020coin} only hold for single pass algorithms. 

By having each coordinate of an underlying vector $x$ correspond to a counter, we can also use Theorem \ref{thm:directSum} to solve the {\sf $\ell_2$-Point Query Problem} and the {\sf $\ell_2$-Heavy Hitters Problem}. Here we also use that the Euclidean norm of the underlying vector is concentrated. 

\begin{theorem}[Multi-Pass Point Query and Heavy Hitters] \label{thm:PQHH}
Suppose $k < \log n$ and $\epsilon^{-2} < n^c$ for a sufficiently small constant $c > 0$.
Any $k$-pass streaming algorithm which, with probability $1-\gamma$ for a sufficiently small constant $\gamma > 0$, solves the {\sf $\ell_2$-Point Query Problem} or the {\sf $\ell_2$-Heavy Hitters Problem} on a vector $x \in \{-\textrm{poly}(n), \ldots, \textrm{poly}(n)\}^d$ in a stream of $n$ updates,  requires $\Omega \left (\frac{\epsl^{-2} \log n}{k} \right )$ bits of memory. 
\end{theorem}

Our $\Omega(\epsilon^{-2} \log n)$ bit lower bound for the {\sf $\ell_2$-Heavy Hitters Problem} can be applied to the {\sf Sparse Recovery Problem} in compressed sensing (see, e.g., \cite{PW11,GLPS17}), which involves an input vector $x \in \{-\textrm{poly}(n), \ldots, \textrm{poly}(n)\}^d$ in a stream, and asks to output an $r$-sparse vector $\hat{x}$ for which
\begin{eqnarray}\label{eqn:recovery}
\|\hat{x}-x\|_2 \leq \Delta \cdot \|x-x_r\|_2,
\end{eqnarray}
where $x_r$ is $x$ with all but the top $r$ coordinates set to $0$. Here $\Delta > 1$ is any fixed constant. 

A standard parameter of sparse recovery is the Signal-to-Noise Ratio (SNR), which is defined to be $\frac{\|x_r\|_2^2}{\|x\|_2^2}$. 
The SNR is at most $1$, and if it is $1$, there is a trivial $\Omega(r (\log n + \log d))$ bit lower bound. Indeed, since the guarantee of (\ref{eqn:recovery}) has multiplicative error, we must have $\hat{x} = x = x_r$ in this case, and it takes $\Omega(r (\log n + \log d))$ bits to encode, for each of the $r$ non-zero locations in $x$, its location and its value. However, when the SNR is a constant bounded away from $1$, this encoding argument no longer applies. Indeed, while one can show an $\Omega(r \log d)$ bits lower bound to encode the identities of $r$ locations, each of their values can now be approximated up to a small multiplicative constant, and so encoding their values requires only $\Omega(r \log \log n)$ bits. 

While an $\Omega(r + \log \log n)$ {\it measurement} lower bound is known for multi-pass streaming algorithms \cite{PW13} for constant SNR bounded away from $1$, perhaps surprisingly in the data stream model, an $\Omega(r \log n)$ bits lower bound for streams of length $n$ and SNR bounded away from $1$ was unknown. As our lower bound for the {\sf $\ell_2$-Heavy Hitters Problem} only requires recovering a large constant fraction of the $\ell_2$-heavy hitters, all of which are comparable in magnitude in our hard instance, and the Euclidean norm is concentrated, we in fact can obtain a lower bound for the {\sf Sparse Recovery Problem} even if the SNR is a constant bounded away from $1$. We note that there is an $O(\log \log n)$-pass streaming algorithm which uses $O(r \log n \log \log n)$ bits of memory to solve the sparse recovery problem for any SNR, see \cite{NSW018} which builds upon \cite{IPW11} (see the text after the proof of Theorem 3.7 in \cite{IPW11} on how to obtain an exactly $r$-sparse output). Our lower bound is thus tight up to  $\textrm{poly}(\log \log n)$ factors.

\begin{theorem}\label{thm:compressedSensing}
(Bit Complexity of Sparse Recovery). 
Suppose $k < \log n$ and $r < n^c$ for a sufficiently small constant $c > 0$.
Any $k$-pass streaming algorithm which, with probability $1-\gamma$ for $\gamma > 0$ a small constant, solves the {\sf Sparse Recovery Problem} for constant SNR in $(0,1)$, requires
$\Omega(\frac{r \log n}{k})$ bits of memory. 
\end{theorem}

\subsubsection{The Needle Problem}
Lovett and Zhang \cite{lovett2023streaming} recently showed the following lower bound for the needle problem. 

\begin{theorem}[\cite{lovett2023streaming}]\label{thm:lz23}
    Any $k$-pass streaming algorithm $\calA$ which distinguishes between the uniform and needle distributions with high probability, where $p$ denotes the needle probability, $n$
    %=O(\sqrt{n})$ 
    the stream length, and $s$ the space, satisfies $k sp^2n \log(n) = \Omega(1)$, provided the domain size $t = \Omega(n^2)$.
\end{theorem}

While this lower bound is nearly tight, it was conjectured by \cite{Andoni08,ToC16,DBLP:conf/esa/CrouchMVW16,lovett2023streaming} that the additional $\log(n)$ term can be removed, and it also was plausible that the $t = \Omega(n^2)$ restriction could be removed. This conjecture is for good reason, as for $n = \Theta(\domin)$ and $p \ll  \frac{1}{n^{2/3}}$ and $k = 1$, an upper bound for estimating frequency moments of \cite{braverman2014optimal} shows that $s p^2 n = O(1)$. Indeed, the upper bound of \cite{braverman2014optimal} shows how to estimate $F_q = \sum_{i=1}^t f_i^q$ up to an arbitrarily small but fixed constant factor in $O(t^{1-2/q})$ bits of memory and a single pass, for any $q > 3$. Notice that in distribution $\boldsymbol{D_0}$, we could choose a proper $n= \Theta(t)$ such that $F_q = n +o(n)$ with high probability. On the other hand, for distribution $\boldsymbol{D_1}$, we have that $F_q > (p \cdot n)^q$, and so if $p = \Theta(1/n^{1-1/q})$, these two distributions can be distinguished by the algorithm of \cite{braverman2014optimal}. In this case the conjecture would say $s = \Omega(1/(n p^2)) = \Omega(n^{1-2/q}) = \Omega(t^{1-2/q})$, which matches the space upper bound of \cite{braverman2014optimal}. 

%\begin{conjecture}[\cite{Andoni08,ToC16,DBLP:conf/esa/CrouchMVW16,lovett2023streaming}]
%\label{conj:LZ}
%    Any $\ell$-pass streaming algorithm which distinguishes between the uniform and needle distributions with high probability, where $p$ is the needle probability, $t$ the number of samples and $s$ the space, satisfies $\ell sp^2t = \Omega(1)$.
%\end{conjecture}

We resolve this conjecture. As a consequence, several other streaming lower bounds mentioned by \cite{MPTW12,DBLP:conf/esa/CrouchMVW16, lovett2023streaming} can be improved automatically. Also, our results also imply that the frequency estimation problem for $q>2$ is as hard in the random order model as in the arbitrary order model (both are $\Omega(t^{1-2/q})$). 

%On the other hand, our proof technique could be of independent interests and give several more applications. In order to prove the main theorem, we introduce a novel notion of information complexity for multi-pass streaming algorithms. This definition is a generalization of previous IC for one-pass streaming algorithms by Braverman et al \cite{braverman2020coin}. By contrast, previous approaches \cite{Andoni08,ToC16,DBLP:conf/esa/CrouchMVW16, lovett2023streaming} were mainly built on black-box reduction to communication complexity (unique disjointness problem). We believe this multi-pass IC may give several more applications in both learning theory and cryptography. Please see Section \ref{sec: tech_contribution} for more discussions. To this end, we now state our main theorem.

\begin{restatable}[Multi-Pass Lower Bound for the Needle Problem]{theorem}{mainthm}
\label{thm:streaming_lower_bound}
Any $k$-pass streaming algorithm $\sfM$ with space $s$ that distinguishes $ \boldsymbol{D_0}$ and $ \boldsymbol{D_1}$ with high probability satisfies
$k sp^2n = \Omega(1)$, where $p$ denotes the needle probability and $n\leq t/100$ denotes the number of samples. 
\end{restatable}

If we use the algorithm of \cite{braverman2014optimal} to the needle problem by the reduction discussed above, we can conclude that our lower bound for the needle problem is tight when $p\ll \frac{1}{n^{2/3}}$. However, we further improve the upper bound by giving a new algorithm and show that:
\begin{theorem}[Improved Upper Bound for the Needle Problem]
    There exists a one-pass streaming algorithm that distinguishes $\boldsymbol{D}_0$ and $\boldsymbol{D}_1$ with high probability and uses $O(\frac{1}{p^2n})$ bits of space when $p\leq \frac{1}{\sqrt{n\log^3 n}}$. 
\end{theorem}
Our upper bound improves upon  \cite{braverman2014optimal}, and shows that our lower bound for the needle problem is indeed tight for any $p \leq \frac{1}{\sqrt{n\log^3 n}} $. The remaining gap only exists in the range of $p > \frac{1}{\sqrt{n\log^3 n}} $. 

In fact, for $p \geq \frac{1}{\sqrt{n}}$, for the $n$-player communication problem, where each player has a stream item and the players speak one at a time from left to right in the message passing model (see, e.g., \cite{jw23} for upper bounds for a number of problems in this model), we show that the problem can be solved by having each player send at most $O((\log \log n) (\log \log \log n))$ bits to the next player. It is not quite a streaming algorithm, as the players need to know their identity, but would be a streaming algorithm if we also allow for a clock, so that we know the number $i$ for the $i$-th stream update, for each $i$. 

\begin{theorem}[Upper Bound for Communication Game]\label{thm:sqrtn}
    There exists an $n$-player one-round communication protocol that distinguishes $\boldsymbol{D}_0$ and $\boldsymbol{D}_1$ with high probability and uses $O((\log \log n )(\log \log \log n))$ bits of space for any $p\geq \frac{1}{\sqrt{n}}$. 
\end{theorem}
Theorem \ref{thm:sqrtn} shows that for $p = \frac{1}{\sqrt{n}}$, the needle problem is strictly easier than the coin problem, and thus the abovementioned algorithm for the needle problem, by first reducing to the coin problem, is suboptimal. Indeed, in the same communication model or for streaming algorithms with a clock, our $\Omega \left (\frac{\log n}{k} \right )$ lower bound in Theorem \ref{thm:coin} applies. Thus, for a constant number $k$ of passes, the coin problem requires $\Omega(\log n)$ bits of memory whereas the needle problem with $p = \frac{1}{\sqrt{n}}$ can be solved with $O((\log \log n )(\log \log \log n))$ bits of memory, showing that there exists a separation between the two problems in the $n$-player communication model. 

{\bf Remark.} Note that our algorithm not only applies to the needle problem, but also could be adapted to a more general setting where the needle is randomly ordered while non-needle items could be in an arbitrary order with some constraints. We would further discuss this in Section \ref{sec:algo1}. 

%This theorem affirms the conjecture by Lovett and Zhang \cite{lovett2023streaming} and is indeed tight for a large range of $p$. Braverman et al. \cite{braverman2014optimal} provided a one-pass streaming algorithm with $O(n^{1-2/k})$ space that approximates $k$-th frequency moment for all $k>3$. If we adopt this algorithm for the needle problem, it could be an $sp^2t=O(1)$ algorithm for all $p$ in the range $t^{-0.999} \leq p \leq t^{-2/3}$.
%\fi

%\subsubsection{Applications}\label{sec:app}

  %
  %As discussed by \cite{DBLP:conf/esa%/CrouchMVW16,lovett2023streaming}, 
  Our improved lower bound for the needle problem can be used to obtain optimal lower bounds in the random order model for arguably the most studied problem in the data stream literature, namely, that of approximating the frequency moments. Starting with the work of Alon, Matias, and Szegedy \cite{alon1996space}, there has been a huge body of work on approximating the frequency moments in arbitrary order streams, see, e.g., \cite{ccc03,BJKS04,Jay09,IW05,bhuvanagiri2006simpler,MW10,andoni2011streaming,Gan11,ganguly_et_al:LIPIcs:2018:9062}, and references therein. As mentioned above, Braverman et al. \cite{braverman2014optimal} gave an upper bound of $O(t^{1-2/q})$ for constant approximation for all $q>3$, which is optimal for arbitrary order insertion streams. 

%\paragraph{Frequency moment estimation problem.} Frequency moment estimation is a classical streaming problem and has been extensively studied since the seminal work by Alon, Matias, and Szegedy \cite{alon1996space}. For streaming samples coming in an \textit{adversarial order}, following much research \cite{ccc03, gronemeier2009asymptotically,Jay09,IW05,bhuvanagiri2006simpler,MW10,andoni2011streaming,Gan11,ganguly_et_al:LIPIcs:2018:9062}, there are now matching upper and bounds of $\Theta(t^{1-2/k})$ on the space for the $q$-th frequency moment estimation. %On the other hand, a lot of algorithms were also proposed by \cite{IW05,bhuvanagiri2006simpler,MW10,andoni2011streaming,Gan11,ganguly_et_al:LIPIcs:2018:9062}. Finally, Braverman et al. \cite{braverman2014optimal} gave a upper bound of $O(n^{1-2/k})$ for all $k>3$, which matches the best lower bounds.%For streaming samples with an \textit{adversarial order}, a series of works by \cite{ccc03, gronemeier2009asymptotically,Jay09} proved the an $\Omega(n^{1-2/k})$ space lower bound for frequency estimation. On the other hand, a lot of algorithms were also proposed by \cite{IW05,bhuvanagiri2006simpler,MW10,andoni2011streaming,Gan11,ganguly_et_al:LIPIcs:2018:9062}. Finally, Braverman et al. \cite{braverman2014optimal} gave a upper bound of $O(n^{1-2/k})$ for all $k>3$, which matches the best lower bounds.

A number of works have also studied the frequency moment estimation problem in {\it randomly ordered} streams. While the $O(t^{1-2/q})$ bit upper bound of \cite{braverman2014optimal} still holds, we did not have a matching lower bound. Chakrabarti, Cormode, and McGregor \cite{ToC16} gave the first non-trivial $\Omega(t^{1-3/q})$ lower bound. A follow-up paper by Andoni et al. \cite{Andoni08} improved this lower bound to $\Omega(t^{1-2.5/q})$. Recently, a lower bound of $\Omega(n^{1-2/q}/\log n)$ was shown by Lovett and Zhang \cite{lovett2023streaming}  provided $t = \Omega(n^2)$. Since a stream of i.i.d. samples is automatically randomly ordered, our Theorem \ref{thm:streaming_lower_bound} resolves this long line of work, giving an $\Omega(t^{1-2/q})$ lower bound. We thus improve the lower bound of \cite{lovett2023streaming} by a logarithmic factor, and also remove the requirement that $t = \Omega(n^2)$. The application to frequency moments follows by applying our theorem with $t = \Theta(n)$ and $p = 1/n^{1-1/q}$ and arguing that the needle problem gives rise to a constant factor gap in the value of the $q$-th frequency moment in the two cases. We note that the work of \cite{gh09} claimed to obtain an $\Omega(t^{1-2/q})$ lower bound for frequency moment estimation in a random order, but was later retracted due to an error which has been pointed out in multiple places, e.g., \cite{DBLP:journals/algorithmica/McGregorPTW16} retracts its lower bounds and points out the error\footnote{This error has also been confirmed with the authors of \cite{gh09} in 2016, and no fix with their techniques seems possible.} in \cite{gh09}.

%via a reduction discussed by \cite{DBLP:conf/esa/CrouchMVW16}. 
%For streaming samples with a \textit{random order}, the upper bound ($O(n^{1-2/k})$) by Braverman et al could be applied directly. However, proving lower bounds in this setting is more challenging. Several papers \cite{Andoni08,ToC16,DBLP:conf/esa/CrouchMVW16, lovett2023streaming} conjectured the $\Omega(n^{1-2/k})$ lower bound should also hold for random order samples. In this direction, Chakrabarti et al \cite{ToC16} first gave an $\Omega(n^{1-3/k})$ lower bound. This bound was laterly improved to $\Omega(n^{1-2.5/k})$ by Andoni et al.\cite{Andoni08}. Until recently, an almost tight bound $\Omega(n^{1-2/k}/\log n)$ has proved by Lovett and Zhang \cite{lovett2023streaming}. Our main theorem affirms this tight lower bounds $\Omega(n^{1-2/k})$ via a reduction discussed by \cite{DBLP:conf/esa/CrouchMVW16}. 
%
%
%\vspace{1ex}
%\paragraph{Sub-sampling of the frequency moment estimation problem.} 

There are other related problems to frequency moment estimation that we also obtain improved lower bounds for, such as frequency moment estimation of sub-sampled streams. McGregor et al. \cite{MPTW12} studied streaming problems in a model where the stream comes in so rapidly that you only can see each element independently with a certain probability. Our 
Theorem \ref{thm:streaming_lower_bound} gives an optimal lower bound for this problem as well, via the reduction in \cite{MPTW12}. Another example concerns stochastic streaming problems such as collision probability estimation studied by Crouch et al. \cite{DBLP:conf/esa/CrouchMVW16}. They provided several lower bounds based on the needle lower bound of \cite{Andoni08}. Our Theorem \ref{thm:streaming_lower_bound} automatically improves their lower bounds via the same reductions. 

%In particular. Built on a needle lower bound  by Guha and Huang \cite{GH09}, McGregor et al. \cite{MPTW12} claimed a optimal lower bound for their problem. However, the journal version \cite{DBLP:journals/algorithmica/McGregorPTW16} retracted the lower bounds due a flaw in \cite{GH09}. Now Theorem \ref{thm:streaming_lower_bound} shows the lower bounds in \cite{MPTW12} still hold. 

%\vspace{1ex}
%\paragraph{Collision probability estimation problem.} Crouch et al. \cite{DBLP:conf/esa/CrouchMVW16} studied stochastic streaming problems such as collision probability estimation. They provided several lower bounds based on the needle lower bound by \cite{Andoni08}. Theorem \ref{thm:streaming_lower_bound} automatically improves their lower bounds by the same reductions. 

%\subsubsection{Potential Applications of the Multi-Pass IC}
%\label{sec: tech_contribution}

%Our proof is mainly based on the new notion of information complexity for multi-pass streaming algorithms. We believe this IC notion could give applications to two directions. 

\subsection{Technical Overview} 
%We give a high-level explanation of our proof in this section.

In Subsection \ref{subsec:singlecoinoverview}, we give a technical overview of the $\Omega\left(\frac{\log n}{k}\right)$ memory lower bound for any $k$-pass streaming algorithm that solves the coin problem, that is, computes the majority bit with at least 0.999 advantage when input is drawn from the uniform distribution on $\zo^n$. In Subsection \ref{subsec:tcoinoverview}, we also give a brief overview of the direct sum theorem for solving multiple instances of the coin problem simultaneously using a multi-pass streaming algorithm. In Subsection \ref{sec:overviewofneedle}, we discuss the techniques used in the proof of the $kp^2s n =\Omega(1)$ trade-off for the $k$-pass  needle problem with parameters: $n$ the length of data stream, $p$ the needle probability and $s$ the space complexity. In Subsection \ref{sec:algooverview}, we briefly introduce the idea of our new upper bounds and algorithms that solves the needle problem efficiently. 

\subsubsection{Multi-pass Streaming Lower Bound for the Coin Problem}\label{subsec:singlecoinoverview}
As in \cite{braverman2020coin}, we show an information-variance tradeoff for any $k$-pass streaming algorithm that computes the majority bit on a length-$n$ stream. Simple counting gives a $\log n$ memory streaming algorithm for computing majority, whereas \cite{braverman2020coin} remarked that to compute the majority with at least 0.999 advantage over the uniform distribution, the expected variance of the sum of input bits, conditioned on the output of the algorithm, should be a constant factor less than the maximum variance possible. However, the techniques of \cite{braverman2020coin} are especially tailored for one-pass algorithms. Let $\sfO$ be a one-pass streaming algorithm for computing majority. Then \cite{braverman2020coin} showed that at most one of the following two statements can be true (assume $\epsl,\delta>0$ to be appropriate constants):
\begin{enumerate}
\item \label{it:po1} $IC(\sfO)=\sum_{j=1}^n \sum_{\ell=1}^j \I(\sfO_j;X_\ell|\sfO_{\ell -1}) \le \delta n\log n$.
\item \label{it:po2} $\bbE_{o_n\sim \sfO_n}\left[\Var\left(\sum_j X_j | \sfO_n=o_n\right)\right] \le (1-\epsl)n$.
\end{enumerate}
In other words, if the output of algorithm $\sfO$ reduces the variance of the sum\footnote{We want some mathematical quantity to measure how much information the output of algorithm $\sfO$ has about the sum/count of the input bits, and variance turns out to be the ideal measure for \cite{braverman2020coin} as well as for our paper.}, then it needs to have high information cost. Let us look more closely at their information complexity measure: $X$ is a length-$n$ stream of $i.i.d.$ uniform bits and $\sfO_j$ is the random variable for the memory state after reading $j$ input bits. Thus, the information cost ($IC(\sfO)$) measures the mutual information between the $j$-th memory state and the $\ell$-th input bit ($\ell\le j$) conditioned on the previous memory state (summed over $j,\ell$). This information measure is only defined for one-pass algorithms.  Furthermore, $IC(\sfO)$ is a useful measure only for product distributions on $X$ as otherwise it cannot be related to the memory used by algorithm $\sfO$. For multi-pass streaming algorithms, however, even after the first pass, the input distribution for the second pass might have spurious correlations (for example, the first pass can remember the parity of $X_1$ and $X_2$). Therefore, we need significantly new techniques to generalize the proof of \cite{braverman2020coin} to multi-pass streaming algorithms. 

Many of the multi-pass streaming lower bounds\footnote{For example, see the multi-pass streaming lower bounds for graph problems in \cite{mps1}.} use a round elimination or a pass elimination technique, and embed a hard instance for $(k-1)$-pass streaming algorithms within the input distribution obtained by conditioning on the output of the first pass of a $k$-pass streaming algorithm. The hope of such a technique is that the output of the $k$-pass algorithm can be used to answer the $(k-1)$-pass problem. Even for $2$ passes, one major challenge in using such a technique for the problem of computing the majority is as follows. The $1$-pass bound of \cite{braverman2020coin} lower bounds an information cost measure that is averaged over the time-steps, thus allowing the output of the first pass to store even $\sqrt{n}$ bits of information about the input. It is not clear how to embed the hard instance for a $1$-pass algorithm (which is the uniform distribution) into such a conditional distribution so that the majority of the instances are correlated; one needs to use the fact that the first pass does not compute the majority at any time-step but can store $n^{\Omega(1)}$ bits of information about the input.

{\bf Our technique.} Instead, to circumvent these challenges, we develop a new simulation technique to prove our multi-pass streaming lower bounds for the coin problem. We give a technical overview of the proof for $k=2$ passes and the beauty of our technique is that it readily generalizes to larger values of $k$ in one shot. Given a stream $X$ of $n$ $i.i.d.$ uniform bits, let $\sfM$ be a two-pass algorithm that goes over $X$ twice in order and computes the majority -- that is, the expected variance of the sum of input bits conditioned on the output of the second pass is a constant factor less than that of the maximum. Informally, for ease of discussion, we refer to the variance reduction as $\sfM$ approximating $\sum_j X_j$ up to an additive error\footnote{This is actually crucial as our proof technique does not work with large multiplicative errors, which is a bottleneck for generalizing  single pass memory lower bounds for the coin problem with larger biases~\cite{BGZ21} to multiple passes. } of $\epsl\sqrt{n}$. Using $\sfM$, we construct a one-pass algorithm $\sfO$, which given a stream of $n$ $i.i.d.$ uniform bits $Y$, also approximates $\sum_j Y_j$ up to an additive error of $\sim\epsl\sqrt{n}$. Let $\sfM_i$ represent the random variable for the output (or memory state) at the end of the $i$-th pass of $\sfM$ ($i\in\{1,2\}$). Algorithm $\sfO$ simulates both passes of $\sfM$ in parallel, for which it needs the output of the first pass or the starting state of the second pass before reading the input $Y$.  $\sfO$ samples the end state of the first pass $m_1\sim \sfM_1$ and executes the two passes assuming the first pass ends in $m_1$. For this, $\sfO$ modifies the input $Y$ to a valid input $X'$ for $\sfM$ such that the output of the first pass of $\sfM$ on $X'$ would be $m_1$. We show that such a modification is possible in a streaming fashion; at the $j$-th time-step, $\sfO$ modifies input bit $Y_j$ to $X'_j$ knowing $m_1$ and the $(j-1)$-th memory states of the two passes. The goal of such a simulation is to be able to use the one-pass lower bound of~\cite{braverman2020coin} for $\sfO$ and thus, we need two things: 1) to obtain an approximation for $\sum_j Y_j$ using the approximation for $\sum_j X'_j$ output by $\sfM_2$ -- the output after the second pass of algorithm $\sfM$, and 2) to relate $IC(\sfO)$ (defined in Item \ref{it:po1}) to a new information measure for $k$-pass algorithms. 

{\bf Approximating $\sum_j Y_j$ using the output of the 2-pass algorithm $\sfM$.} As we noted above, it is challenging to show that $|\sum_j Y_j-\sum_j X'_j| < \epsl \sqrt{n} $, where the marginal distribution on $X'$ is the conditional distribution on reaching $m_1$ at the end of the first pass. Especially, if we just assume a memory upper bound on $m_1$, $|\sum_j Y_j-\sum_j X'_j|$ can be much greater than $ \sqrt{n}$. To bypass this challenge, we maintain a sketch of the modification, that is, of $\sum_j Y_j-\sum_j X'_j$. Storing the difference exactly requires $\log n$ memory, which we cannot afford if we want to use the $\Omega(\log n)$ memory lower bound for single-pass algorithms. However, we are able to  approximate $\sum_j Y_j-\sum_j X'_j$ up to an additive error of $\epsl\sqrt{n}$ using low memory - we actually only give an approximation algorithm whose memory states have low entropy, which is \emph{sufficient} for lifting information cost lower bounds. The key observation that makes such an approximation possible is: $\sfO$ does not need to drastically modify $Y$ (which has the uniform distribution) to obtain $X'$, as the expected KL divergence of the conditional distribution of $X'$ on reaching $m_1$ from the uniform distribution is bounded by the entropy of the random variable $\sfM_1$. We assume $\Ent(\sfM_1)<n^{\delta}$, but this can be much greater than $\log n$; this is vital for our generalization to solving multiple copies of the coin problem simultaneously. 

{\bf New information measure for $k$-pass algorithms.} At the $j$-th time-step, algorithm $\sfO$ stores $\sfM_1$ and the $j$-th memory states for the two passes of $\sfM$ (in addition to the approximation for modifying input $Y$ to $X'$). Looking at the single-pass information measure for this simulation algorithm $\sfO$, we get a natural expression for the an information measure for $k$-pass algorithms $\sfM$ as follows:
\[\miccoin(\sfM)=\sum_{j=1}^n\sum_{\ell=1}^j \I\left(\sfM_{(\le k,j)};X_{\ell}|\sfM_{< k},\sfM_{(\le k,\ell-1)}\right).\]
Here, $X$ is drawn from the uniform distribution over $\zo^n$, and $\sfM_{(i,j)}$ represents the $j$-th memory state for the $i$-th pass. This is related to our $\mic$ notion as follows, and can be of independent interest.
 \[
 \miccoin(\sfM) \leq \sum_{r=1}^{k}\sum_{i=1}^n\sum_{\ell=1}^i \mutualent(\sfM_{(r,i)};X_{\ell}\mid \sfM_{(\leq r,\ell-1)}, \sfM_{(\leq r-1,i)}) \leq \mic(\sfM). 
 \]
 By upper bounding $\mic$ as in Lemma \ref{lem:Multi-passupperbound}, we obtain an upper bound on $\miccoin$. 

{\bf Solving multiple instances of the coin problem. }\label{subsec:tcoinoverview}
We generalize our multi-pass streaming lower bounds to solving multiple instances of the coin problem simultaneously. Informally, given $t$ interleaved input streams generated by $n$ $i.i.d.$ uniform bits each, the goal of a multi-pass streaming algorithm is to output the majority of an arbitrary stream at the end of $k$ passes.  We show that any $k$-pass streaming algorithm that solves the $\tcoins$ requires $\Omega(\frac{t\log n}{k} )$ bits of memory (for $t<n^\delta$). As for the single coin case, we reduce the multiple coin case to the analogous result for one-pass streaming algorithms proven by \cite{braverman2020coin}. We simulate the multi-pass algorithm for the $\tcoins$ using a one-pass algorithm that maintains $t$ approximations for modifying each input stream to a valid stream for the $k$-pass algorithm. For the generalization, we utilize the fact that the single coin simulation works even when the output of the first pass has poly($n$) entropy; for the $\tcoins$ problem, we work with memories as large as $t\log n$.

\subsubsection{Multi-pass Streaming Lower Bound for the Needle Problem}\label{sec:overviewofneedle}
Since we have shown that $\mic(\sfM,\boldsymbol{D}_0)$ is upper bounded by $2ksn$, it suffices to give an $\Omega(1/p^2)$ lower bound for $\mic(\sfM,\boldsymbol{D}_0)$ as we formally present in Lemma \ref{thm:lowerbound}. In the following, we give the intuition behind Lemma \ref{thm:lowerbound}. Its formal proof can be found in Section \ref{sec:mainlemma}. We also provide a more detailed proof sketch for this lemma in Section \ref{sec:intromain}. 

\textit{In the needle problem, we use the notion $\mic(\calA, \boldsymbol{D}_0)$ as we defined before, where $\boldsymbol{D}_0$ stands for the uniform distribution.} For simplicity, we write $\mic(\calA)$ in the needle problem, and it could be easily distinguished from the notion $\miccoin(\calA)$ used in the coin problem. 

\begin{restatable}[]{lemma}{multilower}\label{thm:lowerbound}
In the needle problem, if a $k$-pass streaming algorithm $\calA$ distinguishes between $ \boldsymbol{D_0}$ and $ \boldsymbol{D_1}$ with high probability, then we have
$\mic(\calA)=\Omega(1/p^2).$
\end{restatable}
Let us first consider the special case when $p=1/2$. A useful observation is that the needle problem with $p=1/2$ is very similar to the \textsf{MostlyDISJ} communication problem \cite{KPW21}. Viewing the needle problem with $p=1/2$ as a multiparty communication problem (we name it \textsf{MostlyEq} in Definition \ref{def:mostlyeq}), we can show that its information complexity is $\Omega(1)$ (Lemma \ref{lem:ICresultforCCproblem}). As a direct consequence, the mutual information $\I(\sfM;X)$ between $\sfM=(\sfM_{(i,j)})_{i\in[k],j\in[n]}$ and input $X=(X_1,\cdots,X_n)$ is also $\Omega(1)$. Then, we can prove that $\mic(\sfM) \geq \I(\sfM, X) \geq \Omega(1)$ with some information theory calculations.

Now, let us consider general $p\leq 1/2$. We write $\boldsymbol{D_1}(p)$ to make the parameter $p$ in the needle distribution $\boldsymbol{D}_1$ explicit. Note that for a general $p\leq 1/2$, its needle distribution $\boldsymbol{D_1}(p)$ is a mixture of the needle distribution $\boldsymbol{D_1}(1/2)$ with $p=1/2$ and the uniform distribution $\boldsymbol{D_0}$. That is to say, a data stream $X_1,\cdots,X_n$ from $\boldsymbol{D_1}(p)$ can also be generated as follows: 
\begin{enumerate}
\item let $S$ be a random subset of $[n]$ where each $i\in [n]$ is contained in $S$ independently with probability $2p$;
\item for each $i\in S$, pick $X_i$ following the needle distribution $\boldsymbol{D_1}(1/2)$;
\item for each $i\notin S$, pick $X_i$ following the uniform distribution $\boldsymbol{D_0}$.
\end{enumerate}
Thus, solving the needle problem with general $p$ is equivalent to solving the needle problem with $p=1/2$ hiding in a secret location $S$. Because the streaming algorithm $\sfM$ does not know $S$, if $\sfM$ solves the needle problem for general $p$, then $\sfM$ must solve the needle problem with $p=1/2$ located at $S$ simultaneously for most choices of $S$. 

If $\sfM$ solves the needle problem with $p=1/2$ located at $S=\{\pos_1,\pos_2,\cdots,\pos_{m}\}$, then as shown above, it has 
\begin{equation}\label{eq:p=1/2}
\mic^{S}:=\sum_{i=1}^{k-1}\sum_{j=1}^m\sum_{\ell=1}^i \I(\sfM_{(i,p_{j+1}-1)};X_{p_\ell}\mid \sfM_{(\leq i,p_{\ell}-1)},\sfM_{(\leq i-1,p_{j+1}-1)})=\Omega(1).
\end{equation}
where we define $\pos_{m+1}$ as $\pos_1$. Finally, by taking an expectation over $S$, the L.H.S. of  \eqref{eq:p=1/2} is about $\mic(\sfM)/O(p^2)$ since each term $\mutualent(\sfM_{(i,j)};X_{\ell}\mid \sfM_{(\leq i,\ell-1)}, \sfM_{(\leq i-1,j)})$ or $\mutualent(\sfM_{(i,j)};X_{\ell}\mid \sfM_{(\leq i-1,\ell-1)}, \sfM_{(\leq i-1,j)})$ of $\mic(\sfM)$ appears in the L.H.S. of \eqref{eq:p=1/2} for an $O(p^2)$ fraction of $S$.

\subsubsection{Algorithms for the Needle Problem}\label{sec:algooverview}
%{\color{blue} TODO: describe the blocks, relation to BGW, hash functions assigned, counting, randomly splitting for general p}
%
%{\color{red} Shuo: I write follows, feel free to check it or revise it. }

%\noindent
%In this subsection, 
We introduce the idea behind the following two algorithms: 
%\begin{enumerate}
(1) $\calA_1$ that solves the needle problem with $p\geq 1/\sqrt{n}$ in $O(\log \log n (\log\log\log n))$ space in the communication model; (2) $\calA_2$ that solves the needle problem with $p\leq 1/\sqrt{n\log^3 n}$ in $O\big(1/(p^2t)\big)$ space in the streaming model.
%\end{enumerate}
%The basic idea of the two algorithms is similar. 
As with previous work \cite{braverman2014optimal,braverman2020coin,wz21} for finding $\ell_2$-heavy hitters we partition the stream items into contiguous groups (in previous work, these groups contain $\Theta(\sqrt{n})$ stream items). These works sample $O(1)$ items in each group and track them over a small number of future groups - this is sometimes called {\it pick and drop sampling}. A major difference between our algorithms and these is that we cannot afford to store the identity of an item and track it, as that would require messages of length at least $\log n$ bits. %This is fine and necessary for earlier work as they require returning the identity of the heavy hitter, whereas we are just trying to detect its existence. 
A natural idea is to instead track a small hash of an item but there will be a huge number of collisions throughout the stream if we use fewer than $\log n$ bits.

We start by choosing each group to be of size $\Theta(1/p)$, so each group has one occurrence of the needle with constant probability under distribution $\boldsymbol{D}_1$. We discuss the algorithm $\calA_1$ for  $p=1/\sqrt{n}$ first, so the group size is $\sqrt{n}$.  This can be generalized to any $p\geq 1/\sqrt{n}$ (see Section \ref{sec:upper_bound}). 
%$\calA_1$ divides the data stream into $\sqrt{n}$ data groups as we mentioned above, this could be achieved since under the communication game or the streaming with a clock model, $\calA_1$ could know the index of the coming data.
For universe $[t]$, we randomly sample a hash function projecting $[t]$ to $[C_2]$, where $C_2$ is a constant. For each group, we randomly sample a subset of $[t]$ with size $C_1 t/\sqrt{n}$. Then, $\calA_1$ runs in $\sqrt{n}$ rounds, and processes one data group in each round. In round $i$, we set $C_2$ new counters: for each $j\in [C_2]$, we set a counter that tracks $j$; and when processing the following groups, we check if each element $x$ exists in $i$'s random subset; if it is in, we update its corresponding counter (the counter of group $i$ with the hash value of $x$). After processing each group, we check each counter to see if it is at least a constant times the number of groups processed after it was initialized. If not, we drop it. The intuition is that (1) a counter that does not track the needle survives $r$ rounds with probability less than $e^{-cr}$, and (2) a counter tracking the needle has constant probability to survive. We refer the reader to a slightly longer intuition in Section \ref{sec:algo1}. We may accumulate many counters and the space may hit our $O((\log \log n)(\log \log \log n))$ bound - if so, we throw away all counters and start over. We show that at the critical time when we start processing the needle we will not throw it away.

For the second algorithm $\calA_2$, the idea is to divide the domain $[t]$ into $1/(p^2n)$ \textit{blocks} and simultaneously run $1/(p^2n)$ algorithms similar to $\calA_1$ on each block, checking if the needle exists. We show we only need $O(1/(p^2n))$ space in total when $p\leq 1/\sqrt{n \log^3 n}$ holds. This algorithm is in the standard streaming model since we can afford an additive $\log n$ bit counter. See Section \ref{sec:algo2}. 

\subsection{Future Directions}

 Using the single-pass notion of information complexity in \cite{braverman2020coin}, that we extend to multiple passes, Brown, Bun and Smith \cite{Colt22} showed single-pass streaming lower bounds for several learning problems. A natural question is if our multi-pass techniques can be useful for  learning problems. 
%\paragraph{Improve existing streaming lower bounds.} Unlike previous needle lower bounds \cite{Andoni08,ToC16,DBLP:conf/esa/CrouchMVW16, lovett2023streaming} which were built on black-box reductions to the communication lower bounds of the \textit{unique disjointness problems}, our IC-based proof is more adopted to prove streaming lower bounds directly. This is the main reason that we can achieve tight lower bounds. 
Another potential application is that of Dinur \cite{euro2019}, who shows streaming lower bounds for distinguishing random functions from random permutations.  Also, Kamath et al. \cite{KPW21} study the heavy hitters problem for $O(1)$ pass algorithms. These results have a logarithmic factor gap and use more classical notions of information complexity. Can our techniques apply here?

%Furthermore, \cite{KPW21} proved streaming lower bounds only for $O(1)$-pass streaming algorithms. It would also be interesting to check if our multi-pass IC notion and our new proof techniques could prove tight streaming lower bounds for multi-pass algorithms.

\subsection{Organization}
In Section \ref{sec:prelim}, we give preliminaries. In Section \ref{sec:upper}, we give upper bounds on our multi-pass information complexity notions. We then lower bound the information complexity of the coin problem in Section \ref{sec:single}, while we give our lower bound  for the needle problem in Section \ref{sec:mainlemma}. In Section \ref{sec:upper_bound}, we give two new algorithms for the needle problems to improve the upper bounds. Section \ref{sec:application} gives our other streaming applications.
\subsection{Acknowledgements}
The authors thank three
anonymous STOC reviewers for their helpful suggestions on this paper. 

Mark Braverman is supported in part by the NSF Alan T. Waterman Award, Grant No. 1933331, a Packard Fellowship in Science and Engineering, and the Simons Collaboration on Algorithms and Geometry. 

Qian Li is supported by Hetao Shenzhen-Hong Kong Science and Technology  Innovation Cooperation Zone Project (No.HZQSWS-KCCYB-2024016). 

David P. Woodruff's research is supported in part by a Simons Investigator Award, and part of this work was done while visiting the Simons Institute for the Theory of Computing. 

Jiapeng Zhang's research is supported by NSF CAREER award 2141536.
%and thus obtain the multi-pass lower bound on the coin problem. In Section \ref{sec:mainlemma}, we provide 

\newpage

\tableofcontents

\newpage 

\section{Preliminaries}\label{sec:prelim}

Let $[k]$ denote the set of natural numbers $\{1,2,\ldots,k\}$. Unless mentioned otherwise, we assume $\log$ to be base $2$. We will denote bits by the sets $\{0, 1\}$ and $\zo$ interchangeably. We use capital letters such $X, Y, Z,$ etc., to denote random variables and $x, y, z,$ etc., to denote the values these random variables take. We use the notation $(Z_1\indep Z_2 | Z_3)$ to denote conditional independence -- finite random variables $Z_1$ and $Z_2$ are independent conditioned on $Z_3$, if and only if  
\[\Pr[Z_1=z_1 | Z_2=z_2,Z_3=z_3]= \Pr[Z_1=z_1 | Z_3=z_3], \text{ for all values }z_1,z_2, z_3.\] Given a probability distribution $\calD:\X\rightarrow [0,1]$, we use the notation $x\sim \calD$ when value $x$ is sampled according to distribution $\calD$.   Similarly, we use the notation $z\sim Z$ to denote the process that $Z$ takes value $z$ with probability $\Pr[Z=z]$. We use $\Ber(q)$ to denote the Bernoulli distribution which takes value $1$ with probability $q$ and $0$ with probability $1-q$. We use notations $\bbE[Z]$ and $\Var(Z)$ to denote the expectation and variance of random variable $Z$ respectively. $\bbE[Z|Y=y]$ and $\Var(Z|Y=y)$ denote the expectation and variance of random variable $Z$ conditioned on the event $Y=y$.

\paragraph{Basics of information theory.} Given a random variable $Z$, $\Ent(Z)$ denotes the Shannon entropy of $Z$, that is, $\Ent(Z)=\sum_{z}\Pr(Z=z)\log  (1/\Pr(Z=z))$. We also use 
$\Ent(\calD)$ to denote the entropy of the probability distribution $\calD$. $\I(X;Y|Z)$ represents the mutual information between $X$ and $Y$ 
conditioned on the random variable $Z$. $\I(X;Y|Z)=\Ent(X|Z)-\Ent(X|Y,Z)$, where $\Ent(X|Y)=\bbE_{y\sim Y}\Ent(X|Y=y)\le \Ent(X)$. 
Next, we describe some of the properties of mutual information used in the paper. 
\begin{enumerate}
\item \label{itemmi1} (Chain Rule) $\I(XY;Z)=\I(X;Z)+\I(Y;Z|X)$.
\item \label{itemmi4} For discrete random variables, $X,Y,Z$, $\I(X;Y|Z)=0\iff X\indep Y|Z$.
\item \label{itemmi2} If $\I(Z_1;Z_2|X,Y) = 0$, then $\I(X;Z_2|Y) \ge \I(X;Z_2|Y,Z_1)$.
\item \label{itemmi3} If $\I(Z_1;Z_2|Y) = 0$, then $\I(X;Z_2|Y)\le \I(X;Z_2|Y,Z_1)$.
\end{enumerate}
Property \ref{itemmi1} follows from the chain rule for entropy ($\Ent$). For property \ref{itemmi4}, it is easy to see that if $X\indep Y|Z$, then $\I(X;Y|Z)=0$; the other direction uses strict concavity of the $\log$ function. Properties \ref{itemmi2} and \ref{itemmi3} follow from the observation that 
\[\I(X;Z_2|Y)+\I(Z_1;Z_2|X,Y)=\I(X,Z_1;Z_2|Y)=\I(Z_1;Z_2|Y)+\I(X;Z_2|Y,Z_1).\]
As mutual information is non-negative, if $\I(Z_1;Z_2|X,Y)=0$, then $\I(X;Z_2|Y)\ge \I(X;Z_2|Y,Z_1)$ (because $\I(Z_1;Z_2|Y)\ge 0$) and if $\I(Z_1;Z_2|Y)=0$, then $\I(X;Z_2|Y)\le \I(X;Z_2|Y,Z_1)$.
%{\color{red}Another important property comes from \cite{braverman2020coin}: 
%\begin{lemma}[\cite{braverman2020coin}]
%\label{lemma:BGW}
%For random variables $A,B,C,D$, we have the following two inequality: 
%\begin{itemize}
%    \item If $\mutualent(D;B|C) = 0$, then it holds that   
%$
%\mutualent({A};{B}|{C})\leq \mutualent({A};{B}|{C},{D}).
%$
%\item If $\mutualent({D};{B}|{A},{C}) = 0$, then it holds that   
%$
%\mutualent({A};{B}|{C})\geq \mutualent({A};{B}|{C},{D}).
%$
%\end{itemize}
%\end{lemma}
%}

\paragraph{Multi-pass streaming algorithms.} Given a stream of $n$ input elements, $x_1,\ldots,x_n\in\X^n$, we say $\sfM$ is a $k$-pass algorithm (for $k\ge 1$) when it goes over the entire stream $k$ times in order. We use $\sfM_{(i,j)}$, for $i\in [k], j\in [n]$, to denote the random variable representing the memory state of $\sfM$ in the $i$-th pass after reading $j$ input elements. Let $\sfM_0=\sfM_{(1,0)}$ denote the starting memory state and for ease of notation, let $\sfM_{(i+1,0)}=\sfM_{(i,n)}$ for all $i\in[k]$. We say algorithm $\sfM$ uses private randomness if at every time-step, $\sfM$ uses independent randomness to transition to the next memory state. Let $\Ra_{(i,j)}^\sfM$ ($i\in [k], j\in [n]$) denote the random variable for private randomness used by algorithm $\sfM$ at the $j$-th step of the $i$-th pass.
%\footnote{Interestingly, our lower bounds on information cost and memory hold even for $k$-pass algorithms that share common private randomness at the $j$-th time-step across passes, which is independent across time-steps. This model can be potentially useful for direct sum theorems for multi-pass streaming. However, for the paper, we do not need the stronger model and we assume that $\Ra_{(i,j)}^\sfM$s are independent for all $i\in [k], j\in [n]$.}. 
Let $\Ra^\sfM=\{\Ra_{(i,j)}^\sfM\}_{i\in[k],j\in[n]}$ denote the set of all randomness used by $\sfM$. 

We use lowercase letters  $m_{(i,j)}$ and $r_{(i,j)}$ to represent instantiations of these random variables, $\sfM_{(i,j)}$ and $\Ra_{(i,j)}$ respectively. We use $f_{(i,j)}^\sfM$ to represent the transition function at the $j$-th time-step of the $i$-th pass, that is, \[m_{(i,j)}=f_{(i,j)}^\sfM\left(x_j, m_{(i,j-1)},r_{(i,j)}\right).\] Let $\sfM(x,r)$ denote the output when $\sfM$ executes $k$-passes using input $x$ and randomness $r$. 
We use notation $[a,b]$ in the subscript to represent random variables indexed from $a$ to $b$, for example, $\sfM_{([1,i],j)}$ represents $j$-th memory states for the first $i$ passes, that is, $\sfM_{(1,j)},\ldots, \sfM_{(i,j)}$. We use notations $<b$, $\le b$ in the subscript to represent all the corresponding random variables with index less than $b$ or at most $b$ respectively. For example, $\sfM_{(i,\le j)}$ represents random variables $\sfM_{(i,[0,j])}$.

\paragraph{Information cost for multi-pass streaming algorithms. } As we mentioned before, for a given distribution $\mu$ over $\X^n$, we define the information cost of a $k$-pass protocol $\sfM$ on distribution $\mu$ by:
%\[IC_\mu(\sfM)=\sum_{i=1}^k\sum_{j=1}^n\sum_{\ell=1}^j \I(\sfM_{(\le i,j)};X_{\ell}|\sfM_{\le i},\sfM_{(\le i,\ell-1)}).\]
\begin{align*}
\mic(\sfM,\mu)=&\sum_{i=1}^{k}\sum_{j=1}^n\sum_{\ell=1}^j \left(\sfM_{(i,j)};X_{\ell}\mid \sfM_{(\leq i,\ell-1)}, \sfM_{(\leq i-1,j)}\right)\\
&+\sum_{i=1}^{k}\sum_{j=1}^n\sum_{\ell=j+1}^n \left(\sfM_{(i,j)};X_{\ell}\mid \sfM_{(\leq i-1,\ell-1)}, \sfM_{(\leq i-1,j)}\right).
\end{align*}
Here, $X$ is drawn from $\mu$, and random variables for memory states depend both on the randomness of the input as well as private randomness used by the algorithm. When $\mu$ is clear from context, we will drop it from the notation. For the multi-pass lower bound on the coin problem, we use a different information cost notion for multi-pass streaming algorithms, which might be of independent interest. Below, we condition on the end memory states of all the passes, which allows us to consider information learnt by $j$th memory states of all passes simultaneously; this is particularly useful for analysing one-pass simulation of $\sfM$.  For ease of notation, we use $\sfM_i$ to denote the end memory state of the $i$th pass, that is, $\sfM_i=\sfM_{(i,n)}$ for all $i\in[k]$. Subsequently,  $\sfM_{(< i)}$ represents $\sfM_{[0,i-1]}$ -- end memory states of first $i-1$ passes including the starting memory state. 

\[\miccoin(\sfM,\mu)=\sum_{j=1}^n\sum_{\ell=1}^j \I\left(\sfM_{(\le k,j)};X_{\ell}|\sfM_{< k},\sfM_{(\le k,\ell-1)}\right).\]
%Again, $X$ is drawn from $\mu$, and random variables for memory states depend both on the randomness of the input as well as private randomness used by the algorithm. 
In Section \ref{sec:upper}, we will prove that for all  product distributions $\mu$, $\miccoin(\sfM,\mu)\le \mic(\sfM,\mu)$.

\paragraph{Concentration inequalities.} As commonly used, the mathematical constant $e$ denotes Euler's number. In this paper, we will use two concentration inequalities -- the Chernoff bound~\cite{chernoff} and Berstein's inequality~\cite{bernstein1,bernstein2}. Let $Z_1,Z_2,\ldots,Z_n$ be independent random variables in $\{0,1\}$, and $\nu=\bbE\left[\sum_{j=1}^nZ_j\right]$. Then, by the Chernoff bound,
\begin{equation}
\label{eq:ch1}
\Pr\left[\left|\sum_{j=1}^nZ_j-\nu\right|\ge \delta \nu\right]\le 2\exp{\left(-\frac{\delta^2\nu}{3}\right)}, \;\forall\; 0<\delta<1, \text{ and}
\end{equation}
\begin{equation}
\label{eq:ch2}
\Pr\left[\sum_{j=1}^nZ_j\ge (1+\delta)\nu\right]\le \exp{\left(-\frac{\delta^2\nu}{2+\delta}\right)}, \;\forall\; \delta>0.
\end{equation}
Note that, for all $\delta\ge 1$, $\frac{\delta}{2+\delta}\ge \frac{1}{3}$. Therefore, Equation \eqref{eq:ch2} can be rewritten as
\begin{equation}
\label{eq:ch3}
\Pr\left[\sum_{j=1}^nZ_j\ge \nu+t\right]\le \exp{\left(-\frac{t}{3}\right)}, \;\forall\; t\ge \nu.
\end{equation}

Under Bernstein's inequality, let $Z_1,Z_2,\ldots,Z_n$ be bounded independent random variables such that $\forall i, |Z_i|<\tau$, $\nu=\bbE\left[\sum_{j=1}^nZ_j\right]$ and $\sigma=\frac{1}{n}\sum_{j=1}^n\Var(Z_i)$. Then for all $t>0$
\begin{equation}
\label{eq:ber}
\Pr\left[\left|\sum_{j=1}^nZ_j-\nu\right|\ge t\right]\le 2\exp{\left(-\frac{\frac{1}{2}t^2}{n\sigma+\frac{1}{3}\tau t}\right)}. 
\end{equation}

We will also use the Cauchy-Schwarz inequality, which can be stated as follows: for real numbers $u_1, u_2,\ldots, u_n$ and $v_1,v_2,\ldots, v_n$, 
\[\left(\sum_{i=1}^nu_i\cdot  v_i\right)^2\le \left(\sum_{i=1}^n u_i^2\right) \left(\sum_{i=1}^n v_i^2\right).\]

\section{Properties of New Multi-Pass IC Notion}\label{sec:upper}

In this section, we will prove some important properties of our IC notion. Let $\sfM$ be a $k$-pass streaming algorithm that uses $s$ bits of memory.
Particularly, we prove that, when $X_1,\dots,X_n$ are drawn from a product distribution $\mu$, the following holds:
\begin{itemize}
    \item $\mic(\sfM,\mu)\leq 2ksn$;
    \item $\mic(\sfM,\mu)\geq \miccoin(\sfM,\mu)$.
\end{itemize}

\subsection{Useful Facts about Mutual Information}
Here, we will state several important properties of the conditional mutual information between the memory states and the input data stream without proofs. We defer the detailed proofs to Appendix \ref{apendix:sec3}. 
First, in Claim \ref{cl:indep1}, we show that if we condition on the $j$th memory states and the end memory states, then there is no correlation between the inputs before and after the $j$th time-step. 
%Claim \ref{cl:indep2} proves that input at the $j$th time-step $X_j$ is independent of $(j-1)$th memory state of any pass if we condition on the $(j-1)$th memory states and end memory states of the previous passes. Note that, without this conditioning, $\sfM_{(2,j-1)}$ can contain information about $X_j$ if the first pass remembers it. 

\begin{restatable}[]{claim}{indepi}\label{cl:indep1}
Given a stream of $n$ input elements $X_1,\ldots, X_n$ drawn from a product distribution, that is, 
\[\Pr[X_1=x_1,X_2=x_2,\ldots,X_n=x_n]=\Pr[X_1=x_1]\cdot\Pr[X_2=x_2]\cdot\ldots\cdot\Pr[X_n=x_n],\]
let $\sfM$ be a $k$-pass protocol using private randomness $r\sim \Ra^\sfM$. Then $\forall j\in[n]$, 
\begin{equation}
\label{eqcl:indep1}
\I\left(X_{[1,j]}, \Ra_{(\le k, [1,j])}; X_{[j+1,n]}, \Ra_{(\le k, [j+1,n])}\mid \sfM_{\le i}, \sfM_{(\leq i+1, j)}\right)=0, \; \forall i\in\{0,1,\ldots,k-1\} \text{ and }
\end{equation}
\begin{equation}
\label{eqcl:indep2}
\I\left(X_{[1,j]}, \Ra_{(\le k, [1,j])}; X_{[j+1,n]}, \Ra_{(\le k, [j+1,n])}\mid \sfM_{\le i}, \sfM_{(\le i, j)}\right)=0, \;\forall i\in\{0,1,\ldots,k\}.
\end{equation}

\end{restatable} 
\noindent
Claim \ref{cl:indep1} immediately gives the following two corollaries. 
\begin{restatable}[]{corollary}{indepii}\label{cor:indep1}
Given a stream of $n$ input elements $X_1,\ldots, X_n$ drawn from a product distribution. Let $\sfM$ be a $k$-pass protocol using private randomness $r\sim \Ra^\sfM$. Then $\forall j\in[n]$,  $i\in[k]$,
\[\I\left(X_{[1,j]}, \sfM_{(\le i, [0,j-1])}; X_{[j+1,n]}, \sfM_{(\le i, [j+1,n])}\mid \sfM_{< i}, \sfM_{(\le i, j)}\right)=0.\]
\end{restatable}

%\begin{claim}\label{cl:indep2}
%Given a stream of $n$ input elements $X_1,\ldots, X_n$ drawn from a product distribution, that is, 
%\[\Pr[X_1=x_1,X_2=x_2,\ldots,X_n=x_n]=\Pr[X_1=x_1]\cdot\Pr[X_2=x_2]\cdot\ldots\cdot\Pr[X_n=x_n],\]
%let $\sfM$ be a $k$-pass protocol using private randomness $r\sim \Ra^\sfM$. Then $\forall i\in[k], j\in[n]$, 
%\[\left(X_j \indep  \sfM_{(i,j-1)}\mid \sfM_{<i}, \sfM_{(< i, j-1)}\right).\]
%\end{claim}
%\begin{proof}
%\end{proof}
Corollary \ref{cor:indep2} shows that the input at the $j$th time-step $X_j$ is independent of the $(j-1)$th memory state of any pass if we condition on the $(j-1)$th memory states and end memory states of the previous passes. Note that, without this conditioning, $\sfM_{(2,j-1)}$ can contain information about $X_j$ if the first pass remembers it.
\begin{restatable}[]{corollary}{indepiii}\label{cor:indep2}
Given a stream of $n$ input elements $X_1,\ldots, X_n$ drawn from a product distribution. Let $\sfM$ be a $k$-pass protocol using private randomness $r\sim \Ra^\sfM$. Then $\forall i\in\{0,\ldots, k-1\}, j\in[n]$, 
\[ \I\left(X_j; \sfM_{(i+1,j-1)}\mid \sfM_{\le i}, \sfM_{(\le i, j-1)}\right)=0.\]

\end{restatable}
\noindent
We then generalize the statements above to a more general setting to fit the definition of $\mic$: 
\begin{claim}\label{claim:micindependence}
    Given a stream of $n$ input elements $X_1,\ldots, X_n$ drawn from a product distribution. Let $\sfM$ be a $k$-pass protocol using private randomness $r\sim \Ra^\sfM$. 
    Then, $\forall i\in\{1,\ldots, k\}, j,\ell\in[n]$, if $j\geq \ell$, it holds that
    \begin{enumerate}
        \item $\I(X_{[\ell,j]},\calR_{(\leq k,[\ell,j])};X_{[1,\ell-1]},\calR_{(\leq k,[1,\ell-1])},X_{[j+1,n]},\calR_{(\leq k,[j+1,n])}\mid \sfM_{(\leq i,\ell-1)},\sfM_{(\leq i-1,j)})= 0$;
        \item $\I(X_{[\ell,j]},\calR_{(\leq k,[\ell,j])};X_{[1,\ell-1]},\calR_{(\leq k,[1,\ell-1])},X_{[j+1,n]},\calR_{(\leq k,[j+1,n])}\mid \sfM_{(\leq i,\ell-1)},\sfM_{(\leq i,j)})= 0$.
    \end{enumerate}
\end{claim}
\noindent
This claim could be proved by similar arguments to the proof of Claim \ref{cl:indep1}. To avoid repetitiveness, we omit its proof. 
\subsection{Upper Bound of Multi-Pass IC}
%Recall the statement of Lemma \ref{lem:Multi-passupperbound}:
\multiupper*
\begin{proof}  
We prove this upper bound by the following two statements: 
\begin{enumerate}
    \item For each pass $i$, it holds $$\sum_{j=1}^n\sum_{{\ell}=1}^j \I(\sfM_{(i,j)};X_{\ell}\mid\sfM_{(\leq i,{\ell}-1)},\sfM_{(\leq i-1,j)})\leq s\cdot n.$$
    \item For each pass $i$, it holds $$\sum_{j=1}^n\sum_{{\ell}=j+1}^n \I(\sfM_{(i,j)};X_{\ell}\mid\sfM_{(\leq i-1,{\ell}-1)},\sfM_{(\leq i-1,j)})\leq s\cdot n.$$
\end{enumerate}
We begin with the first claim:
\begin{align*}
&\sum_{j=1}^n\sum_{{\ell}=1}^j 
\I(\sfM_{(i,j)};X_{\ell}\mid\sfM_{(\leq i,{\ell}-1)},\sfM_{(\leq i-1,j)})\\ \leq & \sum_{j=1}^n\sum_{{\ell}=1}^j \I(\sfM_{(i,j)};X_{\ell},\sfM_{(\leq i,\leq {\ell}-2)},X_{< {\ell}}\mid\sfM_{(\leq i,{\ell}-1)},\sfM_{(\leq i-1,j)})  
\\
=& \sum_{j=1}^n\sum_{{\ell}=1}^j\I(\sfM_{(i,j)};\sfM_{(\leq i,\leq {\ell}-2)},X_{< {\ell}}\mid\sfM_{(\leq i,{\ell}-1)},\sfM_{(\leq i-1,j)})+\I(\sfM_{(i,j)};X_{\ell}\mid\sfM_{(\leq i,\leq {\ell}-1)},\sfM_{(\leq i-1,j)},X_{< {\ell}}) \tag{Chain rule}\\
= & \sum_{j=1}^n\sum_{{\ell}=1}^j  \I(\sfM_{(i,j)};X_{\ell}\mid\sfM_{(\leq i,\leq {\ell}-1)},\sfM_{(\leq i-1,j)},X_{< {\ell}})\tag{explained below}\\
\leq & \sum_{j=1}^n\sum_{{\ell}=1}^j  \I(\sfM_{(i,j)};X_{\ell}, \sfM_{(\leq i,{\ell})}\mid\sfM_{(\leq i,\leq {\ell}-1)},\sfM_{(\leq i-1,j)},X_{< {\ell}})\\
= & \sum_{j=1}^n \I(\sfM_{(i,j)};X_{\leq j}, \sfM_{(\leq i,\leq j)}\mid\sfM_{(\leq i-1,j)})\tag{Chain rule} \\
\leq & \;\;s\cdot n.
\end{align*}
 The second equality comes from the following fact:
\begin{align}\label{fact1}
\I(\sfM_{(i,j)};\sfM_{(\leq i,\leq {\ell}-2)},X_{< {\ell}}\mid \sfM_{(\leq i,{\ell}-1)},\sfM_{(\leq i-1,j)})=0,
\end{align}
which is implied by the stronger statement of Claim \ref{claim:micindependence}, that is, 
$$\I\big( X_{[1,\ell-1]}, X _{[j+1,n]},\Ra_{(\leq k,[1,\ell-1])}, \Ra_{(\leq k,[j+1,n])}; X_{[\ell,j]},\Ra_{(\leq k,[\ell,j])}\mid \sfM_{(\leq i,{\ell}-1)},\sfM_{(\leq i-1,j)} \big)=0.$$ 
This is because $\sfM_{(\leq i,\leq {\ell}-2)},X_{< {\ell}}$ are a deterministic function of $ X_{[1,\ell-1]}$, $X _{[j+1,n]}$,$\Ra_{(\leq k,[1,\ell-1])}$, $\Ra_{(\leq k,[j+1,n])}$,
given $\sfM_{(\leq i-1,j)}$,
whereas $\sfM_{(i,j)}$ could be fully determined by $( X_{[\ell,j]},\Ra_{(\leq k,[\ell,j])})$,
given $\calA_{(\leq i,{\ell}-1)}$.  This concludes the first claim.

Similarly, for the second claim, we have: 
\begin{align*}
&\sum_{j=1}^n\sum_{{\ell}=j+1}^n \I(\sfM_{(i,j)};X_{\ell}\mid\sfM_{(\leq i-1,{\ell}-1)},\sfM_{(\leq i-1,j)})\\ \leq & \sum_{j=1}^n\sum_{{\ell}=j+1}^n \I(\sfM_{(i,j)};X_{\ell},\sfM_{(\leq i-1,[j+1,{\ell}-2])},X_{[j+1,{\ell}-1]}\mid\sfM_{(\leq i-1,{\ell}-1)},\sfM_{(\leq i-1,j)}) \\ =& \sum_{j=1}^n\sum_{{\ell}=j+1}^n\I(\sfM_{(i,j)};\sfM_{(\leq i-1,[j+1,{\ell}-2])},X_{[j+1,{\ell}-1]}\mid\sfM_{(\leq i-1,{\ell}-1)},\sfM_{(\leq i-1,j)})\\ &\;\;\;\;\;\;\;\;\;\;\;+\I(\sfM_{(i,j)};X_{\ell}\mid\sfM_{(\leq i-1,{\ell}-1)},\sfM_{(\leq i-1,j)},\sfM_{(\leq i-1,[j+1,{\ell}-2])},X_{[j+1,{\ell}-1]})\tag{Chain rule} \\
\leq & \sum_{j=1}^n\sum_{{\ell}=j+1}^n \I(\sfM_{(i,j)};X_{\ell}, \sfM_{(\leq i-1,{\ell})}\mid\sfM_{(\leq i-1,[j+1,{\ell}-1])},\sfM_{(\leq i-1,j)},X_{[j+1,{\ell}-1]})\\
= & \sum_{j=1}^n \I(\sfM_{(i,j)};X_{[j+1,n]}, \sfM_{(\leq i-1,[j+1,n])}\mid\sfM_{(\leq i-1,j)}) \leq s\cdot n
\end{align*}
Now, combining the two claims, we have 
\[\mic(\calA,\mu)\leq 2k sn,\]
as desired. \qedhere

\subsection{Upper Bound $\miccoin$ by $\mic$}

\begin{lemma}\label{lem:upperboundofmic'}
 Let $(X_1,X_2,\cdots,X_n)$ be drawn from a product distribution $\mu$. Then, for any $k$-pass streaming algorithm $\calA$ on the $n$-bit input stream, we have: 
    \[\mic(\calA,\mu)\geq \sum_{i=1}^k\sum_{j=1}^n\sum_{{\ell}=1}^j \I(\sfM_{(i,j)};X_{\ell}\mid\sfM_{(\leq i,{\ell}-1)},\sfM_{(\leq i-1,j)})\geq \miccoin(\calA,\mu).\]
    As a corollary,  $\miccoin(\calA,\mu)\leq k s n $, where $s$ is the amount of memory used by $\sfM$. 
\end{lemma}
\begin{proof}
    We have the following inequality: 
    \begin{align*}
        \miccoin(\calA,\mu)=& \sum_{j=1}^n\sum_{\ell=1}^j \I(\calA_{(\leq k,j)};X_{\ell}\mid \calA_{(<k,n)},\calA_{(\leq k,\ell-1)}) \\
        =& \sum_{j=1}^n\sum_{\ell=1}^j \I(\calA_{(\leq k-1,j)};X_{\ell}\mid \calA_{(<k-1,n)},\calA_{(\leq k-1,\ell-1)},\calA_{(k-1,n)},\calA_{(k,{\ell}-1)})\\
        &+\;\;\;\;\;\;\;\;\;\I(\calA_{(k,j)};X_{\ell}\mid \calA_{(<k,n)},\calA_{(\leq k,{\ell}-1)},\calA_{(\leq k-1,j)})\tag{Chain rule} \\
        \leq& \sum_{j=1}^n\sum_{\ell=1}^j \I(\calA_{(\leq k-1,j)};X_{\ell}\mid \calA_{(<k-1,n)},\calA_{(\leq k-1,\ell-1)})+\I(\calA_{(k,j)};X_{\ell}\mid \calA_{(\leq k,{\ell}-1)},\calA_{(<k,j)}).
    \end{align*}

    \noindent    
   The inequality comes from two observations together with Property \ref{itemmi2}: 
    \begin{itemize}
        \item $\I(\calA_{(\le k-1,n)},\calA_{(k,{\ell}-1)};X_{\ell}\mid \calA_{(<k,j)},\calA_{(\leq k-1,\ell -1)}) = 0$, and thus, 
        \[\I(\calA_{(\leq k-1,j)};X_{\ell}\mid \calA_{(< k,n)},\calA_{(\leq k,\ell-1)})\leq \I(\calA_{(\leq k-1,j)};X_{\ell}\mid \calA_{(<k-1,n)},\calA_{(\leq k-1,\ell-1)});\]
        \item $\I(\calA_{(<k,n)};X_{\ell}\mid \calA_{(\leq k,j)},\calA_{(\leq k,\ell-1)}) = 0$, and thus, $$\I(\calA_{(k,j)},X_{\ell}\mid \calA_{(<k,n)},\calA_{\leq k,{\ell}-1},\calA_{(<k,j)}) \leq  \I(\calA_{(k,j)},X_{\ell}\mid \calA_{(\leq k,{\ell}-1)},\calA_{(<k,j)}).$$
        
    \end{itemize}
    The two observations could be proved by the similar arguments to the proof of Equation~\eqref{fact1}, and we omit the proofs here. Then, we do the decomposition process above recursively and conclude that: 
    \[
    \sum_{i=1}^k\sum_{j=1}^n\sum_{{\ell}=1}^j \I(\sfM_{(i,j)};X_{\ell}\mid\sfM_{(\leq i,{\ell}-1)},\sfM_{(\leq i-1,j)}) \geq \miccoin(\calA,\mu).     \qedhere
    \]
\end{proof}
\end{proof}

\section{Multi-Pass Lower Bound for the Coin Problem}\label{sec:single}
In this section, we assume that $X_1,\ldots, X_n$ are drawn from the uniform distribution over $\zo^n$; we will drop $\mu$ for the rest of the section. The main theorem of this section is Theorem \ref{th:main} in Subsection \ref{subsec:single}, which
proves an $\Omega(\log n)$ lower bound on $\miccoin(\sfM,\mu)$ for any $k$-pass algorithm $\sfM$ that solves the coin problem (or computes majority of input bits with large enough constant advantage). In Subsection \ref{sec:multiple}, we extend this lower bound to solving multiple instances of the coin problem simultaneously.

\paragraph{Additional notations.}  We will use notation $\one$ to denote the indicator function, which takes in a Boolean expression (as a subscript) and outputs $1$ if the expression is true and $0$ if it is false. For example, $\one_{X\neq 0}=\begin{cases}0 \text{ if } X=0\\
1 \text{ otherwise}
\end{cases}$. 
 In this section, we will also use laws of total expectation and total variance, which can be stated as follows. By the law of total expectation,
\[\bbE_{y\sim Y}[Z|Y=y]=\bbE[Z],\]
that is, the expected value of the conditional expectation of $Z$ given $Y$ is the same as the expected value of $Z$. While representing conditional expectation, we will also use the notation $\bbE_Y$ instead of $\bbE_{y\sim Y}$. Next, by the law of total variance, the expected conditional variance of $Z$ given $Y$ is at most the variance of $Z$. More generally,
\[\bbE_{X,Y}[\Var(Z|X=x,Y=y)]\le \bbE_{X}\Var(Z|X=x).\]
We give a short proof of the law of total variance:
\begin{align*}
\bbE_{X,Y}[\Var(Z|X=x,Y=y)]
&=\bbE_{X,Y}[\bbE[Z^2|X=x,Y=y]-(\bbE[Z|X=x,Y=y])^2]\\
&=\bbE[Z^2]-\bbE_{X,Y}(\bbE[Z|X=x,Y=y])^2\\
&\le \bbE[Z^2]-\bbE_{X}(\bbE_Y[\bbE[Z|X=x,Y=y]])^2\tag{Jensen's inequality}\\
&=\bbE[Z^2]-\bbE_{X}(\bbE[Z|X=x])^2\\
&=\bbE_{X}[\Var(Z|X=x)].
\end{align*}

We will prove our $k$-pass lower bound for computing majority (or approximating sum) on the uniform distribution by reducing it to the one-pass lower bound proven by \cite{braverman2020coin}, which is stated as follows: 
\begin{theorem}[\cite{braverman2020coin}, Corollary 14]\label{thm:previous1}
Given a stream of $n$ uniform $\zo$ bits $X_1,\ldots, X_n$, let $\sfO$ be a one-pass algorithm that uses private randomness. For all $\epsl>c_0 n^{-\frac{1}{20}}$, there exists $\de\ge c_1 \epsl^5$ (for small enough
	constant $c_1>0$ and large enough constant $c_0>0$), such that if
		\begin{equation}
	\label{eqbgw:1}
	IC(\sfO)=\sum_{j=1}^n\sum_{\ell=1}^j\I(\sfO_i;X_{\ell}|\sfO_{\ell-1}) \le \de n \log n, 
	\end{equation}
	then 
	\begin{equation}
	\label{eqbgw:2}
	\bbE_{\sfO_n}\left[\left(\bbE\left[ \sum_{j=1}^n X_j ~ \middle| ~
\sfO_n=o_n \right]\right)^2  \right] \le \epsl n.  
	\end{equation}
 Here, $\sfO_j$ represents the memory state of the one-pass algorithm $\sfO$ after reading $j$ input elements.

\end{theorem}
%As  $\I(\sfO_j;\sfO_{0}|\sfO_{\ell-1})=\I(\sfO_j;\sfO_{0}|\sfO_{\ell-1}, X_\ell)=0$ (where $\sfO_0$ is the starting memory state), we can rewrite $IC(\sfO)$ as follows: 
%\[IC(\sfO)=\sum_{j=1}^n\sum_{\ell=1}^j \I\left(\sfO_{j};X_{\ell}|\sfO_0,\sfO_{\ell-1}\right).\]
%The above expression is a special case of $\miccoin$ and would be useful in the following subsections.

\subsection{Multi-Pass Lower Bound for Solving Single Coin Problem}\label{subsec:single}

%Main theorem
Our main theorem for the $k$-pass coin problem is stated as follows:
\begin{theorem}\label{th:main}
	Let $\sfM$ be a $k$-pass algorithm on a stream of $n$ $i.i.d.$ uniform $\zo$ bits, $X_1,\ldots,X_n$. For all constants $\ve>0$ and $n$ greater than a sufficiently large constant, there exists constants $\de,\lam>0$, such that if $k<n^\lam$,
		\begin{equation}
	\label{eq:1}
	\miccoin(\sfM) \le \de n \log n\;\;\;\;\;\text{and}\;\;\;\;\; \forall i \in \{0,\ldots,k\},\;\Ent(\sfM_{i})\le n^\lam,
	\end{equation}
	then 
	\begin{equation}
	\label{eq:2}
	\bbE_{\sfM_{(k,n)}}\left(\bbE\left[ \sum_{j=1}^n X_j ~ \middle| ~
\sfM_{(k,n)} \right]  \right)^2 \le \ve n.  
	\end{equation}
\end{theorem}

Theorem \ref{th:main} and Lemma \ref{lem:upperboundofmic'}, along with (\cite{braverman2020coin}, Claim 6)-- which proved an $\Omega(n)$ lower bound on the L.H.S. of Equation \eqref{eq:2} for any algorithm whose output computes majority with 0.999 advantage,  give us the following corollary.

\begin{corollary}
\label{cor:singlemaj}
Let $X_1,X_2,\ldots, X_n$ be a stream of $i.i.d.$ uniform $\zo$ bits. Let $k<\log n$, and $\sfM$ be a $k$-pass streaming algorithm (possibly using private randomness)  which goes over the stream $k$ times in order and outputs the majority bit with at least 0.999 probability (over the input distribution and private randomness). Then $\sfM$ uses at least $\Omega\left(\frac{\log n}{k}\right)$ memory. 
\end{corollary}

\proof[Proof of Theorem \ref{th:main}] Let $\sfM$ be a $k$-pass algorithm on a stream of $n$ $i.i.d.$ uniform $\zo$ bits, $X_1,\ldots,X_n$, such that $k\le n^\lam$, $\Ent(\sfM_i)\le n^\lam, \; \forall i\in\{0,\ldots,k\}$ and 
\[\bbE_{\sfM_{(k,n)}}\left(\bbE\left[ \sum_{j=1}^n X_j ~ \middle| ~
\sfM_{(k,n)} \right]  \right)^2 \ge \ve n.\]
Using $\sfM$, we construct a single pass algorithm $\sfO$ (Algorithm \ref{al:a}) such that, given a stream of $n$ $i.i.d.$ uniform $\zo$ bits $Y_1,\ldots, Y_n$, the following holds
\begin{enumerate}
\item for the information cost of $\sfO$ (Lemma \ref{cl:ainf}):\begin{equation}\label{eq:thmmain1}IC(\sfO)\; = \;\sum_{j=1}^n \sum_{\ell=1}^j \I(\sfO_{j};Y_\ell| \sfO_{\ell-1})\;\le \; \miccoin(\sfM) + n\cdot \left(50+6\log \log n +\log{\left(\frac{k\cdot n^{\lambda}}{\ve^2}\right)}\right).
\end{equation}
Recall, $\sfO_j$ denotes the random variable for the memory state of $\sfO$ after reading $j$ bits\footnote{Note that the simulation algorithm $\sfO$ uses private randomness.}. 
\item for the output of algorithm $\sfO$ (Lemma \ref{cl:asum}): \begin{align*}
\bbE_{\sfO_n}\left[\left(\bbE\left[ \sum_{j=1}^n Y_j ~ \middle| ~
\sfO_n =\an_n\right]\right)^2  \right] \;\ge\; \frac{\ve}{2}\cdot n. 
\end{align*}
\end{enumerate}
Theorem \ref{thm:previous1} implies that there exists $\de_\ve>0$ such that $IC(\sfO)\ge  \de_\ve \cdot n\log n$. If $\lam<\frac{\de_\ve}{10}$, then for sufficiently large $n$, 
\[50+6\log \log n +\log{\left(\frac{k\cdot n^{\lambda}}{\ve^2}\right)}\le 50+6\log \log n +\log{\left(\frac{1}{\ve^2}\right)}+2\lam\log{n}< \frac{\de_\ve}{2} \log n.\]
Therefore, Equation \eqref{eq:thmmain1} implies that $\miccoin(\sfM)\ge IC(\sfO)-\frac{\de_\ve}{2}\cdot n \log n\ge \frac{\de_\ve}{2}\cdot n \log n$. Taking $\de=\de_\ve/2$ proves the theorem. 
 
\paragraph{Construction of \ $\sfO$.} Informally, $\sfO$ executes $k$ passes of $\sfM$ in parallel. Before reading the input bits $Y_1,\ldots, Y_n$, $\sfO$ samples memory states at the end of first $k-1$ passes from the joint distribution on $(\sfM_0,\ldots,\sfM_{k-1})$. $\sfO$ then modifies the given input $Y$ to $X'$ such that the parallel execution of the $k-1$ passes of the algorithm $\sfM$ on $X'_1,\ldots, X'_n$ end in the sampled memory states. $\sfO$ also maintains an approximation for the modification, that is of $\sum_{j=1}^n (X'_j-Y_j)$; this helps $\sfO$ to compute $\sum Y_j$ as long as $\sfM$ computes $\sum X'_j$ after $k$ passes. As we want $\sfO$ to have comparable information cost to that of $\sfM$, the approximation of the modification should take low memory\footnote{Note that it takes $\log n$ bits of memory to store $\sum_{j=1}^n (X'_j-Y_j)$ exactly.}. The key observation that makes such an approximation possible is: since the KL divergence of the distribution $X$, conditioned on reaching memory states $\sfM_0,\ldots,\sfM_{k-1}$, from the uniform distribution is bounded by the entropy of $(\sfM_0,\ldots,\sfM_{k-1})$ (which we assume to be $<< n$), algorithm $\sfO$ does not need to drastically modify $Y$ (which has a uniform distribution). Still, we cannot afford to store the modification exactly; however, a cruder approximation suffices, which can be computed using low memory. 

As described above, algorithm $\sfO$ has two components, 1) imitate $k$ passes of $\sfM$ simultaneously, and 2) maintain an approximation for modifying input $Y$ to a valid input $X'$ for the first $k-1$ passes of $\sfM$. To formally describe algorithm $\sfO$ (in Section \ref{subsec:a}), we first state these two components separately as algorithms $\Imi$ (in Section \ref{subsec:im}) and $\Apr$ (in Section \ref{subsec:apr}) respectively.

\subsubsection{Single-Pass Algorithm $\Imi$ Imitating $k$ Passes of $\sfM$}\label{subsec:im}
Recall that $\sfM$ is a $k$-pass algorithm that runs on a stream of $n$ $i.i.d.$ uniform $\zo$ bits, $X_1, X_2,\ldots, X_n$. We describe algorithm $\Imi$ in Algorithm \ref{al:im}. Let $\Imi_j$ (where $j\in[n]$) represent the random variable for the memory state of Algorithm $\Imi$ after reading $j$ inputs bits, and $\Imi_0$ be a random variable for the starting memory state for the algorithm. The input $Y$ to the algorithm $\Imi$ is drawn from the uniform distribution on $\zo^n$. 

Let $\sfM'_i$ denote the random variable associated with value $m'_i$ ($i\in\{0,1,\ldots, k-1)$. The distribution of $\sfM'_{<k}$ is defined at Step \ref{st:im3} of Algorithm \ref{al:im}. Let $\{\sfM'_{(i,j)}\}_{i\in[k],j\in\{0,\ldots,n\}}$ denote the random variables associated with values $\{m'_{(i,j)}\}_{i\in[k],j\in\{0,\ldots,n\}}$. The distribution of $\sfM'_{(i,j)}$ ($j\in[n]$) is defined at Step \ref{st:im2} of Algorithm \ref{al:im}, and of $\sfM'_{(i,0)}$ is defined at Step \ref{st:im4}. Let $\{X'_j\}_{j\in[n]}$ denote the random variable for value $x'_j$ in Step \ref{st:im1} of Algorithm \ref{al:im}. These distributions depend on the joint distribution on $(X,\sfM_{\le k}, \sfM_{(\le k, [1,n])})$ and the uniform distribution of $Y$. 

\begin{claim}\label{cl:imd}
The joint distribution on $X, \sfM_{< k}, \{\sfM_{(\le k, [0,n])}\}$ is identical to that on $X',\sfM'_{< k}, \{\sfM'_{(\le k, [0,n])}\}$.
\end{claim} 
In Algorithm \ref{al:im}, the random variables  $X',\sfM'_{< k}, \{\sfM'_{(\le k, [0,n])}\}$ are sampled with the aim of replicating the joint distribution. For a formal proof of the above claim, refer to Appendix \ref{sec:apim}. 

\begin{algorithm}[H]
\caption{Single pass algorithm $\Imi$ imitating $k$ passes of $\sfM$}
\label{al:im}
\textbf{Input}: a stream of $n$ bits $y_1,\ldots, y_n$, drawn from uniform distribution on $\zo^n$
\begin{algorithmic}[1]
%\STATE Sample $m'_0,m'_1,\ldots,m'_k\sim (\sfM_0,\sfM_1,\ldots,\sfM_k)$
\STATE \label{st:im3} Sample $m'_0,m'_1,\ldots,m'_{k-1}\sim (\sfM_0,\sfM_1,\ldots,\sfM_{k-1})$
\COMMENT{$\Imi$ samples memory states for the end of first $k-1$ passes}
%\STATE $\imi_0\leftarrow (m'_0,m'_1,\ldots,m'_k)$
\STATE $\imi_0\leftarrow (m'_0,m'_1,\ldots,m'_{k-1})$
\COMMENT{$\Imi$ stores these memory states for the entire algorithm}
\STATE \label{st:im4} $\forall i\in[k]$, $m'_{(i,0)}\leftarrow m'_{(i-1)}$  
\COMMENT{Starting memory states for the $k$ passes of $\sfM$}
\FOR {$j=1$ to $n$} 
%\STATE \label{st:im1} $\beta_j\leftarrow \left(\Pr\left[X_j=1\;|\; \sfM_{(\le k,j-1)}=m'_{(\le k,j-1)}, \sfM_{\le k}=m'_{\le k}\right]-\frac{1}{2}\right)$ \COMMENT{Can be calculated using $im_{j-1}$}
\STATE \label{st:im1} $\beta_j\leftarrow \left(\Pr\left[X_j=1\;|\; \sfM_{(\le k,j-1)}=m'_{(\le k,j-1)}, \sfM_{< k}=m'_{<k}\right]-\frac{1}{2}\right)$ \COMMENT{Can be calculated using $im_{j-1}$}
\IF {$\beta_j > 0$} 
\IF{$y_j=1$}
 \STATE $x_j'\leftarrow y_j$ 
 \ELSE
 \STATE $x_j'\leftarrow y_j$ with probability $1-2\beta_j$, and $x_j'\leftarrow 1$ otherwise
 \ENDIF
\ELSIF{$\beta_j \le 0$}
\IF{$y_j=1$}
 \STATE  $x_j'\leftarrow y_j$ with probability $1+2\beta_j$, and $x_j'\leftarrow -1$ otherwise
 \ELSE
 \STATE $x_j'\leftarrow y_j$
\ENDIF
\ENDIF
 %\STATE \label{st:im2} Sample $(m'_{(1,j)},m'_{(2,j)}, \ldots, m'_{(k,j)})$ from the joint distribution on \[\left((\sfM_{(1,j)},\sfM_{(2,j)},\ldots, \sfM_{(k,j)})~\middle|~(\sfM_{(\le k,j-1)}=m'_{(\le k,j-1)},\; \sfM_{\le k}=m'_{\le k},\; X_j=x_j'\right)\]\\
 \STATE \label{st:im2} Sample $(m'_{(1,j)},m'_{(2,j)}, \ldots, m'_{(k,j)})$ from the joint distribution on \[\left((\sfM_{(1,j)},\sfM_{(2,j)},\ldots, \sfM_{(k,j)})~\middle|~(\sfM_{(\le k,j-1)}=m'_{(\le k,j-1)},\; \sfM_{< k}=m'_{< k},\; X_j=x_j'\right)\]\\
\COMMENT{Given $im_{(j-1)}$, $\Imi$ executes $j$th time-step for all passes of $\sfM$}
 %\STATE $im_j\leftarrow (m'_{(1,j)},m'_{(2,j)}, \ldots, m'_{(k,j)},m'_0,m'_1,\ldots,m'_k)$\\
 \STATE $im_j\leftarrow (m'_{(1,j)},m'_{(2,j)}, \ldots, m'_{(k,j)},m'_0,m'_1,\ldots,m'_{k-1})$\\
 \COMMENT{At the $j$th time-step, $\Imi$ stores these memory states}
\ENDFOR 
\end{algorithmic}
\end{algorithm}

\begin{claim}\label{cl:iminf}
The information cost of algorithm $\Imi$ is at most the information cost of the $k$-pass algorithm $\sfM$, that is,
\[IC(\Imi)=\sum_{j=1}^n\sum_{\ell=1}^j \I\left(\Imi_{j};Y_{\ell}|\Imi_{\ell-1}\right) \le \sum_{j=1}^n\sum_{\ell=1}^j \I\left(\sfM_{(\le k,j)};X_{\ell}|\sfM_{< k},\sfM_{(\le k,\ell-1)}\right)=\miccoin(\sfM).\]
Here, $X$ and $Y$ are both drawn from the uniform distribution on $\zo^n$.

\end{claim}
The inequality follows  from a data processing inequality on $Y\rightarrow X'\rightarrow$ Algorithm \ref{al:im}, and the equivalence of random variables stated in Claim \ref{cl:imd}. See Appendix \ref{sec:apim} for a formal proof.

\subsubsection{Low Information Algorithm for Approximating Sum}\label{subsec:apr}

%simpler algorithm doesn't work

We develop $\Apr$ for the general problem of approximating the sum of $n$ elements, each in $\{-1,0,1\}$. The problem is as follows: given parameters $\gamma>0$, $\B>\gamma \sqrt{n}$, and a stream of $n$ elements $\e_1,\ldots,\e_n \in \{-1,0,1\}$ jointly drawn from a distribution $\calD$ (such that $\bbE_{\e\sim \calD}\left[\sum_{j=1}^n \one_{\e_j\neq 0}\right]\le \B$), the aim is to output $\left(\sum_{j=1}^n \e_j\right)$ up to an additive error of $\gamma \sqrt{n}$. Let $R^{\Apr}=(R_1^\Apr,\ldots,R_n^\Apr)$, where $R_j^\Apr$ denotes the random variable for private randomness used by algorithm $\Apr$ at the $j$th time-step; we formalize the error guarantee as $\bbE_{\e\sim \calD, r\sim R^\Apr} \left[\left|\Apr(\e,r)-\sum_{j=1}^n \e_j\right|\right]\le \gamma\sqrt{n}$. Additionally, we establish that the streaming algorithm $\Apr$ (described in Algorithm \ref{al:apr}) has low information cost -- the memory state at each time-step has low entropy, that is, $\forall i\in [n]$, $\Ent (\Apr_i)\sim 2\log{\left( \frac{\B}{\gamma\sqrt{n}}\right)}$. Note that, the exact computation of $\sum_j \e_j$ requires $\log n$ memory. 

Informally, $\Apr$ samples each $\e_j$ with probability $p\sim \frac{\B}{\gamma^2 n}$ and maintains their sum using a counter $\Deltasm$. It is easy to see that $\Deltasm/p$ is an approximation of $\sum_j \e_j$ (with additive error $\gamma\sqrt{n}$) as long as $\sum_j \one_{\e_j\neq 0}$ is bounded by $\sim\B$. As $\B$ is an upper bound only on the expectation of $\sum_j \one_{\e_j\neq 0}$, the algorithm $\Apr$ needs to find another way to approximate the sum whenever $\sum_j \one_{\e_j\neq 0}>> B$. For this, $\Apr$ maintains two more counters $\Cb$ and $\Deltalg$, where $\Cb$ counts the number of elements $\e_j$ sampled in the sum $\Deltasm$, and $\Deltalg$ stores $\sum_j \e_j$ exactly whenever counter $\Cb$ becomes $ >> p \B$. $\Apr$ is formally described in Algorithm \ref{al:apr}.

\begin{algorithm}[H]
\caption{Algorithm $\Apr$ for approximate sum}
\label{al:apr}
\textbf{Input stream}: $\e_1,\ldots, \e_n$, drawn from joint distribution $\calD$ on $\{-1,0,1\}^n$\\
\textbf{Given parameters}: $\gamma>0$, $ \B> \gamma \sqrt{n}$\\
\textbf{Goal}: $\forall \e\in\{-1,0,1\}^n$, $\left|\Apr(\e,r)-\sum_{j=1}^n \e_j\right|\le \frac{\gamma}{2}\sqrt{n}$ with probability at least $1-\frac{1}{n^3}$ over the private randomness $r\sim R^\Apr$
\begin{algorithmic}[1]
\STATE Let $p =\min\left\{ 6000\log^2 n\cdot \left(\frac{\B}{\gamma^2 n}\right),1\right\}\;$ \hspace{1cm}\COMMENT{probability of sampling}
\STATE $\Deltasm\leftarrow 0$ \hspace{1cm}\COMMENT{$\Deltasm$ maintains an approximation for $p\cdot\left(\sum_j \e_j\right)$} 
\STATE $\Cb\leftarrow 0$ \hspace{1cm}\COMMENT{$\Cb$ approximates $p\cdot\left(\sum_j \one_{\e_j\neq 0}\right)$}
\STATE $\Deltalg\leftarrow 0$ \hspace{1cm}\COMMENT{$\Deltalg$ computes $\left(\sum_j \e_j\right)$ exactly when $\Deltasm/p$ is not a good approximation}
\FOR{$j=1$ to $n$}
\IF{$\Cb<  20\log n \cdot p \B$}
\STATE \label{step:r} Let $r_j$ be $1$ with probability $p$ and $0$ otherwise 
\IF{$r_j=1$}
\STATE $\Deltasm\leftarrow \Deltasm + \e_j$
\STATE $\Cb\leftarrow \Cb+\one_{\e_j\neq 0}$ \hspace{1cm}\COMMENT{Sample $\e_j$ and update the counters with probability $p$}
\ENDIF
\ELSIF{$\Cb\ge 20\log n \cdot p \B$}
\STATE $\Deltalg\leftarrow \Deltalg + \e_j$ 
\ENDIF 
%\COMMENT{$\Deltasm/p$ is a good approximation for $\left(\sum_j \e_j\right)$ only if $\left(\sum_j \one_{\e_j\neq 0}\right) \le 80\B \log n $}

\ENDFOR
\RETURN $\max\{\min\{\Deltasm/p + \Deltalg, n\},-n\}\;$ \hspace{1cm}\COMMENT{Project $\Deltasm/p + \Deltalg$ within $[-n,n]$}
\end{algorithmic}
\end{algorithm}
In Propositions \ref{cl:aprsum} and \ref{cl:aprinf}, we establish the approximation and information cost guarantees for the algorithm $\Apr$. Before, we note that $\Apr$ uses private randomness at Step \ref{step:r} of the algorithm and define $R_j^\Apr$ to be a $\Ber(p)$ random variable for all $j\in[n]$. 

\begin{proposition}\label{cl:aprsum}
$\forall \gamma>\frac{4}{\sqrt{n}}, \B>\gamma\sqrt{n}$, 
%distribution $\calD$ on $\{-1,0,1\}^n$ such that $\bbE_{\e\sim \calD}\left[\sum_{j=1}^n \one_{\e_j\neq 0}\right]\le \B$,
the output of Algorithm \ref{al:apr} ($\Apr)$ on every input stream $\e\in\{-1,0,1\}^n$ satisfies the following with probability at least $1-\frac{1}{n^3}$ over the private randomness $r\sim R^\Apr$,
\[\left|\Apr(\e,r)-\sum_{j=1}^n \e_j\right|\le \frac{\gamma}{2}\sqrt{n},\]
which further implies that $\forall$ distribution $\calD$ on $\{-1,0,1\}^n$, 
\[\bbE_{\e\sim \calD,r\sim R^\Apr} \left[\left(\Apr(\e, r)-\sum_{j=1}^n \e_j\right)^2\right]\le \gamma^2n.\]
\end{proposition}

\proof[Proof Sketch] It's easy to see that $\Deltasm/p$ approximates $\left(\sum_j \e_j\right)$ up to an additive error of $O(\gamma\sqrt{n}$),  with high probability over $r$, as long as $\left(\sum_j \one_{\e_j\neq 0}\right) \le 80\B \log n $. After counter $\Cb$ crosses $20\log n \cdot p \B$, $\Deltalg$ exactly computes the sum of the remaining elements. Thus, it remains to show that with high probability over $r$, $\Cb$ reaches $20\log n \cdot p \B$ before the number of non-zero elements in the stream becomes larger that $80\B \log n$.
See Appendix \ref{sec:apapr} for a detailed proof. 
\qed

%For every input $\e$, we show that $\Apr$ outputs $\left(\sum_{j=1}^n \e_j\right)\pm \gamma \sqrt{n}$, with probability at least $1-\frac{1}{n}$, by using the following two claims:

\begin{proposition}\label{cl:aprinf}
$\forall \gamma>\frac{4}{\sqrt{n}}, \B>\gamma\sqrt{n}$, distributions $\calD$ on $\{-1,0,1\}^n$ such that $\bbE_{\e\sim \calD}\left[\sum_{j=1}^n \one_{\e_j\neq 0}\right]\le \B$, memory states of Algorithm \ref{al:apr} ($\Apr$) satisfies the following:
\[\forall j\in\{0,\ldots,n\}, \Ent(\Apr_j) \le 40+6\log \log n+2\log \left({\frac{B}{\gamma \sqrt{n}}}\right)\]
Here, $\Apr_j$ denotes the random variable for $\Apr$'s memory state after reading $j$ input elements, and it depends on input $\e$, as well as the private randomness $r\sim\Ra^\Apr$ used by the algorithm.

\end{proposition}

\proof[Proof Sketch]
At every time-step, $\Apr$ maintains three counters $\Deltasm$, $\Cb$ and $\Deltalg$. As the algorithm stops increasing counts $\Deltasm$ and $\Cb$ whenever $\Cb$ reaches $20\log n\cdot p\B$, entropy of both $\Deltasm$ and $\Cb$ is at most $O(\log {(20\log n\cdot p\B)})=O(\log\log n +\log(\B^2/\gamma^2n))$. To bound the entropy of counter $\Deltalg$, it is enough to show that, with high probability, $\Deltalg$ remains 0 at all time-steps.  Firstly, we show that, with high probability over stream $\e$ being drawn from $\calD$, $\left(\sum_j \one_{\e_j\neq 0}\right)\le 4\B\log n$. Next, we show that, whenever the number of non-zero elements in the stream are less than $4\B\log n$, $\Cb$ remains below $20\log n\cdot p\B$ with high probability over $r$; thus, $\Deltalg$ remains 0. See Appendix \ref{sec:apapr} for a formal proof. 
\qed

\subsubsection{Single-Pass Algorithm $\sfO$ for Computing Majority Using $k$-Pass Algorithm $\sfM$}\label{subsec:a}
In this subsection, we describe the one-pass algorithm $\sfO$ that approximates the sum almost as well as the $k$-pass algorithm $\sfM$, while having similar information cost to $\sfM$. $\sfO$ runs Algorithm \ref{al:im} ($\Imi$) to imitate the $k$-passes of $\sfM$ -- $\Imi$ modifies input bit $y_j$ at the $j$-th time-step to bit $x'_j$.  In parallel, $\sfO$ runs Algorithm \ref{al:apr} ($\Apr$) on the modification -- the $j$-th input element to $\Apr$ is $(y_j-x'_j)\in\{-1,0,1\}$. After reading $y_j$, $\sfO$ runs $j$th time-steps of algorithms $\Imi$ and $\Apr$ (the input to $\Apr$ is generated on the fly), and stores the $j$th memory states of both the algorithms. While describing $\sfO$ formally in Algorithm \ref{al:a} (where parameters $\gamma$ and $\B$ would be decided later), we will restate algorithm $\Imi$ and use $\Apr$ as a black-box. As used in Subsection \ref{subsec:apr}, $\Ra_j^\Apr$ represents the private randomness used by algorithm $\Apr$ at the $j$th time-step and $\Apr_j$ represents the random variable for $j$th memory state ($r_j$ and $apr_j$ represent their instantiations). The input to $\Apr$ is denoted by $\e$. Let $f^{\Apr}_j$ $(j\in[n])$ represent the $j$th transition function for algorithm $\Apr$, that is, $apr_j=f^{\Apr}_j(apr_{j-1},\e_j,r_j)$. Let $\Apr_j(apr_{j-1},\e_j)$ denote the random variable for the $j$th memory state, when the $j$th input element is $\e_j$, $(j-1)$th memory state is $apr_{j-1}$ and private randomness $r_j$ is drawn from $\Ra^\Apr_j$. 
\begin{algorithm}[H]
\caption{Single pass algorithm $\sfO$ using $k$-pass algorithm $\sfM$ for computing majority}
\label{al:a}
\textbf{Input}: a stream of $n$ $i.i.d$ uniform $\zo$ bits $Y_1,\ldots, Y_n$\\
\textbf{Given parameters}: $\gamma>0, \B>\gamma \sqrt{n}$\\
\textbf{Goal}: approximate $\sum_{j=1}^n Y_j$.
\begin{algorithmic}[1]
\STATE  Sample $m'_0,m'_1,\ldots,m'_{k-1}\sim (\sfM_0,\sfM_1,\ldots,\sfM_{k-1})$
\COMMENT{Sample memory states for the end of the  first $k-1$ passes}
%\STATE $\imi_0\leftarrow (m'_0,m'_1,\ldots,m'_k)$
\STATE $\imi_0\leftarrow (m'_0,m'_1,\ldots,m'_{k-1})$
\STATE $\forall i\in[k]$, $m'_{(i,0)}\leftarrow m'_{(i-1)}$  
\COMMENT{Starting memory states for the $k$ passes of $\sfM$}
\STATE Initialize Algorithm \ref{al:apr} ($\Apr$) with parameters $\gamma$ and $\B$ 
\STATE Sample $\apr_0\sim \Apr_0$ \COMMENT{For Algorithm \ref{al:apr}, the starting state is deterministic}
\FOR {$j=1$ to $n$} 
%\STATE \label{st:im1} $\beta_j\leftarrow \left(\Pr\left[X_j=1\;|\; \sfM_{(\le k,j-1)}=m'_{(\le k,j-1)}, \sfM_{\le k}=m'_{\le k}\right]-\frac{1}{2}\right)$ \COMMENT{Can be calculated using $im_{j-1}$}
\STATE \label{st:a1} $\beta_j\leftarrow \left(\Pr\left[X_j=1\;|\; \sfM_{(\le k,j-1)}=m'_{(\le k,j-1)}, \sfM_{< k}=m'_{<k}\right]-\frac{1}{2}\right)$ \COMMENT{Can be calculated using $im_{j-1}$}
\IF {$\beta_j \ge 0$} 
\IF{$y_j=1$}
 \STATE $x_j'\leftarrow y_j$ 
 \ELSE
 \STATE $x_j'\leftarrow y_j$ with probability $1-2\beta_j$, and $x_j'\leftarrow 1$ otherwise
 \ENDIF
\ELSIF{$\beta_j \le 0$}
\IF{$y_j=1$}
 \STATE  $x_j'\leftarrow y_j$ with probability $1+2\beta_j$, and $x_j'\leftarrow -1$ otherwise
 \ELSE
 \STATE $x_j'\leftarrow y_j$
\ENDIF
\ENDIF \label{st:a2}
\STATE $\e_j\leftarrow (y_j-x_j')$\COMMENT{Setting $j$th input element to $\Apr$}
\STATE Sample $apr_j\sim \Apr_j(apr_{j-1}, \e_j)$ \COMMENT{Run $j$th step of $\Apr$ given $apr_{j-1}$}
 %\STATE \label{st:im2} Sample $(m'_{(1,j)},m'_{(2,j)}, \ldots, m'_{(k,j)})$ from the joint distribution on \[\left((\sfM_{(1,j)},\sfM_{(2,j)},\ldots, \sfM_{(k,j)})~\middle|~(\sfM_{(\le k,j-1)}=m'_{(\le k,j-1)},\; \sfM_{\le k}=m'_{\le k},\; X_j=x_j'\right)\]\\
 \STATE  Sample $(m'_{(1,j)},m'_{(2,j)}, \ldots, m'_{(k,j)})$ from the joint distribution on \[\left((\sfM_{(1,j)},\sfM_{(2,j)},\ldots, \sfM_{(k,j)})~\middle|~(\sfM_{(\le k,j-1)}=m'_{(\le k,j-1)},\; \sfM_{< k}=m'_{< k},\; X_j=x_j'\right)\]\\
\COMMENT{Given $im_{(j-1)}$, execute $j$th time-step for all passes of $\sfM$}
 %\STATE $im_j\leftarrow (m'_{(1,j)},m'_{(2,j)}, \ldots, m'_{(k,j)},m'_0,m'_1,\ldots,m'_k)$\\
 \STATE $im_j\leftarrow (m'_{(1,j)},m'_{(2,j)}, \ldots, m'_{(k,j)},m'_0,m'_1,\ldots,m'_{k-1})$\\
\STATE \label{st:a3} $\an_j\leftarrow (im_j,apr_j)$\COMMENT{$j$th memory state of algorithm $\sfO$}
\ENDFOR \\
\textbf{Output}: $\an_n=(im_n, apr_n)=(m'_{(\le k,n)},m'_{<k}, apr_n)$
%\RETURN $apr_n+\bbE\left[\sum_{j=1}^n X'_j\middle| \sfM'_{(k,n)}=m'_{(k,n)}\right]$ \COMMENT{Can be calculated using $\an_n$; $im_n$ contains $m'_{(k,n)}$}
\end{algorithmic}
\end{algorithm}

The random variables $Y,X'$ are as defined in Subsection \ref{subsec:im}, where $Y$ is drawn from uniform distribution on $\zo^n$ and $X'_j$ corresponds to value $x_j'$ as in Algorithm \ref{al:a}. Let $\calD$ be the joint distribution generated by Algorithm \ref{al:a} on inputs to $\Apr$ that is, the joint distribution on $(Y_1-X'_1,Y_2-X'_2,\ldots, Y_n-X'_n)$. Let $\Ee_j$ be the random variable for the $j$th input element to $\Apr$, that is, $\Ee_j=Y_j-X'_j$.  Let $\{\sfM'_{(i,j)}\}_{i\in[k],j\in\{0,\ldots,n\}}$ and $\sfM'_{<k}$ be random variables as defined in Subsection \ref{subsec:im} (these are random variables for corresponding values that appear in Algorithm \ref{al:a}). We prove the following claim for distribution $\calD$, when the end states of the $k$-pass algorithm $\sfM$ has bounded entropy.

\begin{claim}\label{cl:acomb}
$\bbE_{\e\sim \calD}\left[\sum_{j=1}^n\one_{\e_j\neq 0}\right]=\bbE\left[\sum_{j=1}^n\one_{Y_j\neq X'_j}\right]\le \sqrt{n\cdot \Ent(\sfM_{<k})}$.
\end{claim}
Informally, we can relate the probability of modification at step $j$ to the information that end memory states have about $X_j$, conditioned on the previous memory states. The claim follows from the fact that the sum of this information over $j$, is bounded by the entropy of the end states. Note that the above claim is tight if an end state computes the majority bit. See Appendix \ref{sec:apsingle} for a formal proof.
Recall that, $\sfM$ is a $k$-pass algorithm such that $\Ent(\sfM_i)\le n^\lambda$ for all $i\in\{0,\ldots,k\}$.  %Additionally, we assumed that the output of $\sfM$ is correlated with sum of the input bits, that is,
%\[\bbE_{\sfM_{(k,n)}}\left(\bbE\left[\sum_{j=1}^n X_j\middle| \sfM_{(k,n)}=m_{(k,n)}\right]\right)^2\ge \ve n.\]
Claim \ref{cl:acomb} immediately gives the following corollary.
\begin{corollary}\label{cor:acomb}
As $\Ent(\sfM_i)\le n^\lambda, \forall i\in\{0,\ldots, k\}$, 
\[\bbE_{\e\sim \calD}\left[\sum_{j=1}^n\one_{\e_j\neq 0}\right]=\bbE\left[\sum_{j=1}^n\one_{Y_j\neq X'_j}\right]\le \sqrt{n}\cdot \sqrt{k n^\lambda}.\]
\end{corollary}

Corollary \ref{cor:acomb} suggests a value for parameter $\B$ that Algorithm \ref{al:a} should run Algorithm $\Apr$ on, so as to use approximation guarantees from Claim \ref{cl:aprsum}. Let $\sfO$ be Algorithm \ref{al:a} with parameters $\gamma=\frac{\ve}{10}$ and $\B=\sqrt{k\cdot n^{1+\lambda}} $. We prove the following lemmas regarding information cost and output of algorithm $\sfO$.   See Appendix \ref{sec:apsingle} for detailed proofs of these lemmas. 

\begin{lemma}\label{cl:ainf}
For all $\ve>\frac{100}{\sqrt{n}}$, $\lambda>0$, $IC(\sfO)\le \miccoin(\sfM)+n\cdot \left(50+6\log \log n +\log{\left(\frac{k\cdot n^{\lambda}}{\ve^2}\right)}\right)$.
\end{lemma}
Once we have Claim \ref{cl:iminf} and Proposition \ref{cl:aprinf} in place, Lemma \ref{cl:ainf} follows from careful disentanglement of the information costs for subroutines $\Imi$ and $\Apr$ used in Algorithm $\sfO$. 

\begin{lemma}\label{cl:asum}
For all $\ve>\frac{100}{\sqrt{n}}$, 
\[\bbE_{\sfM_{(k,n)}}\left(\bbE\left[\sum_{j=1}^n X_j\middle| \sfM_{(k,n)}=m_{(k,n)}\right]\right)^2\ge \ve n\implies \bbE_{\sfO_{n}}\left(\bbE\left[\sum_{j=1}^n Y_j\middle| \sfO_n=\an_n\right]\right)^2\ge \frac{\ve}{2} n.\]
\end{lemma}
Intuitively, Lemma \ref{cl:asum} shows that if output of the $k$-pass algorithm $\sfM$ gives information about $\sum_j X_j$ (measured by reduction in the variance), then the output of one-pass algorithm $\sfO$ also gives information about $\sum_j Y_j$ -- sum of the input stream to $\sfO$. The former guarantee implies the output of $\sfO$ contains information about $\sum_{j}X_j'$ (the modified input); as $\sfO$ stores an approximation for $\sum_j (Y_j -X'_j)$, this implies that it also has information about $\sum_j Y_j$. All that remains to show is that the approximation for modification has an additive error of at most $O(\ve\sqrt{n})$, with high probability. For this, we use Claim \ref{cl:acomb} and $\ell_2$ approximation guarantee for $\Apr$ from Claim \ref{cl:aprsum}. 

\subsection{Lower Bound for Solving Multiple Copies Simultaneously}\label{sec:multiple}

In this section, we will define the $\tcoins$ -- given $t$ interleaved input streams generated by $n$ $i.i.d.$ uniform $\zo$ bits each, output the majority of an arbitrary stream at the end of $k$ passes. We want to show that any $k$-pass streaming algorithm that solves the $\tcoins$ requires $\Omega(\frac{t\log n}{k} )$ bits of memory. Formally, the input is a sequence of $nt$ updates of the form $(x_j, s_j)\in \zo\times [t]$; here updates with $s_j=s$ are interpreted as the input stream for the $s$th instance of the coin problem. After $k$-passes over these updates, the algorithm is given a random index $s^*\in [t]$ and the goal is to output the majority for the input stream corresponding to the $s^*$th instance. We refer to the sequence $\{s_j\}_{j\in[nt]}$ as the order of the stream, which dictates how the $t$ instances are interleaved. Given an order $\{s_j\}_{j\in[nt]}$, for all $s\in[t]$, let $J_s\subseteq [nt]$ be the set of $j$ for which $s_j=s$. Let $q_s\colon [|J_s|]\rightarrow [nt]$ be defined as follows: $q_s(u)=j$ if $x_j$ is the $u$-th element corresponding to the $s$th instance of the coin problem ($s_j=s$). We will first prove a memory/information lower bound for algorithms on inputs with a fixed \textit{good} order (as defined in \cite{braverman2020coin}).
\begin{definition}[Good Order]\label{def:gdord}
An order $\{s_j\}_{j\in[nt]}$ is called \emph{good} if for all $s\in [t]$,
$$|J_s|\ge n/2,$$ 
and for all $s\in[t], \sqrt{n}\le u$ and $u<v \le |J_s|$,
$$q_s(v)-q_s(v-u)\ge \frac{k}{2}u.$$
\end{definition}
It is easy to see that if we get the first input bit for all instances, then get the second bit for all of them and so on, then it is a good order. That is, $s_j=((j-1)\mod t) +1, j\in[nt]$, is a good order. 
\cite{braverman2020coin} showed the following for one-pass streaming algorithms solving the $\tcoins$:
\begin{theorem}[\cite{braverman2020coin}, Theorem 16]
\label{thmbgw:tcoins}
Given a fixed good order $\{s_j\}_{j\in[nt]}$ (Definition \ref{def:gdord}) and a stream of $nt$ updates of the form $(X_j, s_j)$ where $X_j$s are $i.i.d.$ uniform $\zo$ bits, let $\sfO$ be an $nt$-length single-pass streaming algorithm (that possibly uses private randomness) for the $\tcoins$. For all constants $\ve>0$, there exists $\de>0$, such that if
\begin{equation*}
		\bbE_{\sfO_{nt}}\left[\bbE\left[ \left( \sum_{j\in J_s} X_j \right) ~ \middle| ~
		\sfO_{nt}=o_{nt}\right]^2  \right] > \ve |J_s| 
		\end{equation*} for at least $1/2$ fraction of $s\in[t]$, then
			\begin{equation*}
	\sum_{j=1}^{n t}  \sum_{\ell=1}^j\I(\sfO_{j}; X_{\ell}|\sfO_{\ell-1})>\de n t^2 \log n.
	\end{equation*}

\end{theorem}

We prove the following theorem for a $k$-pass algorithm $\sfM$ solving the $\tcoins$. The proof is similar to proof of Theorem \ref{th:main} and is deferred to Appendix \ref{ap:multiplecoin}.
\begin{theorem}
\label{thm:tcoinsvar}
Given a fixed good order $\{s_j\}_{j\in[nt]}$ (Definition \ref{def:gdord}) and a stream of $nt$ updates of the form $(X_j, s_j)$ where the $X_j$s are $i.i.d.$ uniform $\zo$ bits, let $\sfM$ be an $nt$-length $k$-pass streaming algorithm (possibly using private randomness) for the $\tcoins$. For all constants $\ve>0$ and $n$ greater than a large enough constant, there exists $\de,\lam>0$, such that if $k,t\le n^\lam$, $\Ent(\sfM_i)\le n^\lam, \forall i\in \{0,\ldots,k\}$ and
\begin{equation}
		\label{eqorder:1}
		\bbE_{\sfM_{(k,nt)}}\left[\bbE\left[ \left( \sum_{j\in J_s} X_j \right) ~ \middle| ~
		\sfM_{(k,nt)}=m_{(k,nt)}\right]^2  \right] > \ve |J_s| 
		\end{equation} for at least a $1/2$ fraction of $s\in[t]$, then
			\begin{equation}
	\label{eqorder:2}
	\miccoin(\sfM)=\sum_{j=1}^{n t}  \sum_{\ell=1}^j\I(\sfM_{(\le k,j)}; X_{\ell}|\sfM_{(\le k,\ell-1)},\sfM_{<k})>\de n t^2 \log n.
	\end{equation}
\end{theorem}

Theorem \ref{thm:tcoinsvar} immediately implies a memory lower bound for the $\tcoins$, where the goal is to output the majority of a random instance given at the end of the stream.
\begin{corollary}\label{cor:tcoinsfixed}
Fix a \emph{good order} $\{s_j\}_{j\in[nt]}$. Suppose we are given a sequence of $nt$ $i.i.d.$ stream updates, the $j$-th of which has the form $(X_j, s_j)$, where $X_j$ is chosen uniformly in $\zo$. We interpret this as $t$ independent instances of the coin problem, where the $s$-th instance of the coin problem consists of the sequence of bits $X_{q_s(1)}, \ldots, X_{q_s(|J_s|)}$. Suppose there is a $k$-pass streaming algorithm $\sfM$ which goes over the stream of updates $k$ times in order, and given an $s^*\in_R [t]$ at the end of the stream, outputs the majority bit of the $s^*$-th instance of the coin problem with probability at least $1-\frac{1}{2000}$. Let $k<\log n$ and $t<n^{c}$ for a sufficiently small constant $c>0$. Then $\sfM$ uses $\Omega\left( \frac{t\log n}{k}\right)$ bits of memory.
\end{corollary}
Here, the probability of success is over the input stream, the private randomness used by $\sfM$, and $s^*\in_R [k]$. 
Corollary \ref{cor:tcoinsfixed} follows from Theorem \ref{thm:tcoinsvar} similar to how (\cite{braverman2020coin}, Corollary 17) follows from Theorem \ref{thmbgw:tcoins}. The output of any such algorithm $\sfM$ would satisfy Equation \ref{eqorder:1} and let us assume $\sfM$ uses $<< \frac{t\log n}{k}$ memory. Then, with $k<\log n$ and $t<n^{c}$ for sufficiently small constant $c>0$, conditions on memory states for Theorem \ref{thm:tcoinsvar} are satisfied. Thus, it implies an $\Omega(n t^2 \log n)$ information lower bound on $\miccoin(\sfM)$, which in turn is upper bounded by $knt$ times the memory used by $\sfM$ (Lemma \ref{lem:upperboundofmic'}), giving a contradiction.

When the order $\{s_j\}_{j\in[nt]}$ is such that $s_j=((j-1)\mod t) +1$, then we refer to the above defined $\tcoins$ for this order as the $\simtcoins$. (\cite{braverman2020coin}, Remark 2) proved that a random order $\{s_j\in_R [t]\}_{j\in[nt]}$ is a good order with probability at least $1-\frac{1}{4n}$. Similar to (\cite{braverman2020coin}, Corollary 18), we can then prove the following corollary for the $\tcoins$ for a random order.

\begin{corollary}\label{cor:tcoinsran}
Suppose we are given a sequence of $nt$ $i.i.d.$ stream updates, the $j$-th of which has the form $(X_j, S_j)$, where $X_j$ is chosen uniformly in $\zo$ and $S_j$ is chosen uniformly in $[t]$. We interpret this as $t$ independent instances of the coin problem, where the $s$-th instance of the coin problem consists of the sequence of bits $X_{j_1}, \ldots, X_{j_v}$, if and only if $S_{j_u}=s, \forall 1\le u\le v$. Suppose there is a $k$-pass streaming algorithm $\sfM$ which goes over the stream of updates $k$ times in order, and given an $s^*\in_R [t]$ at the end of the stream, outputs the majority bit of the $s^*$-th instance of the coin problem with probability at least $1-\frac{1}{4000}$. Let $k<\log n$ and $t<n^{c}$ for a sufficiently small constant $c>0$. Then $\sfM$ uses $\Omega\left( \frac{t\log n}{k}\right)$ bits of memory. 
\end{corollary}

\subsection{Data Stream Applications}\label{sec:application}
We prove the theorems for the data stream problems defined in the introduction. For the first theorem, we use Corollary
\ref{cor:singlemaj}.

\begin{theorem}\label{thm:counterF}(MultiPass Counter in Strict Turnstile Streams) 
    Any $k$-pass strict turnstile streaming algorithm for maintaining a counter in a stream of length $n$ up to a sufficiently small constant factor with probability $1-\gamma$ for a sufficiently small constant $\gamma > 0$, requires $\Omega(\frac{\log n}{k})$ bits of memory.  
\end{theorem}
\begin{proof}
Let $S$ be the sum of $n$ coins, each in $\{-1,1\}$ corresponding to the input in Corollary \ref{cor:singlemaj}. The expected value of the maximum absolute value that $S$ takes, during the course of the $n$ stream updates, corresponds to the maximum deviation from the origin of an unbiased $1$-dimensional random walk and is $\Theta(\sqrt{n})$ (see, e.g., \cite{cffh98}), and consequently we have that by a Markov bound with probability at least $1-\gamma$ for an arbitrarily small constant $\gamma > 0$, that for a suitably large constant $C > 0$ depending on $\gamma$, that $|S| \leq C \sqrt{n}$ at all times during the stream. Hence, if we prepend $C \sqrt{n}$ coins, each with value $+1$, to the beginning of the stream, then the resulting counter $v$ satisfies that $v \geq 0$ at all times during the stream. Thus, $v$ is a counter in the strict turnstile model. Call this event $\mathcal{E}$. If event $\mathcal{E}$ holds, then the streaming algorithm has correctness probability at least $1-\gamma$, given that it is a strict turnstile streaming algorithm. 

We also have, by anti-concentration of the binomial distribution (e.g., the probability any particular value of a binomial sum is taken is at most $\frac{{n \choose n/2}}{2^n} = O \left (\frac{1}{\sqrt{n}} \right )$), that with probability at least $1-\gamma$ for an arbitrarily small constant $\gamma > 0$, that the final absolute value $|S|$ satisfies $|S| > C' \sqrt{n}$, where $C' > 0$ is a sufficiently small constant function of $\gamma$. As $S$ is symmetric, with probability $1-\gamma$ we have that if $S > 0$, then $v > (C+C') \sqrt{n}$, whereas if $S \leq 0$, then $v < (C-C') \sqrt{n}$. If event $\mathcal{E}$ also occurs, then if $S > 0$, then $|v| > (C+C') \sqrt{n}$, while if $S \leq 0$, then $|v| < (C-C')\sqrt{n}$. Call this event $\mathcal{F}$.

By a union bound, we have that events $\mathcal{E}$ and $\mathcal{F}$ and the event that the algorithm approximates $|v|$ up to a sufficiently small constant multiplicative factor, all jointly occur with probability at least $1-3\gamma$. Hence, with probability $1-3\gamma$, the output can be used to solve the the coin problem. It follows by Corollary \ref{cor:singlemaj} that for $\gamma > 0$ a sufficiently small constant, that the algorithm must use $\Omega(\frac{\log n}{k})$ bits of memory.  
\end{proof}

For the remaining theorems, we use Corollary \ref{cor:tcoinsfixed}.

\begin{theorem}[$O(1)$-Pass Multi-$\ell_p$-Estimation]\label{thm:multiLpEstimationF}
Suppose $k < \log n$, $t < n^c$ for a sufficiently small constant $c > 0$, and $p > 0$ is a constant. Then any $k$-pass streaming algorithm which solves the {\sf Multi-$\ell_p$-Estimation} Problem on $t$ instances of a stream of $n$ updates for each vector, solving each $\ell_p$-norm estimation problem up to a sufficiently small multiplicative factor and with probability $1-\gamma$ for a sufficiently small constant $\gamma > 0$, requires $\Omega(\frac{t \log n}{k})$ bits of memory. 
\end{theorem}
\begin{proof}
Consider $t$ $1$-dimensional vectors $v^1, \ldots, v^t \in \{-\textrm{poly}(n), \ldots, \textrm{poly}(n)\}$, so each $v^i$ is just a counter. Note that for any $p$-norm for $p > 0$, we have $\|v^i\|_p$ is just the absolute value $|v^i|$ of the counter $v^i$.

As in the proof of Theorem \ref{thm:counterF}, we initialize each $v^i$ to $C \sqrt{n}$ by adding $C \sqrt{n}$ updates of the form $v^i \leftarrow v^i + 1$ to the stream, where $C >0$ is a constant to be determined below. Next, we add to $v^i$ the sum $S^i$ of the coins in the $i$-th instance of the $\simtcoins$ in Corollary \ref{cor:tcoinsfixed}. 

As in the proof of Theorem \ref{thm:counterF}, the expected value of the maximum absolute value that $S^i$ takes, during the course of the $n$ stream updates, is $\Theta(\sqrt{n})$, and consequently for any particular vector $v^i$, we have that by a Markov bound with probability at least $1-\gamma$ for an arbitrarily small constant $\gamma > 0$, that for a suitably large constant $C > 0$ depending on $\gamma$, that $v^i \geq 0$ at all times during the stream. Call this event $\mathcal{E}_i$.

Also as in the proof of Theorem \ref{thm:counterF}, for any particular vector $S^i$, by anti-concentration of the binomial distribution, with probability at least $1-\gamma$ for an arbitrarily small constant $\gamma > 0$, the final absolute value $|S^i|$ of the sum of the $n$ stream updates satisfies $|S^i| > C' \sqrt{n}$, where $C' > 0$ is a sufficiently small constant function of $\gamma$. As $S^i$ is symmetric, with probability $1-\gamma$ we have that if $S^i > 0$, then $v^i > (C+C') \sqrt{n}$, whereas if $S^i \leq 0$, then $v^i < (C-C') \sqrt{n}$. If event $\mathcal{E}_i$ also occurs, then if $S^i > 0$, then $|v^i| > (C+C') \sqrt{n}$, while if $S^i \leq 0$, then $|v^i| < (C-C')\sqrt{n}$. Call this event $\mathcal{F}_i$.

By a union bound, for any particular $i$ we have that the events $\mathcal{E}_i$ and $\mathcal{F}_i$ and the event that the algorithm estimates $\|v^i\|_p$ to a sufficiently small constant multiplicative factor, all jointly occur with probability at least $1-3\gamma$. Hence, with probability $1-3\gamma$ the output can be used to solve the $i$-th instance of the $\simtcoins$. It follows by Corollary \ref{cor:tcoinsfixed} that the algorithm must use $\Omega(\frac{t \log n}{k})$  bits of memory. 
\end{proof}

\begin{theorem}[$O(1)$-Pass Point Query and Heavy Hitters] \label{thm:PQHHF}
Suppose $k < \log n$ and $\epsilon^{-2} < n^c$ for a sufficiently small constant $c > 0$.
Any $k$-pass streaming algorithm which, with probability at least $1-\gamma$ for a sufficiently small constant $\gamma > 0$, solves the {\sf $\ell_2$-Point Query Problem} or the {\sf $\ell_2$-Heavy Hitters Problem} on an underlying vector $x \in \{-\textrm{poly}(n), \ldots, \textrm{poly}(n)\}^d$ in a stream of $n$ updates,  requires $\Omega \left (\frac{\epsl^{-2} \log n}{k} \right )$ bits of memory. 
\end{theorem}
\begin{proof}
We create a $t = C''\epsilon^{-2}$-dimensional vector $v$, where each entry of $v$ is a counter corresponding to a single instance of the $\simtcoins$ in Corollary \ref{cor:tcoinsfixed}.
%, with the $t$ of that corollary set to $t = \frac{1}{\epsilon^2}$. 
Here $C'' > 0$ is a sufficiently small constant to be determined. 

By Remark 4 of \cite{braverman2020coin}, there is a constant $C > 0$ so that with probability at least $1-\epsilon$, we have $\|v\|_2^2 \leq C \cdot C'' \cdot n \epsilon^{-2}$ at the end of the stream. Call this event $\mathcal{E}$. 

Also as in the proof of Theorems \ref{thm:counterF} and \ref{thm:multiLpEstimationF}, for any particular coordinate $v^i$, by anti-concentration of the binomial distribution, with probability at least $1-\gamma$ for an arbitrarily small constant $\gamma > 0$, the final absolute value $|v^i|$ of the sum of the $n$ stream updates satisfies $|v^i| > C' \sqrt{n}$, where $C' > 0$ is a sufficiently small constant function of $\gamma$. As $v^i$ is symmetric, with probability $1-\gamma$,  we have that if $v^i > 0$, then $v^i > C'\sqrt{n}$, whereas if $v^i \leq 0$, then $v^i < -C' \sqrt{n}$. Call this event $\mathcal{F}_i$. 

For a particular $i$, if both $\mathcal{E}$ and $\mathcal{F}_i$ occur, then $|v^i|^2 > (C')^2 n$ while $\|v\|_2^2 < C C'' n \epsilon^{-2}$, and by setting $C'' = (C')^2/C$ we have that estimating $v^i$ up to an additive $\epsilon \|v\|_2$ factor can be used to determine if $v^i > 0$ or $v^i < 0$, and thus solves the $i$-th instance of the coin problem.  It follows by Corollary \ref{cor:tcoinsfixed}, that any algorithm for the {\sf $\ell_2$-Point Query Problem} must use $\Omega(\frac{\epsilon^{-2} \log n}{k})$  bits of memory. Here we assume $\epsilon$ is at most a sufficiently small constant so that event $\mathcal{E}$ occurs with probability at least $1-\gamma$, and we can also choose $\gamma > 0$ to be a sufficiently small constant by choosing $C' > 0$ to be sufficiently small. 

For the {\sf $\ell_2$-Heavy Hitters Problem}, notice that for any particular $i$, if both events $\mathcal{E}$ and $\mathcal{F}_i$ occur, which happens with probability at least $1-2\gamma$, then $(v^i)^2 \geq \epsl^2 \|v\|_2^2$, and so the algorithm for the {\sf $\ell_2$-Heavy Hitters Problem} must return index $i$ an  additive $\epsilon \|v\|_2$ approximation to its value. This approximation determines if $v^i > 0$ or $v^i < 0$, and thus solve the $i$-th instance of the coin problem.  It follows by Corollary \ref{cor:tcoinsfixed}, that any algorithm for the {\sf $\ell_2$-Heavy Hitters Problem} must use $\Omega(\frac{\epsilon^{-2} \log n}{k})$  bits of memory.
\end{proof}

\begin{theorem}\label{thm:compressedSensing}
(Bit Complexity of Sparse Recovery). 
Suppose $k < \log n$ and $r < n^c$ for a sufficiently small constant $c > 0$.
Any $k$-pass streaming algorithm which, with probability at least $1-\gamma$ for a sufficiently small constant $\gamma > 0$, solves the {\sf Sparse Recovery Problem} with
guarantee (\ref{eqn:recovery}) for a constant SNR value bounded away from $1$, requires
$\Omega(\frac{r \log n}{k})$ bits of memory. 
\end{theorem}
\begin{proof}
We create a $t = 3r$-dimensional vector $x$, where each of the first $r$ entries of $x$ is a counter corresponding to a single instance of the coin problem in Corollary \ref{cor:tcoinsfixed}, with the $t$ of that corollary set to $t = r$. 

By Remark 4 of \cite{braverman2020coin}, there is a constant $C > 0$ so that with probability at least $1-\frac{1}{\sqrt{r}}$, we have $\|x\|_2^2 \leq C \cdot r n + \sum_{j=r+1}^{3r} x_j^2$ at the end of the stream. Call this event $\mathcal{E}$. 

For any particular coordinate $x_i$, $1 \leq i \leq r$, by anti-concentration of the binomial distribution, with probability at least $1-\gamma$ for an arbitrarily small constant $\gamma > 0$, the final absolute value $|x_i|$ of the sum of the $n$ stream updates satisfies $|x_i| > C' \sqrt{n}$, where $C' > 0$ is a sufficiently small constant function of $\gamma$. As $x_i$ is symmetric, with probability $1-\gamma$ we have that if $x_i > 0$, then $x_i > C'\sqrt{n}$, whereas if $x_i \leq 0$, then $x_i < -C' \sqrt{n}$. Call this event $\mathcal{F}_i$. 

We set $x_{r+1} = x_{r+2} = \cdots = x_{3r} = \zeta C' \sqrt{n}$ for $\zeta > 0$ a sufficiently small constant, and so if event $\mathcal{E}$ holds, then the SNR is at most $1-\Omega(1)$. Indeed, in this case if $x_S$ is the vector obtained by zero-ing out all coordinates of $x$ outside of a subset $S$ of exactly $r$ coordinates, for any $S$ we have $\|x_S\|_2^2 \leq (1-\Omega(1))\|x\|_2^2$. We choose $3r$ coordinates of $x$ to ensure that if the first $r$ coordinates of $x$ are small and if the algorithm for the {\sf Sparse Recovery Problem} chooses a different set of $r$ coordinates, the SNR is still bounded away from $1$. 
Note that
\begin{eqnarray}\label{eqn:2}
\|x-x_r\|_2^2 \leq 2r \zeta^2 (C')^2 n, 
\end{eqnarray}
which follows since $\|x-x_r\|_2 \leq \|x-x'\|_2$, where $x'$ agrees with $x$ on its first $r$ coordinates. Note since the input is distributional, we can assume the streaming algorithm is deterministic by averaging. For any constant $\Delta \geq 1$, if we set $\zeta$ such that $\frac{\gamma}{\zeta^2} > \Delta$, then the output of the algorithm for the {\sf Sparse Recovery Problem} would satisfy the constant factor multiplicative approximation of (\ref{eqn:recovery}) only if 
\begin{eqnarray}\label{eqn:1}
\|x - \hat{x}\|_2^2 < \gamma r (C')^2 n, 
\end{eqnarray}
which happens with probability $1-\gamma$. Let's call this event $\calG$. Equation \eqref{eqn:1} implies that for all but $\gamma$ fraction of the coordinates in the first $r$ coordinates satisfy 
\begin{eqnarray}\label{eqn:individual}
\hat{x}_i \in [x_i - C' \sqrt{n}, x_i + C' \sqrt{n}].
\end{eqnarray}
Let $\calG_i$ be the event that Equation \ref{eqn:individual} holds. The probability that both $\mathcal{E}$ and $\calG$ holds is at least $1-2\gamma$. 
%, and thus we will take the randomness only over the choice of $i$. Assuming towards a contradiction that (\ref{eqn:individual}) does not hold implies that
For a particular $i$, if $\mathcal{F}_i$ and $\calG_i$ occurs, then $|x_i|^2 > (C')^2 n$ and hence $\hat{x}_i$ can compute the majority.  
%It follows that (\ref{eqn:individual}) holds for a $1-\gamma$ fraction of $i$ with $1 \leq i \leq r$. Say such an $i$ is {\it good}. 
Hence, for a uniformly random $i\in[r]$, the probability that either $\mathcal{E}$ does not hold, $\calG$ does not hold, $\mathcal{F}_i$ does not hold, or $\calG_i$ doesn't not hold, is at most $1/\sqrt{r} + 3\gamma$, which we can assume to be an arbitrarily small constant if $r$ is a sufficiently large constant. 
It follows by Corollary \ref{cor:tcoinsfixed} that any algorithm for the {\sf Sparse Recovery Problem} uses $\Omega(\frac{r \log n}{k})$  bits of memory. 

We note that we can assume $r$ is a sufficiently large constant, since otherwise we just seek an $\Omega \left ( \frac{\log n}{k} \right )$ bit lower bound. This can be obtained by instead reducing from Corollary \ref{cor:singlemaj}, and copying the single coin $r$ times. We again create a $3r$-dimensional vector. The only difference in the proof is that to bound $\|x\|_2^2$, we instead use the fact used in the proof of Theorem \ref{thm:counterF} that a single counter will be at most $C \sqrt{n}$ in absolute value with probability $1-\gamma$ if we choose $C > 0$ to be a sufficiently large constant. 
\end{proof}

\section{Multi-Pass Lower Bound for the Needle Problem}\label{sec:mainlemma}
In the needle problem, to evaluate the accuracy of a streaming algorithm $\calA$, we define its \textit{advantage} on distinguishing $ \boldsymbol{D_0}$ and $ \boldsymbol{D_1}$ by
\[
\Adv_{\calA}( \boldsymbol{D_0},  \boldsymbol{D_1}) := \left|\Pr[\calA( \boldsymbol{D_0})=1]-\Pr[\calA( \boldsymbol{D_1})=1]\right|. 
\]
This notion of advantage is widely used in cryptography and complexity. 
To simplify calculations, we also use  the following definition called \textit{error rate}: 
\[
Err_{\calA}( \boldsymbol{D_0},  \boldsymbol{D_1}) := 1 - \Adv_{\calA}( \boldsymbol{D_0},  \boldsymbol{D_1})= 1-\left|\Pr[\calA( \boldsymbol{D_0})=1]+\Pr[\calA( \boldsymbol{D_1})=0]-1\right|.
\]
Without loss of generality, we assume that $\Pr[\calA( \boldsymbol{D_0})=1]+\Pr[\calA( \boldsymbol{D_1})=0]\leq 1$. Otherwise, we can simply flip the output of $\calA$, which makes the advantage unchanged. In this case, we have that,
\[
Err_{\calA}( \boldsymbol{D_0},  \boldsymbol{D_1}) = \Pr[\calA( \boldsymbol{D_0})=1]+\Pr[\calA( \boldsymbol{D_1})=0].
\]

To simplify the proofs, we assume $\sfM$ does not do any operations during the last pass in this section. We note that the restriction does not affect our main theorem since for any given algorithm that distinguishes $\boldsymbol{D_0}$ and $\boldsymbol{D_1}$, we can construct another restricted algorithm that also distinguishes between $\boldsymbol{D_0}$ and $\boldsymbol{D_1}$. To be more specific, given any $k$-passes and $s$-space bounded streaming algorithm $\calA$ with $Err_{\calA}(\boldsymbol{D}_0,\boldsymbol{D}_1)\leq \epsilon$, we can easily construct an $(k+1)$-passes and $s$-space bounded streaming algorithm $\calA'$ with $Err_{\calA'}(\boldsymbol{D}_0,\boldsymbol{D}_1)\leq \epsilon$ as follows: 
\begin{enumerate}
    \item $\calA'$ simulates $\calA$ in the first $k$ passes.
    \item $\calA'$ does not do any operation during the $(k+1)$th pass. 
\end{enumerate}
It is easy to verify that $Err_{\calA'}(\boldsymbol{D}_0,\boldsymbol{D}_1) = Err_{\calA}(\boldsymbol{D}_0,\boldsymbol{D}_1)$, and if the $(k+1)sp^2n=\Omega(1)$ trade-off holds for $\calA'$, it also holds for $\calA$. 

In this section, we establish a lower bound of information complexity for algorithms with small errors.
\multilower*
Notice that we use $\mic(\calA)$ as an abbreviation for $\mic(\calA,\boldsymbol{D}_0)$ in this section. Then, theorem \ref{thm:streaming_lower_bound} follows directly by Lemma \ref{lem:Multi-passupperbound} and Lemma \ref{thm:lowerbound}: 
\[
2 k  sn\geq \mic(\calA)\geq \Omega(1/p^2).
\]

\subsection{Proof Sketch}\label{sec:intromain}
Since the proof of Lemma \ref{thm:lowerbound} is technical, we first give an informal sketch. Recall that in the needle problem, the algorithm $\calA$ aims to distinguish the following two distributions.

\begin{itemize}
    \item $\boldsymbol{D}_0$: Each of the $n$ samples is uniformly sampled from the domain $[t]$.
    \item $\boldsymbol{D_1}$: First uniformly sample a needle $\alpha \in [t]$. Then each element is sampled independently with probability $p$ equals $\alpha$, and otherwise uniformly sampled from $[t]$. 
\end{itemize}
\paragraph{Average Information Analysis: Decomposition-and-Reduction}Our novel approach (the average information analysis) has two steps: a decomposition and a reduction. Concretely, we first decompose the distribution $\boldsymbol{D_1}$ to a linear combination of many \textit{local needle distributions}. Then, we analysis the local information costs related to those local needle distributions by a reduction to the communication complexity problem called the  \textit{MostlyEq} problem, and we finally conclude our main lemma by taking the average on all those local information costs. This is different from previous proofs such as \cite{alon1996space,Andoni08,ToC16, lovett2023streaming}. Previous approaches were built on direct reductions to the unique disjointness problem. The decomposition is a crucial step allowing us to save the reduction loss.

\paragraph{\textsf{MostlyEq} Problem.}We first introduce the communication problem involved in the reduction step. After introducing this problem, the intuition of the decomposition step will also be more clear.
\begin{definition}[$m$-party \textsf{MostlyEq} problem]\label{def:mostlyeq}
    There are $m$ parties in the communication problem, where the $i$-th party holds an integer $z_i\in[t]$. We promise that $(z_1,\dots,z_m)$ are sampled from either of the following distributions : 
    \begin{enumerate}
        \item Uniform distribution (denoted by $\boldsymbol{P}_{U}$): each $z_i$ is sampled from $[t]$ independently and uniformly.  
        \item Mostly equal distribution (denoted by $\boldsymbol{P}_{Eq}$): first uniformly sample an element $\alpha$ (needle) from $[t]$. Then each $z_i$ independently with probability $1/2$ equals $\alpha$, and uniform otherwise. 
    \end{enumerate}
The goal of the players is to distinguish which case it is.
\end{definition}
\noindent
 The following lemma established the information complexity lower bound for \textsf{MostlyEq}.
\begin{restatable}[]{lemma}{MostlyEq}\label{lem:ICresultforCCproblem}
For any communication protocol $\Pi$ that solves the $m$-party, where $m\leq t/100$, \textsf{MostlyEq} problem with failure probability smaller than $0.1$, we have that,
\begin{align*}
\I\bigg(\Pi(\boldsymbol{P}_U);\boldsymbol{P}_U\bigg) =\Omega(1),
\end{align*}
In other words, the information complexity of $\Pi$ is $\Omega(1)$.
\end{restatable}
\noindent
Here, we define the failure probability of the protocol $\Pi$ as $$Err_{\Pi}(\boldsymbol{P}_U,\boldsymbol{P}_{Eq}):=\Pr[\Pi(\boldsymbol{P}_U)=1]+\Pr[\Pi(\boldsymbol{P}_{Eq})=0].$$
Notice that the \textsf{MostlyEq} problem is similar to the AND function. The main difference is that the inputs $z_1,\dots,z_m$ are integers in \textsf{MostlyEq}. The proof of this lemma builds on the information complexity lower bounds for \text{AND} function which we have introduced before. We defer the proof to Appendix \ref{sec:ICforMostlyEq}.

%\warn{[Qian: put ]}

Now, we describe our decomposition step for the needle distribution $\boldsymbol{D_1}$. Notice that the only difference between $\boldsymbol{D_1}$ and $\boldsymbol{D}_0$ is that $\boldsymbol{D_1}$ contains many needles in random positions. To explain our decomposition, we consider an alternative sampling process for $\boldsymbol{D_1}$.

\begin{framed}
\begin{enumerate}
    \item Sample a set $S\subseteq[n]$ with each element $j\in[n]$ contained in $S$ independently with probability $2p$. 
    \item Uniformly sample a needle $\alpha \in[t]$. 
    \item For each $j\notin S$, the $j$-th streaming sample is uniformly random. 
    \item For each $j\in S$, the $j$-th streaming sample equals to $\alpha$ with probability $1/2$ and uniformly random otherwise.
\end{enumerate}
\end{framed}
In this new sampling process for $\boldsymbol{D}_1$, we first sample a set $S$ containing candidate positions of needles. For each $S$, we use $\boldsymbol{D}^S$ to denote the distribution $(\boldsymbol{D_1}\mid_S)$, i.e., the needle distribution condition on the sampled set in step 1 equals $S$, and we call it the local needle distribution. Now we can decompose $\boldsymbol{D_1}$ as
\[
\boldsymbol{D_1} = \sum a_S\cdot \boldsymbol{D}^S,
\]
where $a_S = (2p)^{|S|}(1-2p)^{t-|S|}$. Here, the summation over distribution denotes the linear combination of distributions $\boldsymbol{D}^S$ with coefficients $a_S$. Furthermore, We can similarly decompose the error of an algorithm $\calA$, i.e,
\begin{align*}
Err_{\mathcal{A}}(\boldsymbol{D}_0, \boldsymbol{D_1}) &= \Pr[\mathcal{A}(\boldsymbol{D}_0)=1]+\Pr[\mathcal{A}(\boldsymbol{D_1})=0] \\
&= \sum_{S} a_{S}\cdot \Pr[\mathcal{A}(\boldsymbol{D}_0)=1]+\sum_{S} a_{S}\cdot \Pr[\mathcal{A}(\boldsymbol{D}^S)=0]\\
&=\sum_{S} a_{S}\cdot Err_{\calA}(\boldsymbol{D}_0,\boldsymbol{D}^S). 
\end{align*}

We then prove Lemma \ref{thm:lowerbound} by two steps. Let $\calA$ be an algorithm that distinguishes $\boldsymbol{D}_0$ and $\boldsymbol{D_1}$ with $Err_{\calA}(\boldsymbol{D}_0,\boldsymbol{D_1})\leq 0.002$,
\begin{itemize}
    \item We show that $Err_{\calA}(\boldsymbol{D}_0,\boldsymbol{D}^S)\leq 0.1$ for an $\Omega(1)$ fraction of sets $S$, and we call the set $S$ with $Err_{\calA}(\boldsymbol{D}_0,\boldsymbol{D}^S)\leq 0.1$ a \textit{good set}. 
    \item For each good set $S$, we show it contributes a good amount of information cost to $\mic(\calA)$. (Formalized by Lemma \ref{lem:sup} below). 
\end{itemize}

Step 1 follows by a simple average argument. Step 2 is more complicated, and we prove it by a reduction to \textsf{MostlyEq}. 

\begin{restatable}[]{lemma}{lemsup}
\label{lem:sup}
Given $S=\{p_1,\cdots,p_m\}\subseteq [n]$, if algorithm $\calA$ distinguishes $\boldsymbol{D}_{0}$ and $\boldsymbol{D}^S$ with error rate 
\[
    Err_{\calA}(\boldsymbol{D}_{0},\boldsymbol{D}^S)\leq 0.1,
\] 
then 
\[
\sum_{i=1}^{k-1}\sum_{j=1}^m\sum_{{\ell}=1}^j \I(\sfM_{(i,p_{j+1}-1)};X_{p_{\ell}}\mid \sfM_{ (\leq i,p_{{\ell}}-1)},\sfM_{( \leq i-1,p_{j+1}-1)})=\Omega(1),
\]
Specially, we define $p_{m+1}$ as $p_1$ here. 
\end{restatable}

\noindent For each set $S$, we denote the \textit{information cost contributed by $S$} as,
\[
\mic^{S}:=\sum_{i=1}^{k-1}\sum_{j=1}^m\sum_{{\ell}=1}^j \I(\sfM_{(i,p_{j+1}-1)};X_{p_{\ell}}\mid \sfM_{( \leq i,p_{\ell }-1)},\sfM_{ (\leq i-1,p_{j+1}-1)}).
\]

Intuitively, the information cost $\mic^{S}$ measures the mutual entropy of $X_{p_{\ell}}$ and $\sfM_{(i,p_{j+1}-1)}$. Here $X_{p_{\ell}}$ is a potential needle in $\boldsymbol{D}^S$ and $\sfM_{(i,p_{j+1}-1)}$ is the memory transcript before receiving another potential needle $p_{j+1}$. If $\calA$ distinguishes $\boldsymbol{D}^S$ from $\boldsymbol{D}_{0}$, then $\sfM_{(i,p_{j+1}-1)}$ must preserve some information about $p_1,\dots,p_{j}$, otherwise it would not catch that $p_{j+1}$ is a needle. We give a formal proof in Section \ref{Lastsubsec}. 

Now we are ready to prove Lemma \ref{thm:lowerbound} by an average information analysis. 

\subsection{Proof of Lemma \ref{thm:lowerbound}} \label{subSec:MAIN}

\begin{proof}[\textbf{Proof of Lemma \ref{thm:lowerbound}}]

Recall that $\boldsymbol{D_1} = \sum_S a_S \boldsymbol{D}^S$. We then have that,
\[
\mathbb{E}_S[\mic^S] = \sum_{S} a_S\cdot \mic^S\geq \sum_{S\text{ is good}}a_S \cdot \mic^S
\]
By Lemma \ref{lem:sup}, we know that $\mic^S=\Omega(1)$ for every good $S$, hence  
\[
\mathbb{E}_S[\mic^S] = \Omega\left(\Pr_{S}\left[S\text{ is good}\right]\right)
\]
Recall that a set $S$ is good if $Err_{\calA}(\boldsymbol{D}_0,\boldsymbol{D}^S)\leq 0.1$. Since $\calA$ has a small error on distinguishing $\boldsymbol{D}_0$ and $\boldsymbol{D_1}$, we have that,
\[
\mathbb{E}_{S}[{Err_{\calA}(\boldsymbol{D}_0,\boldsymbol{D}^S)}]=Err_{\calA}(\boldsymbol{D}_0,\boldsymbol{D_1})\leq 0.01.
\]
Then by a Markov's inequality, 
\[
\Pr_S[S\text{ is not good}]\leq \frac{\mathbb{E}[{Err_{\calA}(\boldsymbol{D}_0,\boldsymbol{D}^S)}]}{0.1} \leq 1/2,
\]
Therefore we show that $\Pr_S[S\text{ is good}] \geq 1/2$ and $\mathbb{E}_S[\mic^S] = \Omega(1)$ as a consequence.
We then build connections between $\mathbb{E}_S[ \mic^S]$ with $ \mic(\calA)$. Recall the definition,
\begin{align*}
\mic(\calA):=&\sum_{i=1}^{k}\sum_{j=1}^n\sum_{{\ell}=1}^j \mutualent(\sfM_{(i,j)};X_{{\ell}}\mid \sfM_{(\leq i,{\ell}-1)}, \sfM_{(\leq i-1,j)})\\ &+\sum_{i=1}^{k}\sum_{j=1}^n\sum_{{\ell}=j+1}^n \mutualent(\sfM_{(i,j)};X_{{\ell}}\mid \sfM_{(\leq i-1,{\ell}-1)}, \sfM_{(\leq i-1,j)}),
\end{align*}
For a random $S$, each term $\I(\sfM_{(i,j)};X_{\ell}|\sfM_{(\leq i,{\ell}-1)},\sfM_{(\leq i-1,j)})$ with $j\geq {\ell}$ appears in $ \mic^S$ with probability exactly $4p^2$ since this happens if only if both $j+1$ and ${\ell}$ are chosen by $S$. Similarly, each term $\I(\sfM_{(i,j)};X_{\ell}|\sfM_{(\leq i-1,{\ell}-1)},\sfM_{(\leq i-1,j)})$ with $j< {\ell}$ appears in $ \mic^S$ with probability at most $4p^2$ since this happens if and only if both $j+1$ and ${\ell}$ are chosen by $S$ and $j+1$ is the smallest element in $S$. Thus, we have 
\[
\mic(\calA)\geq \mathbb{E}_S[ \mic^S]/(4p^2) \geq \Omega(1/p^2)
\]
as desired. 
\end{proof}

\subsection{Proof of Lemma \ref{lem:sup} by a Reduction to \textsf{MostlyEq}}
\label{Lastsubsec}
We first restate this lemma: 
\lemsup*
\noindent
The proof is a combination of two steps.

\begin{enumerate}
    \item\label{step:1} We first prove the following bound via some information theoretical calculations.  
    $$\sum_{i=1}^{k-1}\sum_{j=1}^m\sum_{{\ell}=1}^j \I(\sfM_{(i,p_{j+1}-1)};X_{p_{\ell}}\mid \sfM_{ (\leq i,p_{{\ell}}-1)},\sfM_{ (\leq i-1,p_{j+1}-1)})\geq \I(\widetilde{\sfM};\widetilde{X}),$$
    here $\widetilde{\sfM}$ and $\widetilde{X}$ are defined by $\widetilde{\sfM}:=(\sfM_{(1,p_{1}-1)},\sfM_{(1,p_2-1)},\cdots,\sfM_{(2,p_{1}-1)},\cdots,\sfM_{(k-1,p_{m}-1)},\sfM_{(k,p_1-1)})$
    and $\widetilde{X}:=(X_{p_1},X_{p_2},\cdots,X_{p_m})$.
    \item \label{step:2} We then prove $I(\widetilde{\sfM};\widetilde{X}) \geq \Omega(1)$ by building a reduction to \textsf{MostlyEq}.    
\end{enumerate}
Step \ref{step:1} follows by the this claim:

\begin{restatable}[]{claim}{lemlmx}
\label{lem:lmx}
For any $k$-pass streaming algorithm $\calA$ running on independent inputs $X_1,\cdots,X_n$, and for a fixed set $S=\{p_1,p_2,\cdots,p_m\}\subseteq [n]$, we have that,
\[
\sum_{i=1}^{k-1}\sum_{j=1}^m \mutualent(\sfM_{(i,p_{j+1}-1)};X_{p_j}\mid \sfM_{(\leq i,p_j-1)}, \sfM_{(\leq i-1,p_{j+1}-1)}) \geq \I(\widetilde{\sfM};\widetilde{X}).
\]
\end{restatable}
\noindent 
The claim is out of a lot calculations, so we defer the proof to Appendix \ref{Apen:A}. 
With this claim, we know that:
\begin{align*}
&\sum_{i=1}^{k-1}\sum_{j=1}^m\sum_{{\ell}=1}^j \I(\sfM_{(i,p_{j+1}-1)};X_{p_{\ell}}\mid \sfM_{ (\leq i,p_{\ell}-1)},\sfM_{( \leq i-1,p_{j+1}-1)}) \\\geq& 
\sum_{i=1}^{k-1}\sum_{j=1}^m \mutualent(\sfM_{(i,p_{j+1}-1)};X_{p_j}\mid \sfM_{(\leq i,p_j-1)}, \sfM_{(\leq i-1,p_{j+1}-1)}) \geq\I(\widetilde{\sfM};\widetilde{X}).
\end{align*}
Now, we give the proof for Lemma \ref{lem:sup}, in which we show how the reduction to \textsf{MostlyEq} (Step \ref{step:2}) works.
\begin{proof}[\textbf{Proof of Lemma \ref{lem:sup}}]
As we discussed above, it suffices to prove:
\[
I(\widetilde{\sfM};\widetilde{X})\geq \Omega(1).
\]
The idea is to construct a communication protocol $\Pi$ (Algorithm \ref{algo}) that simulates the streaming algorithm $\mathcal{A}$. From the description of the protocol $\Pi$, we see it fully simulates the algorithm $\calA$. In the simulation, when the input to the protocol $\Pi$ follows $\boldsymbol{P}_U$, the input to $\calA$ follows $\boldsymbol{D}_0$; when the input to the protocol $\Pi$ follows $\boldsymbol{P}_{Eq}$, the input to $\calA$ then follows $\boldsymbol{D}^S$. Hence we have that,

   \begin{itemize}
       \item $Err_{\Pi}(\boldsymbol{P}_U,\boldsymbol{P}_{Eq}) = Err_{\mathcal{A}}(\boldsymbol{D}_0,\boldsymbol{D}^{S})$.
       \item $\I(\widetilde{\sfM};\widetilde{X}) = \I\big(\Pi(\boldsymbol{P}_U);\boldsymbol{P}_U\big).$
   \end{itemize}
Here, the two statements come directly from the simulating process and the restriction which indicates  the memory of $\calA$ retains the same at the last pass. 
We then conclude the proof by Lemma \ref{lem:ICresultforCCproblem}.

\begin{algorithm}[H]\label{algo}
        \caption{Communication Protocol For $\textsf{MostlyEq}$}
        \LinesNumbered
        \KwIn{$\boldsymbol{z}\in [t]^{m}$}
        \KwOut{$\text{ans}\in \{0,1\}$}
        \textbf{Recall:} $S=\{p_1,\dots,p_{m}$\} \;
         \For{player $j$ from $1$ to $m$}{
             let $X_{p_j}=z_j$\; 
             uniformly sample $X_{p_j+1},\cdots,X_{p_{j+1}-1}$  from $[t]$\;
            }
        Player $m$ simulates $\calA_{(1,p_1-1)} = \calA(X_1,\dots,X_{p_1-1})$ and sends $\calA_{(1,p_1-1)}$ to Player $1$\;
        \For{$i$ from $1$ to $k-1$}
            {
            \For{$j$ from $1$ to $m$}{
                Player $j$ simulates $\calA_{(i,p_{j+1}-1)} = \calA(\calA_{i,p_j-1},X_{p_j},\dots, X_{p_{j+1}-1})$\;
                Player $j$ sends  $\calA_{(i,p_{j+1}-1)}$ to Player $j+1$ (send to Player $1$ when $j=m$)\;
            }
        }
        \textbf{return the output of Player $m$};  
     \qedhere
    \end{algorithm}
\end{proof}

\section{Upper Bounds for the Needle Problem} \label{sec:upper_bound}

\noindent
We give the two following algorithms, giving tight upper bounds for the needle problem in a large range of parameters, and near-tight upper bounds for other ranges. Our algorithms significantly improve the algorithm for the needle problem implied by the algorithm of \cite{braverman2014optimal}. 
\begin{enumerate}
    \item In the one-way $n$-player communication game, we give an algorithm $\calA_1$ which solves the needle problem for $p\geq \frac{1}{\sqrt{n}}$ using  $O((\log \log n) (\log\log\log n))$ bits of space. 
    \item We propose a general one-pass streaming algorithm $\calA_2$ which works in both the $n$-player communication model as well as the standard streaming model, which solves the needle problem for any $p\leq \tiny{\frac{1}{\sqrt{n\log^{3} n}}}$ and using $O(\frac{1}{p^2n})$ bits of space. 
\end{enumerate}
Both algorithms distinguish between $\boldsymbol{D}_0$ and $\boldsymbol{D}_1$ with error probability less than an arbitrarily small constant. 

$\calA_1$ shows that our analysis and our lower bound results are tight up to a $\log \log n\cdot \log\log\log n$ factor for the communication game we consider throughout this paper, for the important case when $p=1/\sqrt{n}$. As mentioned in the introduction, this result gives a separation between the coin problem and the needle problem in the $n$-player one-way communication model, and also shows that the needle problem cannot be used to derive an $\Omega(\log n)$ bit lower bound for  $F_2$-estimation in the insertion-only model. We note that an $\Omega(\log n)$ lower bound does hold for $F_2$-estimation in the insertion-only streaming model based on a reduction from the Equality problem \cite{alon1996space}. This lower bound is not an information cost lower bound, and so cannot be used in direct sum arguments, and also does not hold if one does not charge the streaming algorithm to store randomness (i.e., in a the random oracle model). We note that since for $p = 1/\sqrt{n}$ we work in a communication model, each player ``knows its name", that is, this corresponds to a streaming model with a clock, so that given the $i$-th element in the stream, we know the number $i$. In such a model we obtain a true streaming algorithm rather than only a communication protocol. 

$\calA_2$ shows that our lower bound is optimal for the needle problem even in the streaming model of computation, whenever $p\leq \frac{1}{\sqrt{n\log^3 n}}$. Compared to previous upper bounds, this shows that our lower bound is tight in a larger range, improving upon the previous constraint that $p\leq n^{-2/3}$, as implied by the frequency estimation algorithm of \cite{braverman2014optimal}. One could instead use so called $\ell_2$-heavy hitter algorithms, the most efficient of which are given in \cite{BCIW16,BCINWW17}, but they would be suboptimal by a $\log n$ factor. 
%Given our streaming algorithms, remaining gap only exists when $p\in \bigg[\frac{1}{\sqrt{n}},\frac{1}{\sqrt{n\log^3 n}}\bigg]$. 

\subsection{Algorithm for $p\geq 1/\sqrt{n}$}\label{sec:algo1}
% We first roughly introduce rough the of this algorithm:
% \begin{enumerate}
%     \item divide the whole data stream into $\sqrt{t}$ continuous parts, we call them \textit{data groups}, and each contains $\sqrt{t}$ data;
%     \item use random hash functions $h_1$ to sample a constant number (in expectation) of elements from each data group as pivots, and set counters for them to track if they appear a lot of times in the future groups;
%     \item if a counter does not grow rapidly enough, we just throw it away; otherwise, a counter survives in a large number of rounds, we can conclude that the data stream is sampled from $\boldsymbol{D}_1$. 
% \end{enumerate}
We first design an algorithm $\calA_1$ for $p=1/\sqrt{n}$. For general $p\geq 1/\sqrt{n}$, we only consider the first $1/p^2$ data in the data stream, and apply the algorithm $\calA_1$ to this sub-stream with length $1/p^2$ and needle probability $p$. We then begin with the setup for $\calA_1$. We define the two following random hash functions: 
\begin{itemize}
    \item $h_1: \left[\sqrt{n}\right] \rightarrow {[t]\choose C_1 t/\sqrt{n}}$, 
    where ${[t]\choose C_1 t/\sqrt{n}}$ is defined by all subsets $S\subseteq [t]$ with $|S| = C_1 t/\sqrt{n}$. Here $C_1$ is a constant to be determined later.
    \item $h_2: [t]\rightarrow [C_2]$, where $C_2$ is another constant satisfying $C_2=100C_1$. 
\end{itemize}
Note that in the implementation of $\calA_1$, both $h_1,h_2$ are chosen in a random fashion, and are uniformly sampled from all valid functions. Before $\calA_1$ starts to process the data stream, it first divides the length-$n$ data stream into $\sqrt{n}$ continuous parts\footnote{$\calA_1$ is able to do this since it is in the communication model where players know their name, i.e., the $i$-th player knows the value $i$.} called \textit{data groups}, and each data group $i$ contains $\sqrt{n}$ contiguous items from the data stream. The basic memory unit for $\calA_1$ is called a \textit{counter}, which is a tuple $c=(c_1,c_2,c_3)$ consisting of three entries: (1) the \textit{hash value} $c_1$, (2) the \textit{lifespan} $c_2$, (3) the \textit{number of occurrences} $c_3$. We will further talk about the three entries in the formal definition of $\calA_1$. Then, $\calA_1$ proceeds as follows: 

% \begin{algorithm}[H]
%   \KwIn{A data stream with $t$ data: $(x_1,x_2,\cdots,x_t)$.}
%   \KwOut{$ans\in \{0,1\}$. }
%     $\text{cur} \leftarrow 1$;\\
%     for each $i\in [\frac{\log \log t}{\log\log\log t}+\log\log\log t],j\in [C]$, set $\text{ind}_i\leftarrow \text{empty},n_{i,j}\leftarrow 0$;\\
%   \For{$k$ from $1$ to $t$}{
%     \If{$x_k \in h_1(\text{cur})$}{
%         find a $i$ such that $\text{ind}_i=\text{cur}$; if such $i$ does not exists, find a $i$ with $\text{ind}_i=\text{empty}$;\\
%         let $\text{ind}_i = \text{cur},n_{i,h_2(x_k)} \leftarrow n_{i,h_2(x_k)}+1$;\\
%     }
%     \For{$i\in  [\log \log t/\log\log\log t+\log\log\log t]$ }{
%         \If{$\text{ind}_i \neq \text{empty}$ and $x_k \in h_1(\text{ind}_i)$}{
%             let $n_{i,h_2(x_k)} \leftarrow n_{i,h_2(x_k)}+1$;\\
%         }
%     }
%     toss a $(\frac{1}{\sqrt{t}},1-\frac{1}{\sqrt{t}})$ biased coin: $c$;\\
%     \If{$c=1$}{
%         \For{$i\in  [\log \log t/\log\log\log t+\log\log\log t]$}{
%             \If{$\text{ind}_i \neq \text{empty}$}{
%                 let $len = (\text{cur}-\text{ind}_i) (\mod \log \log t)+1$ ;\\
%                 set $n_{i,j}\leftarrow 0$ if $n_{i,j}\leq len/3$;\\
%                 set $\text{ind}_i \leftarrow \text{empty}$ if $\forall j, n_{i,j} = 0$;\\
%             }
%         }
%     }
%   }
%   \caption{$\calA_1${\color{red} to be furnished}}
% \label{algo:decompose}
% % \end{algorithm}
% \noindent
% We then briefly explain the algorithm above: 
\begin{enumerate}
    \item When $\calA_1$ enters a new data group $i$, it inserts $C_2$ counters into $\calA_1$'s memory (for each $j\in[C_2]$, it inserts a counter $(c_1 = j,c_2=0,c_3=0)$).
    \item When receiving a data $x_k$ in group $i$, $\calA_1$ enumerates every counter $(c_1,c_2,c_3)$ stored in memory, and checks: if $x_k$ is the first data element \footnote{To do this, we only need to add one bit for each counter to represent if this counter has been updated in the current group.} in group $i$ that (1) lies in $h_1(i-c_2)$; (2) and if it holds that $h_2(x_k)=c_1$. If both are satisfied, $\calA_1$ updates $c_3$ with $c_3+1$. 
    \item After scanning the whole data group $i$, $\calA_1$ enumerates every counter $(c_1,c_2,c_3)$ to see: if $c_3\geq c_2 /3$ or $c_2\leq 100$. If both are not satisfied, $\calA_1$ removes this counter from the memory. For the remaining counters, $\calA_1$ updates $c_2$ with $c_2+1$, and continues to process group $i+1$. 
    % \item If at some point, the number of active counters is bigger than $$C_2\cdot (5\log \log t/\log\log\log t+10\log\log\log t),$$ $\calA_1$ clears all active counters and continues to process the next group.
    \item If at some point, there exists a counter $(c_1,c_2,c_3)$ with $c_2\geq 10 \log \log n$, we begin to track this counter (which could be a potential needle) in the following $10 \log n$ groups without considering other counters. If $c_3\geq c_2 /3$ no longer holds when we are tracking, we clear all the counters and continue to process the next group; otherwise, we output $1$ after processing the $10 \log n$ groups for which $c_3\geq c_2 /3$ holds throughout the tracking. 
    \item If $\calA_1$ has processed all $\sqrt{n}$ groups without outputting $1$, it outputs $0$. 
\end{enumerate}
We first upper bound the space complexity of the algorithm: 
\begin{lemma}
    $\calA_1$ uses $O((\log \log n)( \log\log\log n))$ space. 
\end{lemma}
\begin{proof}
    The space used by $\calA_1$ comes from the following two parts: 
    \begin{enumerate}
        \item When not tracking potential needles, $\calA_1$ uses at most $O (\log \log n )$ counters since no counter with lifespan bigger than $10\log\log n$ exists by the definition of $\calA_1$. Also, each counter $(c_1,c_2,c_3)$ uses $O(\log \log \log t)$ space since $c_1=O(1),c_3\leq c_2 = O(\log \log n)$. The space complexity is $O((\log \log n)( \log\log\log n))$. 
        \item When tracking potential needles, $\calA_1$ needs only one counter $(c_1,c_2,c_3)$ using $O(\log \log n)$ space since $c_1=O(1),c_3\leq c_2=O(\log n)$. The space complexity is also $O(\log \log n)$. 
    \end{enumerate}
    Thus, we conclude that $O(\log \log n)$ bits of space is able to implement $\calA_1$.  
\end{proof}
\noindent
Next, we bound the error rate of $\calA_1$ under the two distributions: $\boldsymbol{D}_0, \boldsymbol{D}_1$. Formally, we prove the following two lemmas: 
\begin{lemma}
    $\Pr[\calA_1(\boldsymbol{D}_0)=1]\leq 1/n $. 
\end{lemma}
\begin{proof}
    This can be concluded via the following argument together with a union bound:  
    \begin{restatable}[]{claim}{expdecay}\label{claim:expdecay}
        Under $\boldsymbol{D}_0$, a counter of $\calA_1$ survives $r> 100$ rounds with probability at most $e^{-r/5}$. 
    \end{restatable}
    We define a round for $\calA_1$ as processing an entire data group. For the the formal proof of this claim, we refer the reader to Section \ref{sec:omittedupperboundproofs}. With this claim, we have that a counter survives $10 \log n$ rounds with probability at most $1/t^2$. By a union bound, there exists a counter surviving $10 \log n$ rounds with probability at most $C_2/n^{1.5} \leq 1/n$ since there are at most $C_2 \sqrt{n}$ counters and we assume $C_2$ is a constant smaller than $\sqrt{n}$ here. 
\end{proof}
\begin{lemma}\label{lem:A11}
    $\Pr[\calA_1(\boldsymbol{D}_1)=0]\leq  e^{-C_1} +0.05+ o(1)$. 
\end{lemma}
\begin{proof}
    We first define the two events: (1) we define $A$ as the event that $\calA_1$ fails to begin tracking the needle; and (2) we define event $B$ to be the event that the counter for the potential needle fails to survive $10 \log n$ rounds. Thus, we have $\Pr[\calA_1(\boldsymbol{D}_1)=0]\leq \Pr[A\cup B]\leq \Pr[A]+\Pr[B]$. 

    We first assume the needle equals $\alpha \in[t]$, and upper bound $\Pr[A]$. Event $A$ can be decomposed via the following two events: 
    \begin{itemize}
        \item event $A_1$: for any group $i$, $\alpha\notin h_1(i)$;
        \item event $A_2$: there exists a group $i$ such that $\alpha\in h_1(i)$, however, $\calA_1$ fails to track it. 
    \end{itemize} 
    Then, the following inequality comes directly from the definitions $$\Pr[A] = \Pr[A_1\cup A_2]\leq \Pr[A_1]+\Pr[A_2].$$ From a straightforward calculation, we have
    \begin{align*}
        \Pr[A_1] \leq \left(1-\frac{C_1}{\sqrt{n}}\right)^{\sqrt{n}} \leq e^{-C_1}. 
    \end{align*}

    To bound $\Pr[A_2]$, we consider another event $A_2'$: for the smallest $i$ such that $\alpha\in h_1(i)$, $\calA_1$ fails to track it. Clearly we have $A_2\subseteq A_2'$, and thus $\Pr[A_2]\leq \Pr[A_2']$. Then, it suffices to bound $\Pr[A_2']$. Here, we use $i'$ to denote this smallest index. Note that $\calA_1$ fails to track $i'$ in the following case: at some point $\calA_1$ processes data groups indexed in $[i'-10 \log n,i'+10\log\log n)$, there is a counter with lifespan larger than $10 \log \log n$ and $\calA_1$ begins to track this counter as the potential needle. We denote this event as $A_2'$. Then, we have $\Pr[A_2]\leq \Pr[A_2']$. 

    When $\calA_1$ is processing a data group in $[i'-10 \log n,i'+10 \log\log n)$, the probability that there exists a counter surviving $10\log \log n$ rounds is less than $$\frac{1}{1-1/e^{1/5}}e^{-10\log \log n/5}\leq \frac{1}{(1-1/e^{1/5})\log^2 n}.$$ This comes from a similar argument as Claim \ref{claim:expdecay} (the proof is also similar since the counters before $i'$ do not track the needle) together with a summation of a geometric series. Thus, by a union bound, we have that $$\Pr[A_2']\leq 10\log t\cdot \frac{1}{(1-1/e^{1/5})\log^2 n}\leq \frac{100}{\log n} =o(1).$$
    % To bound $p_2$, we divide the counters into two types by their lifespans:
    % \begin{enumerate}
    %     \item counters with lifespans at most $5 \log \log t/ \log \log \log t$;
    %     \item counters with lifespans larger than $5 \log \log t/ \log \log \log t$. 
    % \end{enumerate}
    % When processing a data group in $[i',i'+10\log \log t )$, the number of first type counters is at most $5 C_2  \log \log t/\log \log \log t $. If suffices to show that the number of second type counters is smaller than $10C_2 \log \log \log t $ with high probability. We first assume all counters $(c_1,c_2,c_3)$ satisfy $c_2\leq 10\log\log t$, otherwise it will contribute to the probability $p_1$. The probability that when processing a data group in $[i',i'+10\log \log t )$, there exists a counter with $c_2\geq 10\log\log t$ is less than $\frac{10\log \log t}{(1-1/e^{1/5})\log^2 t} =o(1/\log t)$. This comes from a union bound and a similar argument to Claim \ref{claim:expdecay}. Also, the probability that a counter survives more  than $5 \log \log t/ \log \log \log t$ rounds with probability at most $e^{-\log \log t/ \log \log \log t}$ and for counters of different groups, the probability is independent. Thus, the probability that at least $10 \log \log \log t$ exists is no more than $$2^{10 \log \log t}(e^{-\log \log t/ \log \log \log t})^{10 \log \log \log t}\leq o(1/\log t).$$ Thus, $p_2\leq o(1/\log t)+ \log \log t \cdot o(1/\log t) = o(1)$.
    Hence, it follows that 
    \[\Pr[A_2]\leq \Pr[A_2']\leq  o(1).\]

    It suffices to bound $\Pr[B]$.  Note that for every group $i$, the probability that at least one needle appears in $i$ is $1-(1-1/\sqrt{n})^{\sqrt{n}}\geq 1-1/e\geq 1/2$. Assuming the counter for the needle is $(c_1,c_2,c_3)$, we have that the expected value of $c_3$ after $r$ rounds, the corresponding random variable denoted by $c_3^r$, is larger than $$\mathbb{E}[c_3^r]\geq r/2.$$
    Together with Hoeffding's Inequality, we have that \[\Pr[c_3^r \leq r/3] \leq e^{-r/18}.\]
    Summing over $100 \leq r \leq 10\log n $, we have that 
    \[
    \Pr[B]\leq \sum_{r = 100}^{10\log n } \Pr[c_3^r \leq r/3] \leq 0.05. 
    \]
    We then concluded the lemma by: 
    \[
    \Pr[A]+\Pr[B]\leq e^{-C_1} + 0.05 +o(1). \qedhere
    \]
\end{proof}
\noindent
Combining the three lemmas above, we know $\calA_1$ satisfies that: 
\begin{itemize}
    \item $Err_{\calA_1}(\boldsymbol{D}_0,\boldsymbol{D}_1)\leq 0.1$ for a large enough constant $C_1$;
    \item $\calA_1$ uses $O(\log \log n)$ bits of space for any constant $C_1$. 
\end{itemize}
This concludes the $O((\log \log n)(\log\log\log n))$ space complexity upper bound for the needle problem when $p=1/\sqrt{n}$. 

\paragraph{Extension to hybrid order streams.} Consider a setting of hybrid order data streams where the needle appears $\approx \sqrt{n}$ times in a random order while other items are in an arbitrary order. Also, we have the following constraints on the data stream: 1) the $F_2$-frequency moment of non-needle items is $O(n)$; 2) non-needle items have a number of occurrences at most a small constant times $\log n$. Then, our algorithm $\sfM_1$ also works on this hybrid order streams by similar analysis. However, since this setting is not so related to the topic of this paper, we omit the details. 
\subsection{Streaming Algorithm for $p\leq 1/\sqrt{n\log^3 n}$}\label{sec:algo2}
The idea of $\calA_2$ is similar to $\calA_1$. Again, we start by setting up $\calA_2$: 
\begin{itemize}
    \item First, partition the whole domain $t$ into $\frac{1}{p^2n}$ \textit{blocks}. The $\ell$-th block contains elements in $[(\ell-1)p^2nt+1,\ell p^2nt]$, where we use $\mathcal{B}^{\ell}$ to denote the domain of the $\ell$-th block. 
    \item $\calA_2$ is able to partition the entire data stream into $pn$ number of \textit{data groups} via $pn$ independent Poisson distributions: each data group $i$ contains $N_i \sim P(\frac{1}{4p})$ contiguous data items in the data stream. Here, $P(\frac{1}{4p})$ denotes the Poisson distribution with parameter $\frac{1}{4p}$. Note that the total size $N=\sum_{i=1}^{pn} N_i$ also follows a Poisson distribution $N\sim P(\frac{n}{4})$, and may exceed $n$. In that case, the algorithm $\calA_2$ would fail. However, we will show that the probability of that case is small in later calculations.
    \item $\calA_2$ uses a register $index$ to store the index of the current data group, $index$ uses $O(\log n)$ bits of space since the number of data groups is at most $pn = O(n)$. 
\end{itemize}
We then design $\calA_2$ as a composed algorithm, running $\frac{1}{p^2n}$ independent algorithms $\calA^1,\calA^2,\cdots,\calA^{\frac{1}{p^2n}}$ at the same time, each $\calA^{\ell}$ works similarly to $\calA_1$, and checks if the needle is in the $\ell$-th block. Then, $\calA_2$ outputs $1$ if and only if at least one $\calA^{\ell}$ outputs $1$, and outputs $0$ if and only if all $\calA^{\ell}$ output $0$. Precisely, $\calA^{\ell}$ uses two random hash functions: (1) $h_1^{\ell}: [pn] \rightarrow {{\mathcal{B}^{\ell}}\choose{C_1tp}}$, where ${{\mathcal{B}^{\ell}}\choose{C_1tp}}$ is defined by all subsets $S\subseteq \mathcal{B}^{\ell }$ with $|S|= C_1tp$, and (2) $h_2^{\ell}: \mathcal{B}^{\ell}\rightarrow [C_2]$. The basic memory unit for $\calA^{\ell}$ is also the \textit{counter} $c=(c_1,c_2,c_3)$. $C_1$ is a constant to be determined later, and $C_2$ equals $1000C_1$. Next, we formally define  $\calA^{\ell }$: 
\begin{enumerate}
    \item When $\calA^{\ell}$ enters a new data group $i$, it inserts $C_2$ counters into its memory in the following way: for each $j\in[C_2]$, it inserts a counter $(c_1 = j,c_2=0,c_3=0)$.
    \item When receiving a data element $x_k$ in the data group $i$, first check if $x_k\in \mathcal{B}^{\ell}$. If not, $\calA^{\ell}$ skips the current data $x_k$. Otherwise, $\calA^{\ell}$ enumerates every counter $(c_1,c_2,c_3)$  in its memory, and checks if $x_k$ is the first data element \footnote{Similar to $\calA_1$, we can add one bit to each counter in order to do this.} in data group $i$ that: (1) lies in $h_1^{\ell}(i-c_2)$; (2) and it holds that $h_2^{\ell}(x_k)=c_1$. If both are satisfied, $\calA^{\ell}$ updates $c_3$ with $c_3+1$. 
    \item After scanning the entire data group $i$, $\calA^{\ell}$ enumerates every counter $(c_1,c_2,c_3)$ to see if $c_3\geq  c_2/100 $ or $c_2\leq 10000$ holds. If neither holds, $\calA^{\ell}$ removes this counter from its memory. For the remaining counters, $\calA_1$ updates $c_2$ with $c_2+1$, and continues to process the next data group indexed $i+1$. 
    \item If at some point, there exists a counter $(c_1,c_2,c_3)$ with lifespan $c_2\geq 3\cdot 10^6 \log n$, $\calA^{\ell}$ stops to process the remaining data and outputs $1$. 
    \item If $\calA^{\ell }$ has processed all $pt$ data groups without outputting $1$, then $\calA^{\ell}$ outputs $0$. 
\end{enumerate}

The analysis of the composed algorithms $\calA_2$ follows a similar idea to $\calA_1$'s analysis, and we directly give our result and defer the proofs to Appendix \ref{sec:omittedupperboundproofs}: 
\begin{restatable}[]{theorem}{generalp}\label{thm:general_p}
    When $p\leq \frac{1}{\sqrt{n \log^3 n}}$, $\calA_2$ uses at most $O(\frac{C'}{p^2n})$ bits of memory, where $C'$ is a constant, with probability $1-o(1)$ and solves the needle problem with $Err_{\calA_2}(\boldsymbol{D}_0,\boldsymbol{D}_1)\leq 0.2$.
\end{restatable}
Note that although the $\calA_2$ may exceed our memory constraint with small probability, we can easily transform it into a bounded memory algorithm with a small error probability by simulating $\calA_2$ and aborting if the memory exceeds the constraint.

\normalem
\bibliographystyle{alpha}
\bibliography{reference}

\newcommand{\etalchar}[1]{$^{#1}$}
\begin{thebibliography}{CMVW16}

\bibitem[ABJ{\etalchar{+}}22]{ABJSSWZ22}
Mikl{\'{o}}s Ajtai, Vladimir Braverman, T.~S. Jayram, Sandeep Silwal, Alec Sun, David~P. Woodruff, and Samson Zhou.
\newblock The white-box adversarial data stream model.
\newblock In Leonid Libkin and Pablo Barcel{\'{o}}, editors, {\em {PODS} '22: International Conference on Management of Data, Philadelphia, PA, USA, June 12 - 17, 2022}, pages 15--27. {ACM}, 2022.

\bibitem[AHLW16]{AHL16}
Yuqing Ai, Wei Hu, Yi~Li, and David~P. Woodruff.
\newblock New characterizations in turnstile streams with applications.
\newblock In Ran Raz, editor, {\em 31st Conference on Computational Complexity, {CCC} 2016, May 29 to June 1, 2016, Tokyo, Japan}, volume~50 of {\em LIPIcs}, pages 20:1--20:22. Schloss Dagstuhl - Leibniz-Zentrum f{\"{u}}r Informatik, 2016.

\bibitem[AHNY22]{ahny22}
Ishaq Aden{-}Ali, Yanjun Han, Jelani Nelson, and Huacheng Yu.
\newblock On the amortized complexity of approximate counting.
\newblock {\em CoRR}, abs/2211.03917, 2022.

\bibitem[AKO11]{andoni2011streaming}
Alexandr Andoni, Robert Krauthgamer, and Krzysztof Onak.
\newblock Streaming algorithms via precision sampling.
\newblock In {\em 2011 IEEE 52nd Annual Symposium on Foundations of Computer Science}, pages 363--372. IEEE, 2011.

\bibitem[AKSY20]{mps1}
Sepehr Assadi, Gillat Kol, Raghuvansh~R Saxena, and Huacheng Yu.
\newblock Multi-pass graph streaming lower bounds for cycle counting, max-cut, matching size, and other problems.
\newblock In {\em 2020 IEEE 61st Annual Symposium on Foundations of Computer Science (FOCS)}, pages 354--364. IEEE, 2020.

\bibitem[AMOP08]{Andoni08}
Alexandr Andoni, Andrew McGregor, Krzysztof Onak, and Rina Panigrahy.
\newblock Better bounds for frequency moments in random-order streams.
\newblock {\em CoRR}, abs/0808.2222, 2008.

\bibitem[AMS99]{alon1996space}
Noga Alon, Yossi Matias, and Mario Szegedy.
\newblock The space complexity of approximating the frequency moments.
\newblock {\em Journal of Computer and System Sciences}, 58(1):137--147, 1999.

\bibitem[BBS22]{Colt22}
Gavin Brown, Mark Bun, and Adam Smith.
\newblock Strong memory lower bounds for learning natural models.
\newblock In {\em Conference on Learning Theory}, pages 4989--5029. PMLR, 2022.

\bibitem[BCI{\etalchar{+}}17]{BCINWW17}
Vladimir Braverman, Stephen~R. Chestnut, Nikita Ivkin, Jelani Nelson, Zhengyu Wang, and David~P. Woodruff.
\newblock Bptree: An {\(\mathscr{l}\)}\({}_{\mbox{2}}\) heavy hitters algorithm using constant memory.
\newblock In Emanuel Sallinger, Jan~Van den Bussche, and Floris Geerts, editors, {\em Proceedings of the 36th {ACM} {SIGMOD-SIGACT-SIGAI} Symposium on Principles of Database Systems, {PODS} 2017, Chicago, IL, USA, May 14-19, 2017}, pages 361--376. {ACM}, 2017.

\bibitem[BCIW16]{BCIW16}
Vladimir Braverman, Stephen~R. Chestnut, Nikita Ivkin, and David~P. Woodruff.
\newblock Beating countsketch for heavy hitters in insertion streams.
\newblock In Daniel Wichs and Yishay Mansour, editors, {\em Proceedings of the 48th Annual {ACM} {SIGACT} Symposium on Theory of Computing, {STOC} 2016, Cambridge, MA, USA, June 18-21, 2016}, pages 740--753. {ACM}, 2016.

\bibitem[Ber24]{bernstein1}
Sergei Bernstein.
\newblock On a modification of chebyshev’s inequality and of the error formula of laplace.
\newblock {\em Ann. Sci. Inst. Sav. Ukraine, Sect. Math}, 1(4):38--49, 1924.

\bibitem[Ber37]{bernstein2}
Sergei~N Bernstein.
\newblock On certain modifications of chebyshev’s inequality.
\newblock {\em Doklady Akademii Nauk SSSR}, 17(6):275--277, 1937.

\bibitem[BGKS06]{bhuvanagiri2006simpler}
Lakshminath Bhuvanagiri, Sumit Ganguly, Deepanjan Kesh, and Chandan Saha.
\newblock Simpler algorithm for estimating frequency moments of data streams.
\newblock In {\em Proceedings of the seventeenth annual ACM-SIAM symposium on Discrete algorithm}, pages 708--713, 2006.

\bibitem[BGW20]{braverman2020coin}
Mark Braverman, Sumegha Garg, and David~P Woodruff.
\newblock The coin problem with applications to data streams.
\newblock In {\em 2020 IEEE 61st Annual Symposium on Foundations of Computer Science}, pages 318--329. IEEE, 2020.

\bibitem[BGZ21]{BGZ21}
Mark Braverman, Sumegha Garg, and Or~Zamir.
\newblock Tight space complexity of the coin problem.
\newblock In {\em 62nd {IEEE} Annual Symposium on Foundations of Computer Science, {FOCS} 2021, Denver, CO, USA, February 7-10, 2022}, pages 1068--1079. {IEEE}, 2021.

\bibitem[BJKS04]{BJKS04}
Ziv Bar{-}Yossef, T.~S. Jayram, Ravi Kumar, and D.~Sivakumar.
\newblock An information statistics approach to data stream and communication complexity.
\newblock {\em J. Comput. Syst. Sci.}, 68(4):702--732, 2004.

\bibitem[BKSV14]{braverman2014optimal}
Vladimir Braverman, Jonathan Katzman, Charles Seidell, and Gregory Vorsanger.
\newblock An optimal algorithm for large frequency moments using $\mathcal{O}(n^{1-2/k})$ bits.
\newblock In {\em Approximation, Randomization, and Combinatorial Optimization. Algorithms and Techniques (APPROX/RANDOM 2014)}. Schloss Dagstuhl-Leibniz-Zentrum fuer Informatik, 2014.

\bibitem[CCM08]{ToC16}
Amit Chakrabarti, Graham Cormode, and Andrew McGregor.
\newblock Robust lower bounds for communication and stream computation.
\newblock In {\em Proceedings of the Fortieth Annual ACM Symposium on Theory of Computing}, STOC '08, page 641–650, New York, NY, USA, 2008. Association for Computing Machinery.

\bibitem[CFFH98]{cffh98}
EG~Coffman, Philippe Flajolet, Leopold Flatto, and Micha Hofri.
\newblock The maximum of a random walk and its application to rectangle packing.
\newblock {\em Probability in the Engineering and Informational Sciences}, 12(3):373--386, 1998.

\bibitem[Che52]{chernoff}
Herman Chernoff.
\newblock A measure of asymptotic efficiency for tests of a hypothesis based on the sum of observations.
\newblock {\em The Annals of Mathematical Statistics}, pages 493--507, 1952.

\bibitem[CJP08]{CJP08}
Amit Chakrabarti, T.~S. Jayram, and Mihai Pundefinedtra\c{s}cu.
\newblock Tight lower bounds for selection in randomly ordered streams.
\newblock In {\em Proceedings of the Nineteenth Annual ACM-SIAM Symposium on Discrete Algorithms}, SODA '08, page 720–729, USA, 2008. Society for Industrial and Applied Mathematics.

\bibitem[CKS03]{ccc03}
A.~Chakrabarti, S.~Khot, and Xiaodong Sun.
\newblock Near-optimal lower bounds on the multi-party communication complexity of set disjointness.
\newblock In {\em 18th IEEE Annual Conference on Computational Complexity, 2003. Proceedings.}, pages 107--117, 2003.

\bibitem[CMVW16]{DBLP:conf/esa/CrouchMVW16}
Michael~S. Crouch, Andrew McGregor, Gregory Valiant, and David~P. Woodruff.
\newblock Stochastic streams: Sample complexity vs. space complexity.
\newblock In Piotr Sankowski and Christos~D. Zaroliagis, editors, {\em 24th Annual European Symposium on Algorithms, {ESA} 2016, August 22-24, 2016, Aarhus, Denmark}, volume~57 of {\em LIPIcs}, pages 32:1--32:15. Schloss Dagstuhl - Leibniz-Zentrum f{\"{u}}r Informatik, 2016.

\bibitem[DDKS16]{dinur2016memory}
Itai Dinur, Orr Dunkelman, Nathan Keller, and Adi Shamir.
\newblock Memory-efficient algorithms for finding needles in haystacks.
\newblock In {\em Advances in Cryptology--CRYPTO 2016: 36th Annual International Cryptology Conference, Santa Barbara, CA, USA, August 14-18, 2016, Proceedings, Part II}, pages 185--206. Springer, 2016.

\bibitem[DGKR19]{diakonikolas2019communication}
Ilias Diakonikolas, Themis Gouleakis, Daniel~M Kane, and Sankeerth Rao.
\newblock Communication and memory efficient testing of discrete distributions.
\newblock In {\em Conference on Learning Theory}, pages 1070--1106. PMLR, 2019.

\bibitem[Din20]{euro2019}
Itai Dinur.
\newblock On the streaming indistinguishability of a random permutation and a random function.
\newblock In Anne Canteaut and Yuval Ishai, editors, {\em Advances in Cryptology -- EUROCRYPT 2020}, pages 433--460, Cham, 2020. Springer International Publishing.

\bibitem[FHM{\etalchar{+}}20]{SODA2020randomordermatching}
Alireza Farhadi, MohammadTaghi Hajiaghayi, Tung Mai, Anup Rao, and Ryan~A. Rossi.
\newblock Approximate maximum matching in random streams.
\newblock In {\em Proceedings of the Thirty-First Annual ACM-SIAM Symposium on Discrete Algorithms}, SODA '20, page 1773–1785, USA, 2020. Society for Industrial and Applied Mathematics.

\bibitem[Fla85]{F}
Philippe Flajolet.
\newblock Approximate counting: {A} detailed analysis.
\newblock {\em {BIT}}, 25(1):113--134, 1985.

\bibitem[Gan11]{Gan11}
Sumit Ganguly.
\newblock Polynomial estimators for high frequency moments.
\newblock {\em arXiv preprint arXiv:1104.4552}, 2011.

\bibitem[GH09]{gh09}
Sudipto Guha and Zhiyi Huang.
\newblock Revisiting the direct sum theorem and space lower bounds in random order streams.
\newblock In Susanne Albers, Alberto Marchetti-Spaccamela, Yossi Matias, Sotiris Nikoletseas, and Wolfgang Thomas, editors, {\em Automata, Languages and Programming}, pages 513--524, Berlin, Heidelberg, 2009. Springer Berlin Heidelberg.

\bibitem[GLPS17]{GLPS17}
Anna~C. Gilbert, Yi~Li, Ely Porat, and Martin~J. Strauss.
\newblock For-all sparse recovery in near-optimal time.
\newblock {\em {ACM} Trans. Algorithms}, 13(3):32:1--32:26, 2017.

\bibitem[GM07]{GM07}
Sudipto Guha and Andrew McGregor.
\newblock Space-efficient sampling.
\newblock In {\em Artificial Intelligence and Statistics}, pages 171--178. PMLR, 2007.

\bibitem[GM09]{GM09}
Sudipto Guha and Andrew McGregor.
\newblock Stream order and order statistics: Quantile estimation in random-order streams.
\newblock {\em SIAM Journal on Computing}, 38(5):2044--2059, 2009.

\bibitem[Gro09]{gronemeier2009asymptotically}
André Gronemeier.
\newblock Asymptotically optimal lower bounds on the nih-multi-party information, 2009.

\bibitem[Gro10]{G10}
Andr{\'{e}} Gronemeier.
\newblock {\em Information Complexity and Data Stream Algorithms for Basic Problems}.
\newblock PhD thesis, Technical University Dortmund, Germany, 2010.

\bibitem[GW18]{ganguly_et_al:LIPIcs:2018:9062}
Sumit Ganguly and David~P. Woodruff.
\newblock {High Probability Frequency Moment Sketches}.
\newblock In Ioannis Chatzigiannakis, Christos Kaklamanis, D{\'a}niel Marx, and Donald Sannella, editors, {\em 45th International Colloquium on Automata, Languages, and Programming (ICALP 2018)}, volume 107 of {\em Leibniz International Proceedings in Informatics (LIPIcs)}, pages 58:1--58:15, Dagstuhl, Germany, 2018. Schloss Dagstuhl--Leibniz-Zentrum fuer Informatik.

\bibitem[IPW11]{IPW11}
Piotr Indyk, Eric Price, and David~P. Woodruff.
\newblock On the power of adaptivity in sparse recovery.
\newblock In Rafail Ostrovsky, editor, {\em {IEEE} 52nd Annual Symposium on Foundations of Computer Science, {FOCS} 2011, Palm Springs, CA, USA, October 22-25, 2011}, pages 285--294. {IEEE} Computer Society, 2011.

\bibitem[IW05]{IW05}
Piotr Indyk and David Woodruff.
\newblock Optimal approximations of the frequency moments of data streams.
\newblock In {\em Proceedings of the Thirty-Seventh Annual ACM Symposium on Theory of Computing}, STOC '05, page 202–208, New York, NY, USA, 2005. Association for Computing Machinery.

\bibitem[Jay09]{Jay09}
T.~S. Jayram.
\newblock Hellinger strikes back: A note on the multi-party information complexity of and.
\newblock In Irit Dinur, Klaus Jansen, Joseph Naor, and Jos{\'e} Rolim, editors, {\em Approximation, Randomization, and Combinatorial Optimization. Algorithms and Techniques}, pages 562--573, Berlin, Heidelberg, 2009. Springer Berlin Heidelberg.

\bibitem[JT19]{DBLP:conf/eurocrypt/JaegerT19}
Joseph Jaeger and Stefano Tessaro.
\newblock Tight time-memory trade-offs for symmetric encryption.
\newblock In Yuval Ishai and Vincent Rijmen, editors, {\em Advances in Cryptology - {EUROCRYPT} 2019 - 38th Annual International Conference on the Theory and Applications of Cryptographic Techniques, Darmstadt, Germany, May 19-23, 2019, Proceedings, Part {I}}, volume 11476 of {\em Lecture Notes in Computer Science}, pages 467--497. Springer, 2019.

\bibitem[JW18]{JW18}
Rajesh Jayaram and David~P. Woodruff.
\newblock Data streams with bounded deletions.
\newblock In Jan~Van den Bussche and Marcelo Arenas, editors, {\em Proceedings of the 37th {ACM} {SIGMOD-SIGACT-SIGAI} Symposium on Principles of Database Systems, Houston, TX, USA, June 10-15, 2018}, pages 341--354. {ACM}, 2018.

\bibitem[JW23]{jw23}
Rajesh Jayaram and David~P. Woodruff.
\newblock Towards optimal moment estimation in streaming and distributed models.
\newblock {\em {ACM} Trans. Algorithms}, 19(3):27:1--27:35, 2023.

\bibitem[KNW10]{KNW10}
Daniel~M. Kane, Jelani Nelson, and David~P. Woodruff.
\newblock On the exact space complexity of sketching and streaming small norms.
\newblock In Moses Charikar, editor, {\em Proceedings of the Twenty-First Annual {ACM-SIAM} Symposium on Discrete Algorithms, {SODA} 2010, Austin, Texas, USA, January 17-19, 2010}, pages 1161--1178. {SIAM}, 2010.

\bibitem[KPW21]{KPW21}
Akshay Kamath, Eric Price, and David~P. Woodruff.
\newblock A simple proof of a new set disjointness with applications to data streams.
\newblock In {\em Proceedings of the 36th Computational Complexity Conference}, CCC '21, Dagstuhl, DEU, 2021. Schloss Dagstuhl--Leibniz-Zentrum fuer Informatik.

\bibitem[LZ23]{lovett2023streaming}
Shachar Lovett and Jiapeng Zhang.
\newblock Streaming lower bounds and asymmetric set-disjointness.
\newblock In {\em 2023 IEEE 64th Annual Symposium on Foundations of Computer Science (FOCS)}, pages 871--882, Los Alamitos, CA, USA, nov 2023. IEEE Computer Society.

\bibitem[MNSW95]{miltersen1995data}
Peter~Bro Miltersen, Noam Nisan, Shmuel Safra, and Avi Wigderson.
\newblock On data structures and asymmetric communication complexity.
\newblock In {\em Proceedings of the twenty-seventh annual ACM symposium on Theory of computing}, pages 103--111, 1995.

\bibitem[MPTW12]{MPTW12}
Andrew McGregor, A.~Pavan, Srikanta Tirthapura, and David Woodruff.
\newblock Space-efficient estimation of statistics over sub-sampled streams.
\newblock In {\em Proceedings of the 31st ACM SIGMOD-SIGACT-SIGAI Symposium on Principles of Database Systems}, PODS '12, page 273–282, New York, NY, USA, 2012. Association for Computing Machinery.

\bibitem[MPTW16]{DBLP:journals/algorithmica/McGregorPTW16}
Andrew McGregor, A.~Pavan, Srikanta Tirthapura, and David~P. Woodruff.
\newblock Space-efficient estimation of statistics over sub-sampled streams.
\newblock {\em Algorithmica}, 74(2):787--811, 2016.

\bibitem[MW10]{MW10}
Morteza Monemizadeh and David~P. Woodruff.
\newblock 1-pass relative-error lp-sampling with applications.
\newblock In {\em Proceedings of the Twenty-First Annual ACM-SIAM Symposium on Discrete Algorithms}, SODA '10, page 1143–1160, USA, 2010. Society for Industrial and Applied Mathematics.

\bibitem[MWY13]{MWY13}
Marco Molinaro, David~P. Woodruff, and Grigory Yaroslavtsev.
\newblock Beating the direct sum theorem in communication complexity with implications for sketching.
\newblock In Sanjeev Khanna, editor, {\em Proceedings of the Twenty-Fourth Annual {ACM-SIAM} Symposium on Discrete Algorithms, {SODA} 2013, New Orleans, Louisiana, USA, January 6-8, 2013}, pages 1738--1756. {SIAM}, 2013.

\bibitem[NSWZ18]{NSW018}
Vasileios Nakos, Xiaofei Shi, David~P. Woodruff, and Hongyang Zhang.
\newblock Improved algorithms for adaptive compressed sensing.
\newblock In Ioannis Chatzigiannakis, Christos Kaklamanis, D{\'{a}}niel Marx, and Donald Sannella, editors, {\em 45th International Colloquium on Automata, Languages, and Programming, {ICALP} 2018, July 9-13, 2018, Prague, Czech Republic}, volume 107 of {\em LIPIcs}, pages 90:1--90:14. Schloss Dagstuhl - Leibniz-Zentrum f{\"{u}}r Informatik, 2018.

\bibitem[NY22]{NY22}
Jelani Nelson and Huacheng Yu.
\newblock Optimal bounds for approximate counting.
\newblock In Leonid Libkin and Pablo Barcel{\'{o}}, editors, {\em {PODS} '22: International Conference on Management of Data, Philadelphia, PA, USA, June 12 - 17, 2022}, pages 119--127. {ACM}, 2022.

\bibitem[PW11]{PW11}
Eric Price and David~P. Woodruff.
\newblock {(1} + eps)-approximate sparse recovery.
\newblock In Rafail Ostrovsky, editor, {\em {IEEE} 52nd Annual Symposium on Foundations of Computer Science, {FOCS} 2011, Palm Springs, CA, USA, October 22-25, 2011}, pages 295--304. {IEEE} Computer Society, 2011.

\bibitem[PW13]{PW13}
Eric Price and David~P. Woodruff.
\newblock Lower bounds for adaptive sparse recovery.
\newblock In Sanjeev Khanna, editor, {\em Proceedings of the Twenty-Fourth Annual {ACM-SIAM} Symposium on Discrete Algorithms, {SODA} 2013, New Orleans, Louisiana, USA, January 6-8, 2013}, pages 652--663. {SIAM}, 2013.

\bibitem[Raz16]{DBLP:conf/focs/Raz16}
Ran Raz.
\newblock Fast learning requires good memory: {A} time-space lower bound for parity learning.
\newblock In Irit Dinur, editor, {\em {IEEE} 57th Annual Symposium on Foundations of Computer Science, {FOCS} 2016, 9-11 October 2016, Hyatt Regency, New Brunswick, New Jersey, {USA}}, pages 266--275. {IEEE} Computer Society, 2016.

\bibitem[Sr.78]{M}
Robert H.~Morris Sr.
\newblock Counting large numbers of events in small registers.
\newblock {\em Commun. {ACM}}, 21(10):840--842, 1978.

\bibitem[SSV19]{Sharan19}
Vatsal Sharan, Aaron Sidford, and Gregory Valiant.
\newblock Memory-sample tradeoffs for linear regression with small error.
\newblock In {\em Proceedings of the 51st Annual ACM SIGACT Symposium on Theory of Computing}, STOC 2019, page 890–901, New York, NY, USA, 2019. Association for Computing Machinery.

\bibitem[TT18]{tcc-2018-28986}
Stefano Tessaro and Aishwarya Thiruvengadam.
\newblock Provable time-memory trade-offs: Symmetric cryptography against memory-bounded adversaries.
\newblock In {\em Theory of Cryptography}, volume 11239 of {\em Theory of Cryptography}, pages 3--32. Springer, 2018.

\bibitem[Vio15]{viola2015communication}
Emanuele Viola.
\newblock The communication complexity of addition.
\newblock {\em Combinatorica}, 35:703--747, 2015.

\bibitem[WZ21]{wz21}
David~P. Woodruff and Samson Zhou.
\newblock Separations for estimating large frequency moments on data streams.
\newblock In Nikhil Bansal, Emanuela Merelli, and James Worrell, editors, {\em 48th International Colloquium on Automata, Languages, and Programming, {ICALP} 2021, July 12-16, 2021, Glasgow, Scotland (Virtual Conference)}, volume 198 of {\em LIPIcs}, pages 112:1--112:21. Schloss Dagstuhl - Leibniz-Zentrum f{\"{u}}r Informatik, 2021.

\end{thebibliography}

\appendix

\section{Omitted Proofs in Section \ref{sec:upper}}\label{apendix:sec3}
\subsection{Proof of Claim \ref{cl:indep1}}
We restate Claim \ref{cl:indep1} and prove it: 
\indepi*

\begin{proof}
We prove the claim using induction. For $i=0, j\in [n]$, the fact that 
\[\I\left(X_{[1,j]}, \Ra_{(\le k, [1,j])}; X_{[j+1,n]}, \Ra_{(\le k, [j+1,n])}\mid \sfM_{0}\right)=0,\]
follows easily from independence of random variables $X_1,X_2,\ldots,X_n$, $\{\Ra_{i,j}\}_{i\in[k],j \in[n]}$ and $\sfM_0$ (the starting memory state of $\sfM$, which is independent of the input and private randomness used by the algorithm). Given
\begin{equation}
\label{eq:indep1}
\I\left(X_{[1,j]}, \Ra_{(\le k, [1,j])}; X_{[j+1,n]}, \Ra_{(\le k, [j+1,n])}\mid \sfM_{\le i}, \sfM_{(\le i, j)}\right)=0, \forall j\in[n],
\end{equation} we prove Equation~\eqref{eqcl:indep1} for $i$ and Equation \eqref{eqcl:indep2} for $i+1$, in two steps:
\begin{enumerate}
\item First, we only condition on the $j$-th memory state of the $(i+1)$-th pass (which gives Equation~\eqref{eqcl:indep1}) and show that
\begin{equation}
 \label{eq:indep12}
\I\left(X_{[1,j]}, \Ra_{(\le k, [1,j])}; X_{[j+1,n]}, \Ra_{(\le k, [j+1,n])}\mid \sfM_{\le i}, \sfM_{(\le i, j)}, \sfM_{(i+1,j)}\right)=0.
\end{equation}
\item Then, we condition on the end memory state of the $(i+1)$-th pass (which gives Equation~\eqref{eqcl:indep2}) and show that 
\begin{equation}
 \label{eq:indep13}
 \I\left(X_{[1,j]}, \Ra_{(\le k, [1,j])}; X_{[j+1,n]}, \Ra_{(\le k, [j+1,n])}\mid \sfM_{\le i}, \sfM_{(\le i, j)}, \sfM_{(i+1,j)}, \sfM_{(i+1)}\right)=0,
 \end{equation}
 which proves the induction case. 
\end{enumerate}
We prove Equation \eqref{eq:indep12} as follows:
\begin{align*}
&\I\left(X_{[1,j]}, \Ra_{(\le k, [1,j])}; X_{[j+1,n]}, \Ra_{(\le k, [j+1,n])}\mid \sfM_{\le i}, \sfM_{(\le i, j)}, \sfM_{(i+1,j)}\right)\\
&\;\;\;\;\;\;\;\;\;\;\;\;\le \I\left(X_{[1,j]}, \Ra_{(\le k, [1,j])}, \sfM_{(i+1,j)}; X_{[j+1,n]}, \Ra_{(\le k, [j+1,n])}\mid \sfM_{\le i}, \sfM_{(\le i, j)}\right)\\
&\;\;\;\;\;\;\;\;\;\;\;\;=\I\left(X_{[1,j]}, \Ra_{(\le k, [1,j])}; X_{[j+1,n]}, \Ra_{(\le k, [j+1,n])}\mid \sfM_{\le i}, \sfM_{(\le i, j)}\right)\\
&\;\;\;\;\;\;\;\;\;\;\;\;\;\;\;\;\;\;\;\;\;\;\;\;+\I\left( \sfM_{(i+1,j)}; X_{[j+1,n]}, \Ra_{(\le k, [j+1,n])}\mid X_{[1,j]}, \Ra_{(\le k, [1,j])},\sfM_{\le i}, \sfM_{(\le i, j)}\right) \tag{Chain rule}\\
&\;\;\;\;\;\;\;\;\;\;\;\; =\I\left( \sfM_{(i+1,j)}; X_{[j+1,n]}, \Ra_{(\le k, [j+1,n])}\mid X_{[1,j]}, \Ra_{(\le k, [1,j])},\sfM_{\le i}, \sfM_{(\le i, j)}\right) \tag{using Equation \eqref{eq:indep1}}\\
&\;\;\;\;\;\;\;\;\;\;\;\;=0\tag{as $\sfM_{(i+1,j)}$ is a deterministic function of $\sfM_{i}, X_{[1,j]}$ and $\Ra_{(i+1,[1,j])}$}
\end{align*}
We prove Equation \eqref{eq:indep13} as follows:
\begin{align*}
&\I\left(X_{[1,j]}, \Ra_{(\le k, [1,j])}; X_{[j+1,n]}, \Ra_{(\le k, [j+1,n])}\mid \sfM_{\le i}, \sfM_{(\le i, j)}, \sfM_{(i+1,j)},\sfM_{(i+1)}\right)\\
&\;\;\;\;\;\;\;\;\;\;\;\;\le \I\left(X_{[1,j]}, \Ra_{(\le k, [1,j])}; X_{[j+1,n]}, \Ra_{(\le k, [j+1,n])}, \sfM_{(i+1)}\mid \sfM_{\le i}, \sfM_{(\le i, j)}, \sfM_{(i+1,j)}\right)\\
&\;\;\;\;\;\;\;\;\;\;\;\;=\I\left(X_{[1,j]}, \Ra_{(\le k, [1,j])}; X_{[j+1,n]}, \Ra_{(\le k, [j+1,n])}\mid \sfM_{\le i}, \sfM_{(\le i, j)}, \sfM_{(i+1,j)}\right)\\
&\;\;\;\;\;\;\;\;\;\;\;\;\;\;\;\;\;\;\;\;\;\;\;\;+\I\left(X_{[1,j]}, \Ra_{(\le k, [1,j])}; \sfM_{(i+1)}\mid X_{[j+1,n]}, \Ra_{(\le k, [j+1,n])}, \sfM_{\le i}, \sfM_{(\le i, j)}, \sfM_{(i+1,j)}\right) \tag{Chain rule}\\
&\;\;\;\;\;\;\;\;\;\;\;\; =\I\left(X_{[1,j]}, \Ra_{(\le k, [1,j])}; \sfM_{(i+1)}\mid X_{[j+1,n]}, \Ra_{(\le k, [j+1,n])}, \sfM_{\le i}, \sfM_{(\le i, j)}, \sfM_{(i+1,j)}\right)  \tag{using Equation \eqref{eq:indep12}}\\
&\;\;\;\;\;\;\;\;\;\;\;\;=0\tag{as $\sfM_{(i+1)}$ is deterministic function of $\sfM_{(i+1,j)}, X_{[j+1,n]}$ and $\Ra_{(i+1,[j+1,n])}$}
\end{align*}
This proves the claim using induction.
%\begin{equation}
%\label{eq:indep1a}
%left(X_{[1,j]}, \Ra_{(\le k, [1,j])}\indep X_{[j+1,n]}, \Ra_{(\le k, [j+1,n])}\mid \sfM_{\le i}, \sfM_{(\le i, j)}\right).
%\end{equation}
%For this, we use Bayes rule as follows:
%\begin{align*}
%&\Pr\left[X_{[1,j]}=x_{[1,j]}, \Ra_{(\le k, [1,j])}=r_{(\le k, [1,j])}\mid X_{[j+1,n]}=x_{[j+1,n]}, \Ra_{(\le k, [j+1,n])}=r_{(\le k, [j,n])}, \sfM_{\le i, j}=m_{\le i,j}, \sfM_{\le k}=m_{\le k}\right]\\
%&\;\;\;\;\;\;\;=\Pr\left[X_{[1,j]}=x_{[1,j]}, \sfM_{(\le k, [1,j-1])}=m_{(\le k, [1,j-1])}\mid  \sfM_{(\le k, j)}=m_{(\le k, j)}, \sfM_{\le k}=m_{\le k}\right]
%\end{align*}
\end{proof}
\noindent
\subsection{Proof of Corollary \ref{cor:indep1} and \ref{cor:indep2}}
We first review the statement of Corollary \ref{cor:indep1}, and give a formal proof. 
\indepii*

\begin{proof}
We prove the corollary as follows:
\begin{align*}
&\I\left(X_{[1,j]}, \sfM_{(\le i, [0,j-1])}; X_{[j+1,n]}, \sfM_{(\le i, [j+1,n])}\mid \sfM_{< i}, \sfM_{(\le i, j)}\right)\\
&\;\;\;\;\;\le \I\left(X_{[1,j]}, \sfM_{(\le i, [0,j-1])}, \Ra_{(\le k, [1,j])}; X_{[j+1,n]}, \sfM_{(\le i, [j+1,n])}, \Ra_{(\le k, [j+1,n])}\mid \sfM_{< i}, \sfM_{(\le i, j)}\right)\\
&\;\;\;\;\;=\I\left(X_{[1,j]}, \Ra_{(\le k, [1,j])}; X_{[j+1,n]}, \sfM_{(\le i, [j+1,n])}, \Ra_{(\le k, [j+1,n])}\mid \sfM_{< i}, \sfM_{(\le i, j)}\right)\\
&\;\;\;\;\;\;\;\;\;\;+\I\left( \sfM_{(\le i, [0,j-1])}; X_{[j+1,n]}, \sfM_{(\le i, [j+1,n])}, \Ra_{(\le k, [j+1,n])}\mid X_{[1,j]},\Ra_{(\le k, [1,j])}, \sfM_{< i}, \sfM_{(\le i, j)}\right)\tag{Chain rule}\\
&\;\;\;\;\;=\I\left(X_{[1,j]}, \Ra_{(\le k, [1,j])}; X_{[j+1,n]}, \sfM_{(\le i, [j+1,n])}, \Ra_{(\le k, [j+1,n])}\mid \sfM_{< i}, \sfM_{(\le i, j)}\right) \tag{as $\forall i_1\le i$, $\sfM_{(i_1,[0,j-1])}$ are deterministic functions of $\sfM_{i_1-1}$, $X_{[1,j-1]}$, and $\Ra_{(i_1, [1,j-1])}$}\\
&\;\;\;\;\;=\I\left(X_{[1,j]}, \Ra_{(\le k, [1,j])}; X_{[j+1,n]}, \Ra_{(\le k, [j+1,n])}\mid \sfM_{< i}, \sfM_{(\le i, j)}\right)\\
&\;\;\;\;\;\;\;\;\;\;+\I\left(X_{[1,j]}, \Ra_{(\le k, [1,j])}; \sfM_{(\le i, [j+1,n])}\mid X_{[j+1,n]}, \Ra_{(\le k, [j+1,n])}, \sfM_{< i}, \sfM_{(\le i, j)}\right)\tag{Chain rule}\\
&\;\;\;\;\;=0.
\end{align*}
The last equality follows from  Equation~\eqref{eqcl:indep1} in Claim \ref{cl:indep1} and the fact that $\forall i_1\le i$, $\sfM_{(i_1,[j+1,n])}$ are deterministic functions of $\sfM_{(i_1,j)}$, $X_{[j+1,n]}$, and $\Ra_{(i_1, [j+1,n])}$.
\end{proof}
\noindent
Then, we prove Corollary \ref{cor:indep2}:
\indepiii*
\begin{proof}
We prove the corollary as follows:
\begin{align*}
&\I\left(X_j; \sfM_{(i+1,j-1)}\mid \sfM_{\le i}, \sfM_{(\le  i, j-1)}\right)\\
&\le \I\left(X_j; \sfM_{(i+1,j-1)}, X_{[1,j-1]},\Ra_{(i+1,[1,j-1])}\mid \sfM_{\le i}, \sfM_{(\le i, j-1)}\right)\\
&= \I\left(X_j;  X_{[1,j-1]},\Ra_{(i+1,[1,j-1])}\mid \sfM_{\le i}, \sfM_{(\le i, j-1)}\right)+ \I\left(X_j; \sfM_{(i+1,j-1)}\mid X_{[1,j-1]},\Ra_{(i+1,[1,j-1])},\sfM_{\le i}, \sfM_{(\le i, j-1)}\right)\\
&=0.
\end{align*}
The last equality follows from Equation~\eqref{eqcl:indep2} in Claim \ref{cl:indep1} and the fact that $\forall (i+1)\in[k]$, $\sfM_{(i+1,j-1)}$ is deterministic function of $\sfM_{i}$, $X_{[1,j-1]}$, and $\Ra_{(i+1, [1,j-1])}$.
\end{proof}

\section{Omitted Proofs in Section \ref{sec:single}}

\subsection{Proofs of Claims \ref{cl:imd} and \ref{cl:iminf}}\label{sec:apim}
\proof[\textbf{Proof of Claim \ref{cl:imd}}]
The joint distribution on variables $X, \sfM_{< k}, \{\sfM_{(\le k, [0,n])}\}$ is decided as the $k$-pass algorithm unfolds. We first sample $X$, then $\sfM_0=\sfM_{(1,0)}$ which is independent of $X$, then run the first pass which identifies variables $\sfM_{(1,[1,n])}$ (where $\sfM_1=\sfM_{(1,n)}=\sfM_{(2,0)}$), then we run the second pass and so on. 
However, to prove the equivalence of these joint distributions, we sample the variables for the $k$-passes in parallel. 
We prove the claim using the chain rule by conditioning in the following order (similarly for $\sfM'$ and $X'$):
\begin{enumerate}
\item $(\sfM_0,\sfM_1,\ldots,\sfM_{k-1})$.
\item $(\sfM_{(1,0)},\sfM_{(2,0)},\ldots,\sfM_{(k,0)})$.
\item Repeat Step \ref{it:imd1} to Step \ref{it:imd2} for $j=1$ to $n$.
\item \label{it:imd1} $X_j$ conditioned on $\sfM_{< k}$, $X_{[1,j-1]}$ and $\sfM_{(\le k,[0,j-1])}$.
\item \label{it:imd2} $(\sfM_{(1,j)},\sfM_{(2,j)},\ldots, \sfM_{(k,j)})$ conditioned on $X_j$, $\sfM_{< k}$, $X_{[1,j-1]}$ and $\sfM_{(\le k,[0,j-1])}$.
\end{enumerate}
As $(m'_0,m'_1,\ldots,m'_{k-1})$ is drawn from the joint distribution on $(\sfM_0,\sfM_1,\ldots,\sfM_{k-1})$, by definition 
\[\Pr[\sfM'_{< k}=m'_{< k}]=\Pr[\sfM_{< k}=m'_{< k}], \; \forall \;m'_{< k}.\]
Similarly, by definitions of $\sfM'_{(\le k,0)}$ and $\sfM_{(\le k,0)}$, the distributions of these random variables conditioned on $\sfM'_{<k}$ and $\sfM_{<k}$ respectively, are identical. 
Next, we look at the conditional distribution in Step \ref{it:imd1} for $j\in[n]$. It is easy to see that, given a value of $\beta_j$ in Step \ref{st:im1} of Algorithm \ref{al:im}, $x_j'=1$ with probability $\frac{1}{2}+\beta_j$, as $Y_j\sim \Ber(1/2)$. Therefore $\forall \; m'_{< k}, m'_{([1,k],0)}=m'_{[0,k-1]}, x'_{[1,j-1]}, m'_{(\le k,[1,j-1])} $, 
\begin{align*}
&\Pr\left[X'_j=1\mid \sfM'_{< k}=m'_{< k}, X'_{[1,j-1]}=x'_{[1,j-1]},  \sfM'_{(\le k,[0,j-1])}=m'_{(\le k,[0,j-1])}\right]\\
&\;\;\;\;\;=\Pr\left[X_j=1\mid \sfM_{(\le k,j-1)}=m'_{(\le k,j-1)}, \sfM_{< k}=m'_{< k}\right]\\
&\;\;\;\;\;=\Pr\left[X_j=1\mid \sfM_{< k}=m'_{< k}, X_{[1,j-1]}=x'_{[1,j-1]},  \sfM_{(\le k,[0,j-1])}=m'_{(\le k,[0,j-1])}\right].
\end{align*}
The last equality follows from Corollary \ref{cor:indep1}, which implies that
\[\left(X_j\indep X_{[1,j-1]}, \sfM_{(\le k, [0,j-2])}\middle| \sfM_{< k},\sfM_{(\le k,j-1)} \right).\]
Finally, we look at the conditional distribution in Step \ref{it:imd2}. As $(m'_{(1,j)},m'_{(2,j)}, \ldots, m'_{(k,j)})$ is drawn from a distribution that only depends on values $m'_{(\le k,j-1)}$, $m'_{< k}$ and $x_j'$ (see Step \ref{st:im2} of Algorithm \ref{al:im}), we have, $\forall\; m'_{< k}, m'_{([1,k],0)}=m'_{[0,k-1]}, x'_{[1,j]}, m'_{(\le k,[1,j])} $,  
\begin{align*}
&\Pr\left[\sfM'_{(\le k,j)}=m'_{(\le k,j)}\middle| X'_j=x'_j, \sfM'_{< k}=m'_{< k}, X'_{[1,j-1]}=x'_{[1,j-1]}, \sfM'_{(\le k,[0,j-1])}=m'_{(\le k,[0,j-1])} \right]\\
&\;\;\;\;\;=\Pr\left[\sfM_{(\le k,j)}=m'_{(\le k,j)}\mid \sfM_{(\le k,j-1)}=m'_{(\le k,j-1)},\; \sfM_{< k}=m'_{< k}, X_j=x_j'\right]\\
&\;\;\;\;\;=\Pr\left[\sfM_{(\le k,j)}=m'_{(\le k,j)}\mid X_j=x_j', \sfM_{< k}=m'_{< k}, X_{[1,j-1]}=x'_{[1,j-1]},  \sfM_{(\le k,[0,j-1])}=m'_{(\le k,[0,j-1])}\right]
\end{align*}
The last equality follows from Corollary \ref{cor:indep1}, which implies that
\begin{align*}
\left(\sfM_{(\le k, j)}, X_j\indep X_{[1,j-1]}, \sfM_{(\le k, [0,j-2])}\middle| \sfM_{< k},\sfM_{(\le k,j-1)} \right). &\qedhere
\end{align*}

\begin{proof}[\textbf{Proof of Claim \ref{cl:iminf}}]
As algorithm $\Imi$ remembers $(m'_{(1,j)},m'_{(2,j)}, \ldots, m'_{(k,j)},m'_0,m'_1,\ldots,m'_{k-1})$ after $j$th time-step ($j\in\{0,1,\ldots,n\}$), we rewrite $IC(\Imi)$ as follows:
\begin{align*}
IC(\Imi)&=\sum_{j=1}^n\sum_{\ell=1}^j \I\left(\Imi_{j};Y_{\ell}|\Imi_{\ell-1}\right) \\
&\le \sum_{j=1}^n\sum_{\ell=1}^j \I\left(\Imi_{j};X'_\ell, Y_{\ell}|\Imi_{\ell-1}\right)\\
&=\sum_{j=1}^n\sum_{\ell=1}^j \I\left(\Imi_{j};X'_\ell|\Imi_{\ell-1}\right)+\sum_{j=1}^n\sum_{\ell=1}^j\I\left(\Imi_{j}; Y_{\ell}|X'_\ell,\Imi_{\ell-1}\right)\\
&=\sum_{j=1}^n\sum_{\ell=1}^j \I\left(\Imi_{j};X'_\ell|\Imi_{\ell-1}\right)\tag{explained below}\\
&=\sum_{j=1}^n\sum_{\ell=1}^j \I\left(\sfM'_{(\le k,j)}, \sfM'_{<k};X'_{\ell}|\sfM'_{<k},\sfM'_{(\le k, \ell-1)}\right)\\
&=\sum_{j=1}^n\sum_{\ell=1}^j \I\left(\sfM'_{(\le k,j)};X'_{\ell}|\sfM'_{<k},\sfM'_{(\le k, \ell-1)}\right)\\
&=\sum_{j=1}^n\sum_{\ell=1}^j \I\left(\sfM_{(\le k,j)}; X_{\ell}|\sfM_{<k},\sfM_{(\le k, \ell-1)}\right) \tag{using Claim \ref{cl:imd}}\\
&=\miccoin(\sfM).
\end{align*}
We prove the third equality as follows. 
\begin{align*}
\sum_{j=1}^n\sum_{\ell=1}^j\I\left(\Imi_{j}; Y_{\ell}|X'_\ell,\Imi_{\ell-1}\right)&\le \sum_{j=1}^n\sum_{\ell=1}^j\I\left(\Imi_{j}\Imi_{\ell}; Y_{\ell}|X'_\ell,\Imi_{\ell-1}\right)\\
&\le \sum_{j=1}^n\sum_{\ell=1}^j\I\left(\Imi_{\ell}; Y_{\ell}|X'_\ell,\Imi_{\ell-1}\right)+\sum_{j=1}^n\sum_{\ell=1}^j\I\left(\Imi_{j}; Y_{\ell},X'_\ell,\Imi_{\ell-1}|\Imi_{\ell}\right).
\end{align*}
Since given $x'_\ell$, $im_{\ell-1} = (m'_{(\le k, \ell-1)},m'_{<k}$), $m'_{(\le k, \ell )}$ ($im_\ell$) is drawn from a distribution independent of value $y_{\ell}$ (Step \ref{st:im2} of Algorithm \ref{al:im}), 
\[\I\left(\Imi_{\ell}; Y_{\ell}|X'_\ell,\Imi_{\ell-1}\right)=0, \;\forall \ell\le j\le n.\]
Since $\Imi$ is a one-pass algorithm over a stream of independent input bits, 
\[\I\left(\Imi_{j}; Y_{\ell},X'_\ell,\Imi_{\ell-1}|\Imi_{\ell}\right)=0, \;\forall \ell\le j\le n;\]
because conditioned on $\Imi_\ell$, $\Imi_j$ only depends on the input stream after the $\ell$th time-step and is independent of the inner workings of the algorithm before the $\ell$th time-step. This proves the claim. 
\end{proof}

\subsection{Proofs of Propositions \ref{cl:aprsum} and \ref{cl:aprinf}}\label{sec:apapr}
\proof[\textbf{Proof of Proposition \ref{cl:aprsum}}]
For a given input $\e=(\e_1,\ldots,\e_n)$, let $\ell_\e$ be the smallest index $\ell$ such that $\sum_{j=1}^\ell\one_{\e_j\neq 0}= 80\B\log n$. If $\sum_{j=1}^n\one_{\e_j\neq 0}<80\B\log n$, then set $\ell_\e=n+1$. 
While executing $\Apr$ using randomness $r$, let $j_\e^r$ denote the index when count $\Cb$ reaches $20\log n\cdot p\B$, else set $j_\e^r=n+1$ -- when count $\Cb$ is always less than $20\log n\cdot p\B$. 
%Let $\Cb_\e^j$ be a random variable equal to count $\Cb$ after reading the first $j$ input elements, and 
Let $\calG_\e$ be an event defined as follows: $\calG_\e=\{r\mid j_\e^r>\ell_\e\}$.
%\[\calG_\e=\{\ell_\e\neq \perp \;\land \;(\Cb_\e^{\ell_\e}<  20\log n\cdot p\B)\}.\]

We provide no non-trivial approximation guarantees for the output of the algorithm $\Apr$ conditioned on $\calG_\e$, and use the Chernoff bound to show that $\forall \e\in\{-1,0,1\}^n$, $\Pr[\calG_\e]\le 1/n^5$ (over the private randomness $r\sim R^\Apr$). Let $Z_1,\ldots,Z_{80\B\log n}$ be independent random variables such that 
\[Z_j=\begin{cases}
    1       & \quad \text{with probability } p \\
    0  & \quad \text{otherwise}.
  \end{cases}\]
  Note that if $\ell_\e=n+1$, then $\Pr[\calG_\e]=0$. Otherwise, $\Pr[\calG_\e]=\Pr_{r}\left[j_\e^r>\ell_\e\right]$ is the probability that count $\Cb$ is less than $20\log n \cdot p\B$ after reading $\ell_\e$ elements of $\e$. Note that, as long as $\Cb<20\log n \cdot p\B$, it increases by $\one_{\e_j\neq 0}$ with probability $p$ at the $j$th time-step. Therefore, $\Pr_{r}\left[j_\e^r>\ell_\e\right]$ is equal to the probability that a sum of $\ell_\e$ independent random variables (where $j$th one is $\one_{\e_j\neq 0}$ with probability $p$ and 0 otherwise) is less than $ 20\log n \cdot p\B$. As time-steps with $\e_j=0$ contribute 0 to the sum always and other time-steps contribute $1$ with probability $p$, we have that
\begin{align*}
\Pr[\calG_\e]&=\Pr_{r\mid\forall j,\; r_j\sim \Ber(p)}\left[j_\e^r>\ell_\e\right]\\
&=\Pr\left[\sum_{j=1}^{80\B\log n}Z_j< 20\log n \cdot p\B\right] &\left(\text{as }\sum_{j\le \ell_\e}\one_{\e_j\neq 0}= 80\B\log n\right)\\
&=\Pr\left[\left|\sum_jZ_j-\bbE\left[\sum_jZ_j\right]\right|>40\log n\cdot p\B\right]&\left(\text{as }\bbE\left[\sum_jZ_j\right]=80\log n\cdot p\B\right)\\
&\le 2 \exp{\left(-\frac{\frac{1}{4}(80\log n\cdot p\B)}{3}\right)}&\text{(Chernoff bound, Equation \eqref{eq:ch1})}\\
&<2\exp{\left(-6\log n\cdot p\B\right)}\\
&<1/n^5 &\text{($p\B>1$ as $\B>\gamma\sqrt{n}$ and $\gamma>2/\sqrt{n}$)}
\end{align*}

Note that while executing $\Apr$ on input $\e$ and randomness $r$, $\Deltasm/p+\Deltalg$ can be written as $\left(\sum_{j=1}^{j_\e^r} r_j \e_j\right)/p+\sum_{j=j_\e^r+1}^n\e_j$. Next, we first use Bernstein's inequality to show that $\forall \ell\le \ell_a$, 
\begin{equation}
\label{eq:aprsum1}
\left|\frac{\sum_{j=1}^\ell r_j \e_j}{p}-\sum_{j=1}^{\ell}\e_j\right|\le \frac{\gamma}{2}\sqrt{n} \text{ with probability $1-\frac{1}{2n^4}$ when $r_1,\ldots, r_n\sim i.i.d.\; \Ber(p)$}
\end{equation}
Let $Z_1',\ldots,Z_n'$ be independent random variables such that $Z_j'=\e_j$ with probability $p$ and 0 otherwise. Thus, $\forall j,\; |Z_j'|\le 1$ and $\Var(Z_j')\le p$. For $p=1$, Equation \eqref{eq:aprsum1} is trivially true, and for $p<1$, the above probability can be rewritten as follows:
\begin{align*}
&\Pr_{\forall j, \;r_j\sim\Ber(p)}\left[\left|\frac{\sum_{j=1}^\ell r_j \e_j}{p}-\sum_{j=1}^{\ell}\e_j\right|> \frac{\gamma}{2}\sqrt{n}\right]\\
&\;\;\;\;\;\;\;=\Pr_{\forall j,\; r_j\sim\Ber(p)}\left[\left|\sum_{j=1}^\ell r_j \e_j-\bbE\left[\sum_{j=1}^{\ell}r_j\e_j\right]\right|> \frac{\gamma}{2}p\sqrt{n}\right]\\
&\;\;\;\;\;\;\;=\Pr\left[\left|\sum_{j\le \ell\mid a_j\neq 0} Z_j'-\bbE\left[\sum_{j\le \ell\mid a_j\neq 0}Z_j'\right]\right|> \frac{\gamma}{2}p\sqrt{n}\right]\\
&\;\;\;\;\;\;\;\le 2\exp\left(-\frac{\frac{1}{2}\cdot \left(\frac{\gamma}{2}p\sqrt{n}\right)^2}{p\left(\sum_{j=1}^\ell\one_{\e_j\neq 0}\right)+\frac{1}{3}\left(\frac{\gamma}{2}p\sqrt{n}\right)}\right)
&&\text{(Berstein's inequality, Equation \eqref{eq:ber})}\\
&\;\;\;\;\;\;\;\le 2\exp\left(-\frac{\frac{1}{2}\cdot \left(\frac{\gamma}{2}p\sqrt{n}\right)^2}{p\left(80\log n\B\right)+\frac{1}{3}\left(\frac{\gamma}{2}p\sqrt{n}\right)}\right)
&&\left(\text{as } \ell\le \ell_\e \text{ and }\sum_{j\le \ell_\e}\one_{\e_j\neq 0}= 80\B\log n\right)\\
&\;\;\;\;\;\;\;\le 2\exp\left(-\frac{\frac{1}{2}\cdot \left(\frac{\gamma}{2}p\sqrt{n}\right)^2}{2p\left(80\log n\B\right)}\right)
&&\left(\text{as } \B>\gamma\sqrt{n}\right)\\
&\;\;\;\;\;\;\;\;=2\exp\left(-p\cdot \frac{\gamma^2n}{1280\log n\B}\right)\\
&\;\;\;\;\;\;\;\;<\frac{1}{2n^4}
&&\left(\text{as }p<1\text{ implies } p=6000\log^2 n\cdot\frac{\B}{\gamma^2 n}\right)
\end{align*}
Let $\calG_\e'$ be the event as follows: $\calG_\e'=\left\{r\mid\exists \ell\le \ell_\e, \left|\frac{\sum_{j=1}^\ell r_j \e_j}{p}-\sum_{j=1}^{\ell}\e_j\right|>\frac{\gamma}{2}\sqrt{n}\right\}$. 
Above, we proved that $\Pr_{r\mid \forall j, \;r_j\sim \Ber(p)}[\calG_\e']<\frac{1}{2n^3}$. 
As $\sum_{j=1}^n\e_j\in[-n, n]$, we note that for all real values $v$, \[\left|\max\{\min\{v,n\},-n\}-\sum_{j=1}^n\e_j\right|\le\left|v-\sum_{j=1}^n\e_j\right|.\] Therefore, for all inputs $\e$, as conditioned on $\neg \calG_\e$, $j_\e^r\le \ell_a$,
\begin{align*}
\Pr_{r\mid \forall j, \;r_j\sim \Ber(p)}\left[\left|\Apr(\e,r)-\sum_{j=1}\e_j\right|>\frac{\gamma}{2}\sqrt{n}\right]
&\le\Pr_{r\mid \forall j, \;r_j\sim \Ber(p)}\left[\left|\left(\sum_{j=1}^{j_\e^r} r_j \e_j\right)/p-\sum_{j=1}^{j_\e^r}\e_j\right|>\frac{\gamma}{2}\sqrt{n}\right]\\
&\le \Pr[\calG_\e]+\Pr[\calG_\e']<\frac{1}{n^3}.
\end{align*}

This proves the proposition. The implication follows from the fact that in the worst-case, the output of the algorithm $\Apr$ can be $2n-$far from $\sum_{j=1}^n \e_j$. Therefore, $\forall$ distributions $\calD$ on $\{-1,0,1\}^n$, as $\gamma>4/\sqrt{n}$,
\begin{align*}
\bbE_{\e\sim \calD,r\sim R^\Apr} \left[\left(\Apr(\e, r)-\sum_{j=1}^n \e_j\right)^2\right]< \frac{1}{n^3}\cdot 4n^2 +\frac{\gamma^2n}{4}<\gamma^2n.&\qedhere
\end{align*}

\proof[\textbf{Proof of Proposition \ref{cl:aprinf}}]
As the starting state of $\Apr$ is deterministic, $\Ent(\Apr_0)=0$. At every time-step, $\Apr$ maintains three counters $\Deltasm$, $\Cb$ and $\Deltalg$. As $\Apr$ stops increasing counts $\Deltasm$ and $\Cb$, whenever $\Cb$ reaches $20\log n\cdot p\B$, $\Deltasm$ and $\Cb$ are integers in the sets $\{-20\log n\cdot p\B,\ldots,0,\ldots,20\log n\cdot p\B\}$ and $\{0,\ldots,20\log n\cdot p\B\}$ respectively. And, $\Deltalg$ is an integer in the set $\{-n,\ldots,n\}$. 
%Note that, $\Deltalg=0$ as long as $\Cb<20\log n\cdot p\B$. 

Let $\In$ be the set of inputs $\e$ with at most $4\B\log n$ non-zero indices, that is, $$\In=\left\{\e\mid \sum_{j=1}^n\one_{\e_j\neq 0}< 4\B\log n\right\}.$$ As $\bbE_{\e\sim \calD}\left[\sum_{j=1}^n \one_{\e_j\neq 0}\right]\le \B$, using Markov's inequality, we have that
\begin{equation}
\label{eq:aprinf1}
\Pr_{\e\sim \calD}[\e\not\in\In]=\Pr_{\e\sim \calD}\left[\sum_{j=1}^n\one_{\e_j\neq 0}\ge 4\B\log n\right]\le \frac{\bbE_{\e\sim \calD}\left[\sum_{j=1}^n \one_{\e_j\neq 0}\right]}{4\B\log n}\le \frac{1}{4\log n}.
\end{equation}
While executing $\Apr$ on input $\e$ using randomness $r$, let $j_\e^r$ denote the index when count $\Cb$ reaches $20\log n\cdot p\B$, else set $j_\e^r=n+1$ -- when count $\Cb$ is always less than $20\log n\cdot p\B$. Given input $\e\in\In$, let $\Ra_\e$ be the set of private randomness $r$ such that $j_\e^r=n+1$. Next, we prove that 
\begin{equation}
\label{eq:aprinf2}
\Pr_{r\mid \forall j,\;r_j\sim\Ber(p)}[r\not\in\Ra_\e]<\frac{1}{n} \text{ for all inputs }\e\in\In.
\end{equation}
First, we show that, for $a\in\In$, $\Pr_r[j_\e^r=\ell]< 1/n^2,\;\forall \ell\in[n]$. Equation \eqref{eq:aprinf2} follows using the union bound.  Let $c_\e^\ell=\sum_{j=1}^\ell\one_{\e_j\neq 0}$. As $\e\in\In,\; \forall\ell\in[n],\; c_\e^\ell< 4\B\log n$. Let $Z_1,\ldots,Z_{4\B\log n}$ be independent random variables such that 
\[Z_j=\begin{cases}
    1       & \quad \text{with probability } p \\
    0  & \quad \text{otherwise}.
  \end{cases}\]
  $\Pr_{r}\left[j_\e^r=\ell\right]$ is the probability that count $\Cb$ is equal to $20\log n \cdot p\B$ after reading $\ell$ elements of $\e$, and is less than $20\log n \cdot p\B$ after the $(\ell-1)$th time-step. Note that, as long as $\Cb<20\log n \cdot p\B$, it increases by $\one_{\e_j\neq 0}$ with probability $p$ at the $j$th time-step. Therefore, $\Pr_{r}\left[j_\e^r=\ell\right]$ is at most the probability that a sum of $\ell$ independent random variables (where the $j$th one is $\one_{\e_j\neq 0}$ with probability $p$ and 0 otherwise) is at least $ 20\log n \cdot p\B$. As time-steps with $\e_j=0$ contribute 0 to the sum always and other time-steps contribute $1$ with probability $p$, we have that (recall $c_\e^\ell=\sum_{j=1}^\ell\one_{\e_j\neq 0}$) 
  
\begin{align*}
\Pr_{r\mid \forall j,\;r_j\sim\Ber(p)}[j_\e^r=\ell]&\le\Pr\left[\sum_{j=1}^{c_\e^\ell}Z_j\ge 20\log n \cdot p\B\right]\\
&=\Pr\left[\sum_{j=1}^{c_\e^\ell}Z_j\ge p\cdot c_\e^\ell+16\log n\cdot p\B \right]&\left(\text{as } \bbE\left[\sum_{j=1}^{c_\e^\ell}Z_j\right]=p\cdot c_\e^\ell< 4\log n\cdot p\B\right)\\
&\le \exp{\left(-\frac{16\log n\cdot p\B}{3}\right)}&\text{(Chernoff bound, Equation \eqref{eq:ch3})}\\
&<1/n^2 &\text{($p\B>1$ as $\B>\gamma\sqrt{n}$ and $\gamma>2/\sqrt{n}$)}
\end{align*}

While executing $\Apr$ on an input $\e\in\In$ using randomness $r\in\Ra_\e$, count $\Cb$ never reaches $20\log n \cdot p\B$ and hence, $\Deltalg=0$ for the entirety of the algorithm. Using Equations \eqref{eq:aprinf1} and \eqref{eq:aprinf2},
\begin{equation}
\label{eq:aprinf3}
\Pr_{\e\sim\calD,\; r\mid \forall j,\;r_j\sim\Ber(p)}[\e\in\In\;\land \;r\in\Ra_\e]\ge \left(1-\frac{1}{4\log n}\right)\left(1-\frac{1}{n}\right)>1-\frac{1}{3\log n}.
\end{equation}
Let $\Deltasm_j$, $\Cb_j$ and $\Deltalg_j$ be random variables for values $\Deltasm$, $\Cb$ and $\Deltalg$ after $\Apr$ reads the first $j$ input elements. Here, randomness comes from both the input $\e$ to the algorithm $\Apr$, as well as the private randomness $r$ used by the algorithm. Recall that, $\Apr_j$ represents the random variable for the memory state of $\Apr$ after reading $j$ input elements. Therefore, \[\forall j\in[n], \;\Ent(\Apr_j)=\Ent(\Deltasm_j,\Cb_j,\Deltalg_j)\le \Ent(\Deltasm_j)+\Ent(\Cb_j)+\Ent(\Deltalg_j).\]
Algorithm $\Apr$ also remembers values $p$ and $20\log n \cdot p \B$; but as these are constants, they do not have entropy. As discussed before, $\Deltasm$ and $\Cb$ are integers in the sets $\{-20\log n\cdot p\B,\ldots,20\log n\cdot p\B\}$ and $\{0,\ldots,20\log n\cdot p\B\}$ respectively. And, $\Deltalg$ is an integer in the set $\{-n,\ldots,n\}$. Therefore, $\Ent(\Deltasm_j), \Ent(\Cb_j)\le \log (40\log n\cdot p\B+1)$. Equation \eqref{eq:aprinf3} implies that $\Pr[\Deltalg_j=0]>1-\frac{1}{3\log n}$. Let $q_\ell=\Pr[\Deltalg_j=\ell], \forall \ell\in\{-n,\ldots, n\}$. Then,
\begin{align*}
\Ent(\Deltalg_j)&=q_0\log{\left(\frac{1}{q_0}\right)}+\sum_{\ell\in\{-n,\ldots,-1,1,\ldots, n\}}q_\ell\log{\left(\frac{1}{q_\ell}\right)}\\
&\le q_0\log{\left(\frac{1}{q_0}\right)}+(1-q_0)\log{\frac{2n}{1-q_0}}&\text{(Using Jensen's inequality)}\\
&\le 1+(1-q_0)\log 2n\\
&\le 1+\frac{1}{3\log n}\log 2n \; \le 2&\left(q_0>1-\frac{1}{3\log n}\right)
\end{align*}
Therefore, $\Ent(\Apr_j)\le 2\log{(40\log n\cdot p\B+1)} +2$. As $p\le 6000\log^2 n\cdot \left(\frac{\B}{\gamma^2 n}\right)$, we get 
\begin{align*}
\Ent(\Apr_j)&\le 2+ 2\log{\left(50\log n\cdot 6000\log^2 n\cdot \left(\frac{\B}{\gamma^2 n}\right)\B\right)}\\
&< 40+6\log \log n+2\log{\left(\frac{B}{\gamma\sqrt{n}}\right)}.&\qedhere
\end{align*}

\subsection{Omitted Proofs from Subsection \ref{subsec:a}} \label{sec:apsingle}
\proof[\textbf{Proof of Claim \ref{cl:acomb}}]
Note that, by linearity of expectation, 
$\bbE\left[\sum_{j=1}^n\one_{Y_j\neq X_j'}\right]=\sum_{j=1}^n{\bbE\left[\one_{Y_j\neq X_j'}\right]}$.
Next, we focus on $\bbE\left[\one_{Y_j\neq X_j'}\right]$ for a given $j\in n$. By the law of total expectation, 
\begin{equation}
\label{eq:clacomb1}
\bbE\left[\one_{Y_j\neq X_j'}\right]=\bbE_{\sfM'_{<k},\sfM'_{(\le k, j-1)}}\left[\bbE\left[\one_{Y_j\neq X_j'}\middle| \sfM'_{<k}=m'_{<k}, \sfM'_{(\le k, j-1)}=m'_{(\le k, j-1)}\right]\right].
\end{equation}
Conditioned on $\sfM'_{<k}=m'_{<k}, \sfM'_{(\le k, j-1)}=m'_{(\le k, j-1)}$, the relation between $y_j$ and $x'_j$ is determined by Steps \ref{st:a1} to \ref{st:a2} of Algorithm \ref{al:a}. It is easy to see that $x'_j\neq y_j$ with probability $2|\beta_j|$ whenever $y_j=-\text{sign}(\beta_j)$. Here, $\text{sign}(\beta_j)=1$ if $\beta_j> 0$ and $-1$ otherwise. 

Therefore, 
$\bbE\left[\one_{Y_j\neq X_j'}\middle| \sfM'_{<k}=m'_{<k}, \sfM'_{(\le k, j-1)}=m'_{(\le k, j-1)}\right]$ can be rewritten as 
\begin{align*}
2|\beta_j|\cdot \Pr\left[Y_j\neq \text{sign}(\beta_j)\middle| \sfM'_{<k}=m'_{<k}, \sfM'_{(\le k, j-1)}=m'_{(\le k, j-1)}\right]=|\beta_j|,
\end{align*}
where the equality follows from the fact that $Y_j$ is drawn from uniform distribution on $\zo$, independent of $\Imi_{j-1}=(\sfM'_{<k}, \sfM'_{(\le k, j-1)})$. Therefore, using Step \ref{st:a1},
\begin{equation}
\label{eq:clacomb2}
\bbE\left[\one_{Y_j\neq X_j'}\middle| \sfM'_{<k}=m'_{<k}, \sfM'_{(\le k, j-1)}=m'_{(\le k, j-1)}\right]=\left|\Pr\left[X_j=1\;|\; \sfM_{(\le k,j-1)}=m'_{(\le k,j-1)}, \sfM_{< k}=m'_{<k}\right]-\frac{1}{2}\right|.
\end{equation}
Recall that $X_1, X_2,\ldots, X_n$ is the input to the $k$-pass algorithm $\sfM$, drawn from the  uniform distribution on $\zo^n$. 
Therefore, the entropy of the $j$th bit, $\Ent(X_j)=1$. 
We compare the R.H.S. of Equation~\eqref{eq:clacomb2} to 
\[\Ent(X_j)-\Ent(X_j|\sfM_{(\le k,j-1)}=m'_{(\le k,j-1)}, \sfM_{< k}=m'_{<k}).\]
Let $Z$ be a binary random variable taking value $1$ with probability $q$. Then, $\Ent(Z)=-(q\log q+(1-q)\log{(1-q}))$. It follows from easy calculations that \[\forall q\in[0,1],\;\; 1+q\log q+(1-q)\log {(1-q)}\ge \left(q-\frac{1}{2}\right)^2,\]
where equality holds at $q=1/2$.
That is, $1-\Ent(Z)\ge \left(\Pr[Z=1]-\frac{1}{2}\right)^2$. Therefore, we can rewrite Equation~\eqref{eq:clacomb2} as 
\begin{align*}
\bbE\left[\one_{Y_j\neq X_j'}\middle| \sfM'_{<k}=m'_{<k}, \sfM'_{(\le k, j-1)}=m'_{(\le k, j-1)}\right]&=\left|\Pr\left[X_j=1\;|\; \sfM_{(\le k,j-1)}=m'_{(\le k,j-1)}, \sfM_{< k}=m'_{<k}\right]-\frac{1}{2}\right|\\
&\le \sqrt{1-\Ent(X_j|\sfM_{(\le k,j-1)}=m'_{(\le k,j-1)}, \sfM_{< k}=m'_{<k})}\\
&=\sqrt{\Ent(X_j)-\Ent(X_j|\sfM_{(\le k,j-1)}=m'_{(\le k,j-1)}, \sfM_{< k}=m'_{<k})}.
\end{align*}
Coming back to Equation~\eqref{eq:clacomb1}, we can bound the R.H.S. by 
\[\bbE_{(m'_{(\le k,j-1)},m'_{<k})\sim (\sfM'_{(\le k,j-1)},\sfM'_{< k})}\left[\sqrt{\Ent(X_j)-\Ent(X_j|\sfM_{(\le k,j-1)}=m'_{(\le k,j-1)}, \sfM_{< k}=m'_{<k})}\right]\]
As the joint distribution on $(\sfM'_{(\le k,j-1)},\sfM'_{< k})$ is identical to that on $(\sfM_{(\le k,j-1)},\sfM_{< k})$ (Claim \ref{cl:imd}), we can rewrite the R.H.S. as follows (using concavity of the $\sqrt{z}$ function)
\begin{align*}
&\bbE_{(m'_{(\le k,j-1)},m'_{<k})\sim (\sfM_{(\le k,j-1)},\sfM_{< k})}\left[\sqrt{\Ent(X_j)-\Ent(X_j|\sfM_{(\le k,j-1)}=m'_{(\le k,j-1)}, \sfM_{< k}=m'_{<k})}\right]\\
&\;\;\;\;\;\le \sqrt{\bbE_{(m'_{(\le k,j-1)},m'_{<k})\sim (\sfM_{(\le k,j-1)},\sfM_{< k})}\left[\Ent(X_j)-\Ent(X_j|\sfM_{(\le k,j-1)}=m'_{(\le k,j-1)}, \sfM_{< k}=m'_{<k})\right]}\\
&\;\;\;\;\;=\sqrt{\I\left(X_j;\sfM_{(\le k,j-1)}, \sfM_{< k}\right)}.
\end{align*}
Substituting in Equation \eqref{eq:clacomb1} and summing over $j$, we get that 
\begin{align*}
\bbE\left[\sum_{j=1}^n\one_{Y_j\neq X_j'}\right]&\le \sum_{j=1}^n\sqrt{\I\left(X_j;\sfM_{(\le k,j-1)}, \sfM_{< k}\right)}\\
&\le \sqrt{n\cdot \sum_{j=1}^n \I\left(X_j;\sfM_{(\le k,j-1)}, \sfM_{< k}\right)}. \tag{Using Cauchy–Schwarz inequality}
\end{align*}
Next, we show that $\sum_{j=1}^n \I\left(X_j;\sfM_{(\le k,j-1)}, \sfM_{< k}\right)\le \Ent(\sfM_{<k})$, which proves the claim. By the chain rule, 
\begin{align*}
&\sum_{j=1}^n \I\left(X_j;\sfM_{(\le k,j-1)}, \sfM_{< k}\right)\\
&=\sum_{j=1}^n\sum_{i=0}^{k-1} \I\left(X_j;\sfM_{(i+1,j-1)}, \sfM_{i}\middle| \sfM_{(\le i,j-1)}, \sfM_{<i}\right)\\
&=\sum_{i=0}^{k-1} \sum_{j=1}^n \left(\I\left(X_j; \sfM_{i}\middle| \sfM_{(\le i,j-1)}, \sfM_{<i}\right)+\sum_{i=1}^k \I\left(X_j;\sfM_{(i+1,j-1)}\middle| \sfM_{(\le i,j-1)}, \sfM_{\le i}\right)\right)\\
&=\sum_{i=0}^{k-1} \sum_{j=1}^n \I\left(X_j; \sfM_{i}\middle| \sfM_{(\le i,j-1)}, \sfM_{<i}\right) \tag{$\I\left(X_j;\sfM_{(i+1,j-1)}\middle| \sfM_{(\le i,j-1)}, \sfM_{\le i}\right)=0$ using Corollary \ref{cor:indep2}}\\
&\le \sum_{i=0}^{k-1} \Ent(\sfM_i| \sfM_{<i}) \tag{proved below}\\
&=\Ent(\sfM_{<k}). \tag{Chain rule}
\end{align*}
Recall that $\sfM_i$ is the end memory state of the $i$th pass. For $i=0$, $\I\left(X_j; \sfM_{i}\middle| \sfM_{(\le i,j-1)}, \sfM_{<i}\right)$ can be written as $\I\left(X_j; \sfM_{0}\right)$, which is 0 for all $j\in[n]$. Next, we prove that 
\[\forall i\in[k-1],\;\sum_{j=1}^n \I\left(X_j; \sfM_{i}\middle| \sfM_{(\le i,j-1)}, \sfM_{<i}\right)\le \Ent(\sfM_i|\sfM_{<i}),\] as follows
\begin{align*}
&\sum_{j=1}^n \I\left(X_j; \sfM_{i}\middle| \sfM_{(\le i,j-1)}, \sfM_{<i}\right)\\
&\le \sum_{j=1}^n \I\left(X_j; \sfM_{i}, X_{[1,j-1]}, \sfM_{(\le i,[0,j-2])}\middle| \sfM_{(\le i,j-1)}, \sfM_{<i}\right)\\
&=\sum_{j=1}^n \I\left(X_j; X_{[1,j-1]}, \sfM_{(\le i,[0,j-2])}\middle| \sfM_{(\le i,j-1)}, \sfM_{<i}\right)+\sum_{j=1}^n \I\left(X_j; \sfM_{i}\middle| \sfM_{(\le i,j-1)}, \sfM_{<i}, X_{[1,j-1]}, \sfM_{(\le i,[0,j-2])} \right)\\
&=\sum_{j=1}^n \I\left(X_j; \sfM_{i}\middle| \sfM_{(\le i,j-1)}, \sfM_{<i}, X_{[1,j-1]}, \sfM_{(\le i,[0,j-2])}\right)\tag{using Corollary \ref{cor:indep1}}\\
&\le \sum_{j=1}^n \I\left(X_j, \sfM_{(\le i,j-1)} ; \sfM_{i}\middle| \sfM_{<i}, X_{[1,j-1]}, \sfM_{(\le i,[0,j-2])}\right)\\
&=\I(X_{[1,n]}, \sfM_{(\le i,[0,n-1])},\sfM_i|\sfM_{<i})\tag{Chain rule}\\
&\le \Ent(\sfM_i|\sfM_{<i}).&\qedhere
\end{align*}

\begin{proof}[\textbf{Proof of Lemma \ref{cl:ainf}}]
Recall that input $Y$ to Algorithm $\sfO$ (Algorithm \ref{al:a} with $\gamma=\frac{\ve}{10}$ and $\B=\sqrt{k\cdot n^{1+\lambda}}$) is drawn from the uniform distribution on $\zo^n$. Let $\sfO_j$ ($j\in\{0,\ldots,n\}$) represent the random variable for the $j$ memory state of Algorithm $\sfO$. According to Step \ref{st:a3} of Algorithm \ref{al:a}, $\sfO_j=(\Imi_j,\Apr_j)$. Here, $\Imi_j$ represents the random variable for the $j$th memory state of algorithm \ref{al:im} ($\Imi$) when run on input $Y$, and $\Apr_j$ represents the random variable for the $j$th memory state of algorithm \ref{al:apr} ($\Apr$), when run on input $Y-X'$, with parameters $\gamma=\frac{\ve}{10}$ and $\B=\sqrt{k\cdot n^{1+\lambda}}$. With these parameters, Claim \ref{cl:acomb} implies that input $Y-X'$ to $\Apr$ satisfies conditions of Proposition \ref{cl:aprinf}. Thus, using claim \ref{cl:aprinf}, we get 
\begin{equation}
\label{eq:clainf1}
\forall j\in \{0,1,\ldots,n\}, \;\Ent(\Apr_j) \le 40+6\log \log n +2\log{\left(\frac{\B}{\gamma \sqrt{n}}\right)}\le 50+6\log \log n +\log{\left(\frac{k\cdot n^{\lambda}}{\ve^2}\right)}.
\end{equation}
We analyze information cost of Algorithm $\sfO$ as follows:
\begin{align}
\nonumber
IC(\sfO)&=\sum_{j=1}^n \sum_{\ell=1}^j\I\left(\sfO_j;Y_\ell\middle|\sfO_{\ell-1}\right)\\
\nonumber
&=\sum_{j=1}^n \sum_{\ell=1}^j\I\left(\Imi_j,\Apr_j;Y_\ell\middle|\Imi_{\ell-1},\Apr_{\ell-1}\right)\\
\label{eq:clainf2}
&=\sum_{j=1}^n \sum_{\ell=1}^j\I\left(\Imi_j;Y_\ell\middle|\Imi_{\ell-1},\Apr_{\ell-1}\right)+\sum_{j=1}^n \sum_{\ell=1}^j\I\left(\Apr_j;Y_\ell\middle|\Imi_{\ell-1},\Apr_{\ell-1},\Imi_j\right)
\end{align}
We bound these two quantities separately. We first prove that $\I\left(\Imi_j;Y_\ell\middle|\Imi_{\ell-1},\Apr_{\ell-1}\right)\le \I\left(\Imi_j;Y_\ell\middle|\Imi_{\ell-1}\right)$ for all $1\le \ell\le j\le n$. This can be shown by proving that 
\[\I\left(\Imi_j;\Apr_{\ell-1}\middle|\Imi_{\ell-1},Y_\ell\right)=0.\] Recall that $\Ra^{\Apr}_{j}$ represents the private randomness used by Algorithm \ref{al:apr} ($\Apr$) at the $j$th time-step, which is independent of variables $Y, X'$ and $\{\Imi_{j}\}_{j\in\{0,\ldots,n\}}$.
\begin{align*}
&\I\left(\Imi_j;\Apr_{\ell-1}\middle|\Imi_{\ell-1},Y_\ell\right)\\
&\le \I\left(\Imi_j;\Apr_{\ell-1},Y_{[1,\ell-1]},X'_{[1,\ell-1]},\Ra^{\Apr}_{[1,\ell-1]}\middle|\Imi_{\ell-1},Y_\ell\right)\\
&=\I\left(\Imi_j;Y_{[1,\ell-1]},X'_{[1,\ell-1]},\Ra^{\Apr}_{[1,\ell-1]}\middle|\Imi_{\ell-1},Y_\ell\right)+\I\left(\Imi_j;\Apr_{\ell-1}\middle|\Imi_{\ell-1},Y_\ell,Y_{[1,\ell-1]},X'_{[1,\ell-1]},\Ra^{\Apr}_{[1,\ell-1]}\right)\\
&=\I\left(\Imi_j;Y_{[1,\ell-1]},X'_{[1,\ell-1]},\Ra^{\Apr}_{[1,\ell-1]}\middle|\Imi_{\ell-1},Y_\ell\right)\tag{$\Apr_{\ell-1}$ is deterministic function of $Y_{[1,\ell-1]},X'_{[1,\ell-1]}$ and $\Ra^{\Apr}_{[1,\ell-1]}$}\\
&=\I\left(\Imi_j;Y_{[1,\ell-1]},X'_{[1,\ell-1]}\middle|\Imi_{\ell-1},Y_\ell\right)+\I\left(\Imi_j;\Ra^{\Apr}_{[1,\ell-1]}\middle|\Imi_{\ell-1},Y_\ell,Y_{[1,\ell-1]},X'_{[1,\ell-1]}\right)\\
&=0.
\end{align*}
The last equality follows from the fact that 1) $\Imi_j$ is independent of the input and the inner workings of the first $\ell-1$ steps conditioned on $\Imi_{\ell-1}$, and 2) $\Apr$'s private randomness $\Ra^\Apr_{[1,\ell-1]}$ is independent of $Y,X'$ and memory states of $\Imi$. Therefore, we can write the first quantity in Expression \eqref{eq:clainf2} as
\begin{align}
\nonumber
\sum_{j=1}^n \sum_{\ell=1}^j\I\left(\Imi_j;Y_\ell\middle|\Imi_{\ell-1},\Apr_{\ell-1}\right)&\le \sum_{j=1}^n \sum_{\ell=1}^j\I\left(\Imi_j;Y_\ell\middle|\Imi_{\ell-1}\right)\\
\nonumber
&=IC(\Imi)\\
\label{eq:clainf3}
&\le \miccoin(\sfM).
\end{align}
The last inequality follows from Claim \ref{cl:iminf}. Next, we analyze the second quantity in Expression \eqref{eq:clainf2} for a given $j\in[n]$, as follows:
\begin{align*}
&\sum_{\ell=1}^j\I\left(\Apr_j;Y_\ell\middle|\Imi_{\ell-1},\Apr_{\ell-1},\Imi_j\right)\\
&\le \sum_{\ell=1}^j\I\left(\Apr_j;Y_\ell,\Imi_{[0,\ell-2]},\Apr_{[0,\ell-2]},Y_{[1,\ell-1]}\middle|\Imi_{\ell-1},\Apr_{\ell-1},\Imi_j\right)\\
&=\sum_{\ell=1}^j\I\left(\Apr_j;\Imi_{[0,\ell-2]},\Apr_{[0,\ell-2]},Y_{[1,\ell-1]}\middle|\Imi_{\ell-1},\Apr_{\ell-1},\Imi_j\right)\\
&\;\;\;\;\;\;\;\;\;+\sum_{\ell=1}^j\I\left(\Apr_j;Y_\ell\middle|\Imi_{\ell-1},\Apr_{\ell-1},\Imi_{[0,\ell-2]},\Apr_{[0,\ell-2]},Y_{[1,\ell-1]},\Imi_j\right)\\
&=\sum_{\ell=1}^j\I\left(\Apr_j;Y_\ell\middle|\Imi_{\ell-1},\Apr_{\ell-1},\Imi_{[0,\ell-2]},\Apr_{[0,\ell-2]},Y_{[1,\ell-1]},\Imi_j\right)\tag{explained below}\\
&\le\sum_{\ell=1}^j\I\left(\Apr_j;\Imi_{\ell-1},\Apr_{\ell-1},Y_\ell\middle|\Imi_{[0,\ell-2]},\Apr_{[0,\ell-2]},Y_{[1,\ell-1]},\Imi_j\right)\\
&=\I\left(\Apr_j;\Imi_{[0,j-1]},\Apr_{[0,j-1]},Y_{[1,j]}\middle|\Imi_j\right)\tag{Chain rule}\\
&\le \Ent(\Apr_j).
\end{align*}
The second equality follows from the fact that, for a single-pass algorithm $\sfO$, $$\I\left(\sfO_j;\sfO_{[0,\ell-2]},Y_{[1,\ell-1]}\middle|\sfO_{\ell-1}\right)=0,$$ which implies that 
\[\I\left(\Apr_j, \Imi_j;\Apr_{[0,\ell-2]},\Imi_{[0,\ell-2]},Y_{[1,\ell-1]}\middle|\Apr_{\ell-1},\Imi_{\ell-1}\right)=0.\]
Summing over $j\in[n]$ and using Equation \eqref{eq:clainf1}, we get 
\begin{align}
\label{eq:clainf4}
\sum_{j=1}^n \sum_{\ell=1}^j\I\left(\Apr_j;Y_\ell\middle|\Imi_{\ell-1},\Apr_{\ell-1},\Imi_j\right)\le n\cdot \left(50+6\log \log n +\log{\left(\frac{k\cdot n^{\lambda}}{\ve}\right)}\right).
\end{align}
The claim follows from substituting Equations~\eqref{eq:clainf3} and \eqref{eq:clainf4} in Expression \eqref{eq:clainf2}.
\end{proof}

\begin{proof}[\textbf{Proof of Lemma \ref{cl:asum}}]
Note that as $Y$ is drawn from uniform distribution over $\zo^n$, we have $\bbE\left[\left(\sum_{j=1}^nY_j\right)^2\right]=n$. Therefore, 
\begin{align*}
\bbE_{\sfO_n}\left[\Var\left(\sum_{j=1}^n Y_j\middle| \sfO_n=\an_n\right)\right]&=\bbE_{\sfO_n}\left[\bbE\left(\sum_{j=1}^n Y_j\right)^2\middle| \sfO_n=\an_n\right]-\bbE_{\sfO_n}\left(\bbE\left[\sum_{j=1}^n Y_j\middle| \sfO_n=\an_n\right]\right)^2\\
&=n-\bbE_{\sfO_n}\left(\bbE\left[\sum_{j=1}^n Y_j\middle| \sfO_n=\an_n\right]\right)^2.
\end{align*}
We will prove that an upper bound of $\left(1-\frac{\ve}{2}\right)n$ on the expected variance of $\sum Y_j$ conditioned on $\sfO_n$ to prove the claim. As $X$ is also distributed uniformly over $\zo^n$, 
\begin{equation}
\nonumber
\bbE_{\sfM_{(k,n)}}\left(\bbE\left[\sum_{j=1}^n X_j\middle| \sfM_{(k,n)}=m_{(k,n)}\right]\right)^2\ge \ve n\implies \bbE_{\sfM_{(k,n)}}\left[\Var\left(\sum_{j=1}^n X_j\middle| \sfM_{(k,n)}=m_{(k,n)}\right)\right]\le n\cdot (1-\ve). 
\end{equation}
As $X,\sfM_{(k,n)}$ are identically distributed to $X',\sfM'_{(k,n)}$ (Claim \ref{cl:imd}), we get 
\begin{equation}
\label{eq:clasum1}
\bbE_{\sfM'_{(k,n)}}\left[\Var\left(\sum_{j=1}^n X'_j\middle| \sfM'_{(k,n)}=m'_{(k,n)}\right)\right]\le n\cdot (1-\ve). 
\end{equation}
Recall that $\Ee_j=Y_j-X'_j$ represents the random variable for the $j$th input element to subroutine $\Apr$ in Algorithm $\sfO$, $\Apr_n$ represents the random variable for the output of $\Apr$ and  $\Ra^{\Apr}_{j}$ represents the private randomness used at the $j$th time-step. Note that random variable $\Ee$ has joint distribution $\calD$ on $(\Ee_1,\ldots, \Ee_n)$. With parameters $\gamma=\frac{\ve}{10}$ and $\B=\sqrt{k\cdot n^{1+\lambda}}$, input $\Ee$ to algorithm $\Apr$ satisfies conditions of Proposition \ref{cl:aprsum}, which implies that
\begin{equation}
\label{eq:clasum2}
\bbE\left[\left(\sum_{j=1}^n\Ee_j-\Apr_n\right)^2\right]=\bbE_{\e\sim \calD,r\sim\Ra^\Apr}\left[\left(\sum_{j=1}^n\e_j-\Apr(\e ,r)\right)^2\right]\le \gamma^2 n.
\end{equation}
We upper bound the expected variance of $\sum Y_j$ conditioned on $\sfO_n$ as follows:
\begin{align*}
&\bbE_{\sfO_n}\left[\Var\left(\sum_{j=1}^n Y_j\middle| \sfO_n=\an_n\right)\right]\\
&=\bbE_{\sfO_n}\left[\Var\left(\sum_{j=1}^n Y_j-\sum_{j=1}^n X'_j + \sum_{j=1}^n X'_j\middle| \sfO_n=\an_n\right)\right]\\
&\le \bbE_{\sfO_n}\left[\Var\left(\sum_{j=1}^n Y_j-\sum_{j=1}^n X'_j\middle| \sfO_n=\an_n\right)\right]
+\bbE_{\sfO_n}\left[\Var\left(\sum_{j=1}^n X'_j\middle| \sfO_n=\an_n\right)\right]\\
&\;\;\;\;\;\;\;+2\bbE_{\sfO_n}\sqrt{\Var\left(\sum_{j=1}^n Y_j-\sum_{j=1}^n X'_j \middle| \sfO_n=\an_n\right)\cdot\Var\left(\sum_{j=1}^n X'_j\middle| \sfO_n=\an_n\right)}\\
&\le \bbE_{\sfO_n}\left[\Var\left(\sum_{j=1}^n Y_j-\sum_{j=1}^n X'_j\middle| \sfO_n=\an_n\right)\right]
+\bbE_{\sfO_n}\left[\Var\left(\sum_{j=1}^n X'_j\middle| \sfO_n=\an_n\right)\right]\\
&\;\;\;\;\;\;\;+2\sqrt{\bbE_{\sfO_n}\left[\Var\left(\sum_{j=1}^n Y_j-\sum_{j=1}^n X'_j \middle| \sfO_n=\an_n\right)\right]\cdot\bbE_{\sfO_n}\left[\Var\left(\sum_{j=1}^n X'_j\middle| \sfO_n=\an_n\right)\right]}.\tag{using Cauchy-Schwarz inequality}
\end{align*}
We prove that \[\bbE_{\sfO_n}\left[\Var\left(\sum_{j=1}^n X'_j\middle| \sfO_n=\an_n\right)\right]\le n(1-\ve),\] and 
\[\bbE_{\sfO_n}\left[\Var\left(\sum_{j=1}^n Y_j-\sum_{j=1}^n X'_j\middle| \sfO_n=\an_n\right)\right]\le \gamma^2 n=\frac{\ve^2}{100}n,\]
which implies that 
\[\bbE_{\sfO_n}\left[\Var\left(\sum_{j=1}^n Y_j\middle| \sfO_n=\an_n\right)\right]\le n(1-\ve)+n\cdot\frac{\ve^2}{100}+ 2n\cdot\frac{\ve}{10}\le n\left(1-\frac{\ve}{2}\right).\]
As $\sfO_n=(\Apr_n,\Imi_n)=(\Apr_n, \sfM'_{<k}, \sfM'_{(\le k,n)})$, we can rewrite the variances as 
\begin{align*}
&\bbE_{\sfO_n}\left[\Var\left(\sum_{j=1}^n X'_j\middle| \sfO_n=\an_n\right)\right]\\
&=\bbE_{\Apr_n,\sfM'_{<k}, \sfM'_{(\le k,n)}}\left[\Var\left(\sum_{j=1}^n X'_j\middle| \Apr_n=apr_n,\sfM'_{<k}=m'_{<k}, \sfM'_{(\le k,n)}=m'_{(\le k,n)}\right)\right]\\
&\le\bbE_{\sfM'_{(k,n)}}\left[\Var\left(\sum_{j=1}^n X'_j\middle| \sfM'_{(k,n)}=m'_{(k,n)}\right)\right]\tag{by law of total variance}\\
&\le n\cdot(1-\ve). \tag{using Equation~\eqref{eq:clasum1}}
\end{align*}
Next, we prove the bound on variance of $\left(\sum_{j=1}^n Y_j-\sum_{j=1}^n X'_j\right)$ as follows: 
\begin{align*}
 &\bbE_{\sfO_n}\left[\Var\left(\sum_{j=1}^n Y_j-\sum_{j=1}^n X'_j\middle| \sfO_n=\an_n\right)\right]\\
 &= \bbE_{\Apr_n,\Imi_n}\left[\Var\left(\sum_{j=1}^n Y_j-\sum_{j=1}^n X'_j\middle| \Apr_n=apr_n,\Imi_n=im_n\right)\right]\\
 &=\bbE_{\Apr_n,\Imi_n}\left[\Var\left(\sum_{j=1}^n Y_j-\sum_{j=1}^n X'_j-\Apr_n\middle| \Apr_n=apr_n,\Imi_n=im_n\right)\right]\tag{adding a constant to variance}\\
 &\le \bbE\left[\left(\sum_{j=1}^n Y_j-\sum_{j=1}^n X'_j-\Apr_n\right)^2\right]\\
&=\bbE\left[\left(\sum_{j=1}^n \Ee_j - \Apr_n\right)^2\right] \tag{$(\Ee_1,\Ee_2,\ldots,\Ee_n)$ is the input to algorithm $\Apr$}\\
&\le \gamma^2 n. \tag{using Equation \eqref{eq:clasum2}}
\end{align*}
This finishes the proof of the claim. 
\end{proof}

\subsection{Proof of Theorem \ref{thm:tcoinsvar}}\label{ap:multiplecoin}
Let $\sfM$ be a $k$-pass algorithm on a stream of $nt$ updates of the form $(X_j, s_j), j\in[nt]$ where $X_j$s are $i.i.d.$ uniform $\zo$ bits, and $\{s_j\}_{j\in[nt]}$ is a good order. Let $k,t\le n^\lam$, $\Ent(\sfM_i)\le n^\lam$, $\forall i\in \{0,\ldots, k\}$ ($\lam>0$ to be decided later) and for at least $1/2$ fraction of $s\in[t]$,
\begin{equation*}
		\bbE_{\sfM_{(k,nt)}}\left[\bbE\left[ \left( \sum_{j\in J_s} X_j \right) ~ \middle| ~
		\sfM_{(k,nt)}=m_{(k,nt)}\right]^2  \right] > \ve |J_s|.
		\end{equation*}
  We prove Theorem \ref{thm:tcoinsvar} similarly to the proof of Theorem \ref{th:main} by constructing a single-pass algorithm $\sfO$ for the $\tcoins$ using $\sfM$ and then using Theorem \ref{thmbgw:tcoins}. Using $\sfM$, we construct a single pass algorithm $\sfO$ (Algorithm \ref{al:at}) that uses private randomness such that, given a stream of $nt$ updates of the form $(Y_j,s_j)$ where $Y_j$s $i.i.d.$ uniform $\zo$ bits, the following holds
\begin{enumerate}
\item for the information cost of $\sfO$:
\begin{equation}\label{eq:tcoins1}
IC(\sfO)\; = \;\sum_{j=1}^{nt} \sum_{\ell=1}^ \I(\sfO_{j};Y_\ell| \sfO_{\ell-1})\;\le \; \miccoin(\sfM) + nt \cdot t\left(50+6\log \log {nt} +\log{\left(\frac{kt\cdot n^{\lambda}}{\ve^2}\right)}\right).
\end{equation}
Here, $\sfO_j$ denotes the random variable for the memory state of $\sfO$ after reading $j$ updates.

\item for the output of algorithm $\sfO$: for at least a $1/2$ fraction of $s\in[t]$, \begin{align*}
\bbE_{\sfO_{nt}}\left[\left(\bbE\left[ \sum_{j\in J_s} Y_j ~ \middle| ~
\sfO_{nt} =\an_{nt}\right]\right)^2  \right] \;\ge\; \frac{\ve}{2}\cdot |J_s|. 
\end{align*}
\end{enumerate}
Theorem \ref{thmbgw:tcoins} implies that there exists $\de_\ve>0$ such that $IC(\sfO)\ge  \de_\ve \cdot nt^2\log n$. If $\lam<\frac{\de_\ve}{10}$, then for sufficiently large $n$, 
\[50+6\log \log {nt} +\log{\left(\frac{kt\cdot n^{\lambda}}{\ve^2}\right)}\le 50+6(1+\lam)\log \log n +\log{\left(\frac{1}{\ve^2}\right)}+3\lam\log{n}< \frac{\de_\ve}{2} \log n.\]
Therefore, Equation \eqref{eq:tcoins1} implies that $\miccoin(\sfM)\ge IC(\sfO)-\frac{\de_\ve}{2}\cdot nt^2 \log n\ge \frac{\de_\ve}{2}\cdot nt^2 \log n$. Taking $\de=\de_\ve/2$ proves the theorem. 

Informally, similar to the proof of Theorem \ref{th:main}, $\sfO$ runs $k$ passes of $\sfM$ in parallel. Before reading the input stream $(Y_1,s_1),\ldots, (Y_{nt},s_{nt})$, $\sfO$ samples memory states at the end of first $k-1$ passes from the joint distribution on $(\sfM_0,\ldots,\sfM_{k-1})$. $\sfO$ then modifies the given input $Y$ to $X'$ such that the parallel execution of the $k-1$ passes of the algorithm $\sfM$ on $(X'_1,s_1),\ldots, (X'_{nt},s_{nt})$ end in the sampled memory states. $\sfO$ also maintains an approximation for the modification for all instances, that is of $\sum_{j\in J_s} (X'_j-Y_j)$, $\forall s\in[t]$; this helps $\sfO$ to compute $\sum_{j\in J_s} Y_j$ as long as $\sfM$ computes $\sum_{j\in J_s} X'_j$ after $k$ passes. As before, we want $\sfO$ to have comparable information cost to that of $\sfM$ and we maintain the approximation of the modification using Algorithm \ref{al:apr}. As in the proof of Theorem \ref{th:main}, we choose parameters $\gamma=\frac{\ve}{10}$ and $\B=\sqrt{kt\cdot n^{1+\lam}}$. For each $s\in[t]$, we run a separate copy of the approximation algorithm, which we denote by $\Apr^s$. We describe the one-pass algorithm $\sfO$ formally in Algorithm \ref{al:at}. For imitating $k$ passes of $\sfM$, Algorithm \ref{al:at} executes an Algorithm \ref{al:im} style subroutine that we denote by $\Imi$ -- $\Imi$ modifies input bit $y_j$ at the $j$th time-step to bit $x'_j$, given $im_{j-1}$.  In parallel, $\sfO$ runs $t$ copies of Algorithm \ref{al:apr} on the modification, where $(y_j-x'_j)\in \{-1,0,1\}$ is the next input element to the $s_j$th copy, that is, $\Apr^{s_j}$. At every time-step, $\sfO$ stores latest memory states of algorithms $\Imi$ and $\Apr^s$, $\forall s\in[t]$. We can assume that $\sfO$ stores the order $\{s_j\}_{j\in[nt]}$ throughout the algorithm, which would not affect the information cost as the order is determinisitic. While describing $\sfO$ formally in Algorithm \ref{al:at}, we will use Algorithm \ref{al:apr} as a black-box. It would be useful to define the inverse of function $q_s\colon [|J_s|]\rightarrow [nt]$, where recall that $q_s(u)=j$ if $y_j$ is the $u$-th element corresponding to the $s$th instance of the coin problem ($s_j=s$). 
\[q^{-1}_s(j)=\begin{cases}
0 \text{ if } j=0\\
q^{-1}_s(j-1)+1\text{ if }s_j=s \\
q^{-1}_s(j-1)\text{ if }s_j\neq s.
\end{cases}\]
Thus, $q^{-1}_s(j)$ computes the number of inputs that have been seen for the $s$th instance after $j$ stream updates. We denote the input stream to $\Apr^s$ by $\e^s$. At the $j$th step, we only run the next step for the algorithm $\Apr^{s_j}$ -- feed in its $q^{-1}_{s_j}(j)$th input element, that is, $\e^{s_j}_{q^{-1}_{s_j}(j)}=y_j-x'_j$. 
%As used in Subsection \ref{subsec:apr}, $\Ra_j^\Apr$ represents the private randomness used by algorithm $\Apr$ at the $j$th time-step and $\Apr_j$ represents the random variable for $j$th memory state ($r_j$ and $apr_j$ represent their instantiations).  Let $f^{\Apr}_j$ $(j\in[n])$ represent the $j$th transition function for algorithm $\Apr$, that is, $apr_j=f^{\Apr}_j(apr_{j-1},\e_j,r_j)$. 
Let $\Apr_j$ represent the random variable for the memory states of $t$ copies of Algorithm \ref{al:apr} after $\sfO$ reads $j$ stream updates. Note that, as the $s$th copy has seen $q^{-1}_s(j)$ updates till the $j$th time-step, $\Apr_j=\left\{\Apr^{s}_{q^{-1}_s(j)}\right\}_{s\in[t]}$ where $\Apr^{s}_{u}$ is the random variable for $u$th memory state of $s$th copy of Algorithm \ref{al:apr}.
As in Subsection \ref{subsec:a}, let $\Apr^s_u(apr^s_{u-1},\e^s_u)$ denote the random variable for the $u$th memory state, when the $u$th input element is $\e^s_u$, and $(u-1)$th memory state is $apr^s_{u-1}$. 

Let $\Imi_j$ ($j\in\{0,\ldots,nt\}$) denote the random variable associated with values $im_j$ (defined at Step \ref{st:at7} of Algorithm \ref{al:at}). Note that, $\Imi$ can be seen as a one-pass algorithm on input sequence $(Y_j,s_j)_{j\in[nt]}$. As in Subsection \ref{subsec:im}, let $\sfM'_i$ denote the random variable associated with value $m'_i$ ($i\in\{0,1,\ldots, k-1)$. The distribution of $\sfM'_{<k}$ is defined at Step \ref{st:at1} of Algorithm \ref{al:at}. Let $\{\sfM'_{(i,j)}\}_{i\in[k],j\in\{0,\ldots,nt\}}$ denote the random variables associated with values $\{m'_{(i,j)}\}_{i\in[k],j\in\{0,\ldots,nt\}}$. The distribution of $\sfM'_{(i,j)}$ ($j\in[nt]$) is defined at Step \ref{st:at5} of Algorithm \ref{al:at}, and of $\sfM'_{(i,0)}$ is defined at Step \ref{st:at2}. Let $\{X'_j\}_{j\in[nt]}$ denote the random variable for value $x'_j$ in Step \ref{st:at3} of Algorithm \ref{al:at}.  These distributions depend on the joint distribution of $(X,\sfM_{\le k}, \sfM_{(\le k, [1,nt])})$ and the uniform distribution on $Y$. Let $\calD^s$ ($s\in[t]$) be the joint distribution generated by Algorithm \ref{al:at} on inputs to $\Apr^s$ that is, the joint distribution on $\left(Y_{q_s(1)}-X'_{q_s(1)},Y_{q_s(2)}-X'_{q_s(2)},\ldots, Y_{q_s(|J_s|)}-X'_{q_s(|J_s|)}\right)$. Let $\Ee^s_u$ be the random variable for the $u$th input element to $\Apr^s$, that is, $\Ee^s_u=Y_{q_s(u)}-X'_{q_s(u)}$. 

\begin{algorithm}[H]
\caption{Single pass algorithm $\sfO$ using $k$-pass algorithm $\sfM$ for $\tcoins$}
\label{al:at}
\textbf{Input}: a stream of $nt$ updates of the form $(y_j,s_j)$ where $y_j$s are drawn independently and uniformly from $\zo$ and $\{s_j\}_{j\in[nt]}$ is a good order\\
\textbf{Goal}: approximate $\sum_{j\in J_s} y_j$ for all $s\in[t]$.
\begin{algorithmic}[1]
\STATE  \label{st:at1}Sample $m'_0,m'_1,\ldots,m'_{k-1}\sim (\sfM_0,\sfM_1,\ldots,\sfM_{k-1})$
\COMMENT{Sample memory states for the end of first $k-1$ passes}
%\STATE $\imi_0\leftarrow (m'_0,m'_1,\ldots,m'_k)$
\STATE $\imi_0\leftarrow (m'_0,m'_1,\ldots,m'_{k-1})$
\STATE \label{st:at2}$\forall i\in[k]$, $m'_{(i,0)}\leftarrow m'_{(i-1)}$  
\COMMENT{Starting memory states for the $k$ passes of $\sfM$}
\STATE Initialize $t$ copies of Algorithm \ref{al:apr}: $\Apr^s$, ($s\in [t]$) with parameters $\gamma=\frac{\ve}{10}$, $\B=\sqrt{kt\cdot n^{1+\lam}}$ and length of input stream being $|J_s|$. 
\STATE For all $s\in[t]$, sample $apr^s_0\sim \Apr^s_0$ \COMMENT{For Algorithm \ref{al:apr}, the starting state is deterministic}
\STATE $apr_0=\{apr^s_0\}_{s\in[t]}$. 
\FOR {$j=1$ to $nt$} 
%\STATE \label{st:im1} $\beta_j\leftarrow \left(\Pr\left[X_j=1\;|\; \sfM_{(\le k,j-1)}=m'_{(\le k,j-1)}, \sfM_{\le k}=m'_{\le k}\right]-\frac{1}{2}\right)$ \COMMENT{Can be calculated using $im_{j-1}$}
\STATE \label{st:at3} $\beta_j\leftarrow \left(\Pr\left[X_j=1\;|\; \sfM_{(\le k,j-1)}=m'_{(\le k,j-1)}, \sfM_{< k}=m'_{<k}\right]-\frac{1}{2}\right)$ \COMMENT{Can be calculated using $im_{j-1}$}
\IF {$\beta_j \ge 0$} 
\IF{$y_j=1$}
 \STATE $x_j'\leftarrow y_j$ 
 \ELSE
 \STATE $x_j'\leftarrow y_j$ with probability $1-2\beta_j$, and $x_j'\leftarrow 1$ otherwise
 \ENDIF
\ELSIF{$\beta_j \le 0$}
\IF{$y_j=1$}
 \STATE  $x_j'\leftarrow y_j$ with probability $1+2\beta_j$, and $x_j'\leftarrow -1$ otherwise
 \ELSE
 \STATE $x_j'\leftarrow y_j$
\ENDIF
\ENDIF \label{st:at4}
\STATE $\e^{s_j}_{q^{-1}_{s_j}(j)}\leftarrow (y_j-x_j')$\COMMENT{Setting the next input element to $\Apr^{s_j}$}
\STATE Sample $apr^{s_j}_{q^{-1}_{s_j}(j)}\sim \Apr^{s_j}_{q^{-1}_{s_j}(j)}\left(apr^{s_j}_{\left(q^{-1}_{s_j}(j)-1\right)}, \e^{s_j}_{q^{-1}_{s_j}(j)}\right)$ \\
\COMMENT{Run next step of $\Apr^{s_j}$ given $apr_{j-1}=\left\{apr^s_{q^{-1}_{s}(j-1)}\right\}_{s\in[t]}$, as $q^{-1}_{s_j}(j)-1=q^{-1}_{s_j}(j-1)$}
\STATE $apr_j\leftarrow$ ``$apr_{j-1}$ with $s_j$th coordinate updated to $apr^{s_j}_{q^{-1}_{s_j}(j)}$"
 %\STATE \label{st:im2} Sample $(m'_{(1,j)},m'_{(2,j)}, \ldots, m'_{(k,j)})$ from the joint distribution on \[\left((\sfM_{(1,j)},\sfM_{(2,j)},\ldots, \sfM_{(k,j)})~\middle|~(\sfM_{(\le k,j-1)}=m'_{(\le k,j-1)},\; \sfM_{\le k}=m'_{\le k},\; X_j=x_j'\right)\]\\
 \STATE  \label{st:at5} Sample $(m'_{(1,j)},m'_{(2,j)}, \ldots, m'_{(k,j)})$ from the joint distribution on \[\left((\sfM_{(1,j)},\sfM_{(2,j)},\ldots, \sfM_{(k,j)})~\middle|~(\sfM_{(\le k,j-1)}=m'_{(\le k,j-1)},\; \sfM_{< k}=m'_{< k},\; X_j=x_j')\right)\]\\
\COMMENT{Given $im_{(j-1)}$, execute $j$th time-step for all passes of $\sfM$ when the $j$th stream update is $(x_j',s_j)$}
 %\STATE $im_j\leftarrow (m'_{(1,j)},m'_{(2,j)}, \ldots, m'_{(k,j)},m'_0,m'_1,\ldots,m'_k)$\\
 \STATE \label{st:at7}$im_j\leftarrow (m'_{(1,j)},m'_{(2,j)}, \ldots, m'_{(k,j)},m'_0,m'_1,\ldots,m'_{k-1})$\\
\STATE \label{st:at6} $\an_j\leftarrow (im_j,apr_j)$\COMMENT{$j$th memory state of algorithm $\sfO$}
\ENDFOR \\
\textbf{Output}: $\an_n=(im_n, apr_n)=\left(m'_{(\le k,n)},m'_{<k}, \left\{apr^s_{|J_s|}\right\}_{s\in[t]}\right)$
%\RETURN $apr_n+\bbE\left[\sum_{j=1}^n X'_j\middle| \sfM'_{(k,n)}=m'_{(k,n)}\right]$ \COMMENT{Can be calculated using $\an_n$; $im_n$ contains $m'_{(k,n)}$}
\end{algorithmic}
\end{algorithm}
\noindent
Similarly to Claim \ref{cl:imd}, we prove the following claim.
\begin{claim}\label{cl:tcoinsd}
The joint distribution on $X, \sfM_{< k}, \{\sfM_{(\le k, [0,nt])}\}$ is identical to that on $X',\sfM'_{< k}, \{\sfM'_{(\le k, [0,nt])}\}$.
\end{claim} 
\proof
The joint distribution on variables $X, \sfM_{< k}, \{\sfM_{(\le k, [0,nt])}\}$ is decided as the $k$-pass algorithm unfolds. We first sample $X$, then $\sfM_0=\sfM_{(1,0)}$ which is independent of $X$, then run the first pass on $(X_j,s_j)_{j\in[nt]}$, which identifies the variables $\sfM_{(1,[1,nt])}$ (where $\sfM_1=\sfM_{(1,nt)}=\sfM_{(2,0)}$), then we run the second pass and so on. 
However, to prove the equivalence of these joint distributions, we sample the variables for the $k$-passes in parallel. 
We prove the claim using the chain rule by conditioning on the following order (similarly for $\sfM'$ and $X'$):
\begin{enumerate}
\item $(\sfM_0,\sfM_1,\ldots,\sfM_{k-1})$.
\item $(\sfM_{(1,0)},\sfM_{(2,0)},\ldots,\sfM_{(k,0)})$.
\item Repeat Step \ref{it:tcoinsd1} to Step \ref{it:tcoinsd2} for $j=1$ to $nt$.
\item \label{it:tcoinsd1} $X_j$ conditioned on $\sfM_{< k}$, $X_{[1,j-1]}$ and $\sfM_{(\le k,[0,j-1])}$.
\item \label{it:tcoinsd2} $(\sfM_{(1,j)},\sfM_{(2,j)},\ldots, \sfM_{(k,j)})$ conditioned on $X_j$, $\sfM_{< k}$, $X_{[1,j-1]}$ and $\sfM_{(\le k,[0,j-1])}$.
\end{enumerate}
As $(m'_0,m'_1,\ldots,m'_{k-1})$ is drawn from the joint distribution on $(\sfM_0,\sfM_1,\ldots,\sfM_{k-1})$, by definition 
\[\Pr[\sfM'_{< k}=m'_{< k}]=\Pr[\sfM_{< k}=m'_{< k}], \; \forall \;m'_{< k}.\]
Similarly, by definitions of $\sfM'_{(\le k,0)}$ and $\sfM_{(\le k,0)}$, distributions of these random variables conditioned on $\sfM'_{<k}$ and $\sfM_{<k}$ respectively, are identical. 
Next, we look at the conditional distribution in Step \ref{it:tcoinsd1} for $j\in[nt]$. It is easy to see that, given a value of $\beta_j$ in Step \ref{st:at3} of Algorithm \ref{al:at}, $x_j'=1$ with probability $\frac{1}{2}+\beta_j$, as $Y_j\sim \Ber(1/2)$. Therefore $\forall \; m'_{< k}, m'_{([1,k],0)}=m'_{[0,k-1]}, x'_{[1,j-1]}, m'_{(\le k,[1,j-1])} $, 
\begin{align*}
&\Pr\left[X'_j=1\mid \sfM'_{< k}=m'_{< k}, X'_{[1,j-1]}=x'_{[1,j-1]},  \sfM'_{(\le k,[0,j-1])}=m'_{(\le k,[0,j-1])}\right]\\
&\;\;\;\;\;=\Pr\left[X_j=1\mid \sfM_{(\le k,j-1)}=m'_{(\le k,j-1)}, \sfM_{< k}=m'_{< k}\right]\\
&\;\;\;\;\;=\Pr\left[X_j=1\mid \sfM_{< k}=m'_{< k}, X_{[1,j-1]}=x'_{[1,j-1]},  \sfM_{(\le k,[0,j-1])}=m'_{(\le k,[0,j-1])}\right].
\end{align*}
The last equality follows from Corollary \ref{cor:indep1} (input $((X_1,s_1),\ldots,(X_{nt},s_{nt}))$ is drawn from a product distribution), which implies that
\[\left(X_j\indep X_{[1,j-1]}, \sfM_{(\le k, [0,j-2])}\middle| \sfM_{< k},\sfM_{(\le k,j-1)} \right).\]
Finally, we look at the conditional distribution in Step \ref{it:tcoinsd2}. As $(m'_{(1,j)},m'_{(2,j)}, \ldots, m'_{(k,j)})$ is drawn from a distribution that only depends on values $m'_{(\le k,j-1)}$, $m'_{< k}$ and $x_j'$ (see Step \ref{st:at5} of Algorithm \ref{al:im}), we have, $\forall\; m'_{< k}, m'_{([1,k],0)}=m'_{[0,k-1]}, x'_{[1,j]}, m'_{(\le k,[1,j])} $,  
\begin{align*}
&\Pr\left[\sfM'_{(\le k,j)}=m'_{(\le k,j)}\middle| X'_j=x'_j, \sfM'_{< k}=m'_{< k}, X'_{[1,j-1]}=x'_{[1,j-1]}, \sfM'_{(\le k,[0,j-1])}=m'_{(\le k,[0,j-1])} \right]\\
&\;\;\;\;\;=\Pr\left[\sfM_{(\le k,j)}=m'_{(\le k,j)}\mid \sfM_{(\le k,j-1)}=m'_{(\le k,j-1)},\; \sfM_{< k}=m'_{< k}, X_j=x_j'\right]\\
&\;\;\;\;\;=\Pr\left[\sfM_{(\le k,j)}=m'_{(\le k,j)}\mid X_j=x_j', \sfM_{< k}=m'_{< k}, X_{[1,j-1]}=x'_{[1,j-1]},  \sfM_{(\le k,[0,j-1])}=m'_{(\le k,[0,j-1])}\right]
\end{align*}
The last equality follows from Corollary \ref{cor:indep1}, which implies that
\begin{align*}
\left(\sfM_{(\le k, j)}, X_j\indep X_{[1,j-1]}, \sfM_{(\le k, [0,j-2])}\middle| \sfM_{< k},\sfM_{(\le k,j-1)} \right). &\qedhere
\end{align*}

The following claim about the one-pass subroutine $\Imi$ in Algorithm \ref{al:at} follows similarly as Claim \ref{cl:iminf}.
\begin{claim}\label{cl:tcoinsiminf}
The information cost of algorithm $\Imi$ is at most the information cost of $k$-pass algorithm $\sfM$, that is,
\[IC(\Imi)=\sum_{j=1}^{nt}\sum_{\ell=1}^j \I\left(\Imi_{j};Y_{\ell}|\Imi_{\ell-1}\right) \le \sum_{j=1}^{nt}\sum_{\ell=1}^j \I\left(\sfM_{(\le k,j)};X_{\ell}|\sfM_{< k},\sfM_{(\le k,\ell-1)}\right)=\miccoin(\sfM).\]
Here, $X$ and $Y$ are both drawn from uniform distribution on $\zo^{nt}$.

\end{claim}

Next, we prove the following claim for the input distributions $\calD^s, s\in[t]$ to $t$ copies of Algorithm \ref{al:apr} ($\Apr^s, s\in[t]$) used in Algorithm \ref{al:at}. 
\begin{claim}\label{cl:tcoinscomb}
For all $s\in [t]$,
\[\bbE_{\e^s\sim \calD^s}\left[\sum_{u=1}^{|J_s|}\one_{\e^s_u\neq 0}\right]=\bbE\left[\sum_{u=1}^{|J_s|}\one_{Y_{q_s(u)}\neq X'_{q_s(u)}}\right]\le\bbE\left[\sum_{j=1}^{nt}\one_{Y_{j}\neq X'_{j}}\right] \le \sqrt{nt\cdot \Ent(\sfM_{<k})}\le \sqrt{kt\cdot n^{1+\lam}}=\B.\]
\end{claim}
\proof
 At the $j$th time-step, $y_j-x_j'$ is added to the input stream of algorithm $\Apr^{s_j}$. Thus, for all $s\in[t]$, the input stream for algorithm $\Apr^s$ is the sequence $\{y_j-x_j'\}_{j\in J_s}$ and by definition, 
\[\bbE_{\e^s\sim \calD^s}\left[\sum_{u=1}^{|J_s|}\one_{\e^s_u\neq 0}\right]=\bbE\left[\sum_{u=1}^{|J_s|}\one_{Y_{q_s(u)}\neq X'_{q_s(u)}}\right]=\bbE\left[\sum_{j\in J_s}\one_{Y_{j}\neq X'_{j}}\right]<\bbE\left[\sum_{j=1}^{nt}\one_{Y_{j}\neq X'_{j}}\right]. \]
Similarly to the proof of Claim \ref{cl:acomb}, we prove that 
\begin{align*}
\bbE\left[\sum_{j=1}^{nt}\one_{Y_{j}\neq X'_{j}}\right]&\le \sum_{j=1}^{nt} \sqrt{\I(X_j; \sfM_{(\le k, j-1)},\sfM_{<k} )}\\
&\le \sqrt{nt\cdot \sum_{j=1}^{nt}\I(X_j; \sfM_{(\le k, j-1)},\sfM_{<k} )}.
\end{align*}
Again, as in the proof of Claim \ref{cl:acomb}, we show that $\sum_{j=1}^{nt}\I(X_j; \sfM_{(\le k, j-1)},\sfM_{<k} )\le \Ent(\sfM_{<k})$, which finishes the proof.
\qed

We are ready to state our claims regarding information cost and output of Algorithm $\sfO$. 
\begin{claim}\label{cl:tcoinsinf}
For all $\ve>\frac{100}{\sqrt{n}}$, $\lambda>0$, $IC(\sfO)\le \miccoin(\sfM)+nt^2\cdot \left(50+6\log \log nt +\log{\left(\frac{kt\cdot n^{\lambda}}{\ve^2}\right)}\right)$.
\end{claim}
\begin{proof}
Recall that input $Y$ to Algorithm $\sfO$ is drawn from the uniform distribution on $\zo^{nt}$. Let $\sfO_j$ ($j\in\{0,\ldots,nt\}$) represent the random variable for the $j$th memory state of Algorithm $\sfO$. According to Step \ref{st:at6} of Algorithm \ref{al:at}, $\sfO_j=(\Imi_j,\Apr_j)$. Here, $\Imi_j$ represents the random variable for the $j$th memory state of algorithm $\Imi$ when run on input $Y$, and $\Apr_j=\{\Apr^s_{q^{-1}_s(j)}\}_{s\in[t]}$. Here $\Apr^s_u$ represents the random variable for the $u$th memory state of algorithm \ref{al:apr} ($\Apr^s$), when run on input sequence $\left\{Y_{(q^s(u))}-X'_{(q^s(u))}\right\}_{u\in[|J_s|]}$, with parameters $\gamma=\frac{\ve}{10}$ and $\B=\sqrt{kt\cdot n^{1+\lambda}}$. With these parameters and the fact the $\frac{n}{2}\le |J_s|\le nt$, Claim \ref{cl:tcoinscomb} implies that input to $\Apr^s$ satisfies conditions of Proposition \ref{cl:aprinf}. Thus, using claim \ref{cl:aprinf}, we get for all $s\in[t]$,
\begin{equation}
\label{eq:tcoinsinf1}
\forall u\in \{0,1,\ldots,|J_s|\}, \;\Ent(\Apr^s_u) \le 40+6\log \log |J_s| +2\log{\left(\frac{\B}{\gamma \sqrt{|J_s|}}\right)}\le 50+6\log \log {nt} +\log{\left(\frac{kt\cdot n^{\lambda}}{\ve^2}\right)}.
\end{equation}
We analyze information cost of Algorithm $\sfO$ as we did in proof of Lemma \ref{cl:ainf}:
\begin{align}
\nonumber
IC(\sfO)&=\sum_{j=1}^{nt} \sum_{\ell=1}^j\I\left(\sfO_j;Y_\ell\middle|\sfO_{\ell-1}\right)\\
\nonumber
&=\sum_{j=1}^{nt} \sum_{\ell=1}^j\I\left(\Imi_j,\Apr_j;Y_\ell\middle|\Imi_{\ell-1},\Apr_{\ell-1}\right)\\
\label{eq:tcoinsinf2}
&=\sum_{j=1}^{nt} \sum_{\ell=1}^j\I\left(\Imi_j;Y_\ell\middle|\Imi_{\ell-1},\Apr_{\ell-1}\right)+\sum_{j=1}^{nt} \sum_{\ell=1}^j\I\left(\Apr_j;Y_\ell\middle|\Imi_{\ell-1},\Apr_{\ell-1},\Imi_j\right)
\end{align}
Exactly as in the proof of Lemma \ref{cl:ainf}, we show that 
\begin{align}
\nonumber
\sum_{j=1}^n \sum_{\ell=1}^j\I\left(\Imi_j;Y_\ell\middle|\Imi_{\ell-1},\Apr_{\ell-1}\right)&\le \sum_{j=1}^n \sum_{\ell=1}^j\I\left(\Imi_j;Y_\ell\middle|\Imi_{\ell-1}\right)\\
\nonumber
&=IC(\Imi)\\
\label{eq:tcoinsinf3}
&\le \miccoin(\sfM).
\end{align}
The last inequality follows from Claim \ref{cl:tcoinsiminf}. Next, we analyze the second quantity in Expression \eqref{eq:tcoinsinf2}, for a given $j\in[nt]$, exactly as done in the proof of Lemma \ref{cl:ainf} and show that:
\begin{align*}
\sum_{\ell=1}^j\I\left(\Apr_j;Y_\ell\middle|\Imi_{\ell-1},\Apr_{\ell-1},\Imi_j\right)
&\le \Ent(\Apr_j)\\
&\le \sum_{s\in[t]}\Ent\left(\Apr^s_{q^{-1}_s(j)}\right).
\end{align*}
Summing over $j\in[nt]$ and using Equation \eqref{eq:tcoinsinf1}, we get 
\begin{align}
\label{eq:tcoinsinf4}
\sum_{j=1}^{nt} \sum_{\ell=1}^j\I\left(\Apr_j;Y_\ell\middle|\Imi_{\ell-1},\Apr_{\ell-1},\Imi_j\right)\le nt\cdot t\left(50+6\log \log {nt} +\log{\left(\frac{kt\cdot n^{\lambda}}{\ve}\right)}\right).
\end{align}
The claim follows from substituting Equations~\eqref{eq:tcoinsinf3} and \eqref{eq:tcoinsinf4} in Expression \eqref{eq:tcoinsinf2}.
\end{proof}

\begin{claim}\label{cl:tcoinssum}
For all $\ve>\frac{100}{\sqrt{n}}$, $s\in[t]$, 
\[E_{\sfM_{(k,nt)}}\left(\bbE\left[\sum_{j\in J_s} X_j\middle| \sfM_{(k,nt)}=m_{(k,nt)}\right]\right)^2\ge \ve |J_s|\implies  \bbE_{\sfO_{nt}}\left(\bbE\left[\sum_{j\in J_s} Y_j\middle| \sfO_{nt}=\an_{nt}\right]\right)^2\ge \frac{\ve}{2} |J_s|.\]
\end{claim}
Claim \ref{cl:tcoinssum} immediately proves that for at least $1/2$ fraction of $s\in [t]$,
\[\bbE_{\sfO_{nt}}\left(\bbE\left[\sum_{j\in J_s} Y_j\middle| \sfO_{nt}=\an_{nt}\right]\right)^2\ge \frac{\ve}{2} |J_s|.\]

\begin{proof}
Fix $s\in[t]$. Note that as $Y$ is drawn from uniform distribution over $\zo^{nt}$, $\bbE\left[\left(\sum_{j\in J_s} Y_j\right)^2\right]=|J_s|$. Therefore, 
\begin{align*}
\bbE_{\sfO_{nt}}\left[\Var\left(\sum_{j\in J_s} Y_j\middle| \sfO_{nt}=\an_{nt}\right)\right]&=\bbE_{\sfO_{nt}}\left[\bbE\left(\sum_{j\in J_s} Y_j\right)^2\middle| \sfO_{nt}=\an_{nt}\right]-\bbE_{\sfO_{nt}}\left(\bbE\left[\sum_{j\in J_s} Y_j\middle| \sfO_{nt}=\an_{nt}\right]\right)^2\\
&=|J_s|-\bbE_{\sfO_{nt}}\left(\bbE\left[\sum_{j\in J_s} Y_j\middle| \sfO_{nt}=\an_{nt}\right]\right)^2.
\end{align*}
We will prove that an upper bound of $\left(1-\frac{\ve}{2}\right)|J_s|$ on the expected variance of $\sum_{j\in J_s} Y_j$ conditioned on $\sfO_{nt}$ to prove the claim. As $X$ is also distributed uniformly over $\zo^{nt}$, 
\begin{align*}
&\bbE_{\sfM_{(k,nt)}}\left(\bbE\left[\sum_{j\in J_s} X_j\middle| \sfM_{(k,nt)}=m_{(k,nt)}\right]\right)^2\ge \ve |J_s|\\
&\implies \bbE_{\sfM_{(k,nt)}}\left[\Var\left(\sum_{j\in J_s} X_j\middle| \sfM_{(k,nt)}=m_{(k,nt)}\right)\right]\le |J_s|\cdot (1-\ve). 
\end{align*}
As $X,\sfM_{(k,nt)}$ are identically distributed to $X',\sfM'_{(k,nt)}$ (Claim \ref{cl:tcoinsd}), we get 
\begin{align}
\nonumber
&\bbE_{\sfM_{(k,nt)}}\left(\bbE\left[\sum_{j\in J_s} X_j\middle| \sfM_{(k,nt)}=m_{(k,nt)}\right]\right)^2\ge \ve |J_s|\\
\label{eq:tcoinssum1}
&\implies \bbE_{\sfM'_{(k,nt)}}\left[\Var\left(\sum_{j\in J_s}X'_j\middle| \sfM'_{(k,nt)}=m'_{(k,nt)}\right)\right]\le |J_s|\cdot (1-\ve). 
\end{align}
Recall that $\Ee^s_u=Y_{q_s(u)}-X'_{q_s(u)}$ represents the random variable for the $u$th input element to subroutine $\Apr^s$ in Algorithm $\sfO$. $\Apr^s_{|J_s|}$ represents the random variable for the output of $\Apr^s$ and let $\Ra^{\Apr^s}_{u}$ represents the private randomness used at the $u$th time-step. Note that, random variable $\Ee^s$ has joint distribution $\calD^s$ on $(\Ee^s_1,\ldots, \Ee^s_{|J_s|})$. With parameters $\gamma=\frac{\ve}{10}$ and $\B=\sqrt{kt\cdot n^{1+\lambda}}$ and the fact that $n/2\le |J_s|\le nt$, input $\Ee^s$ to algorithm $\Apr^s$ satisfies conditions of Proposition \ref{cl:aprsum}, which implies that 
\begin{equation}
\label{eq:tcoinssum2}
\bbE\left[\left(\sum_{u=1}^{|J_s|}\Ee^s_u-\Apr^s_{|J_s|}\right)^2\right]=\bbE_{\e^s\sim \calD^s,r\sim\Ra^{\Apr^s}}\left[\left(\sum_{u=1}^{|J_s|}\e^s_u-\Apr^s(\e^s ,r)\right)^2\right]\le \gamma^2 |J_s|.
\end{equation}
As in proof of Lemma \ref{cl:asum}, we upper bound the expected variance of $\sum_{j\in J_s} Y_j$ conditioned on $\sfO_{nt}$ as follows:
\begin{align}
\nonumber
&\bbE_{\sfO_{nt}}\left[\Var\left(\sum_{j\in J_s} Y_j\middle| \sfO_{nt}=\an_{nt}\right)\right]\\
\nonumber
&\le \bbE_{\sfO_{nt}}\left[\Var\left(\sum_{j\in J_s} Y_j-\sum_{j\in J_s} X'_j\middle| \sfO_{nt}=\an_{nt}\right)\right]
+\bbE_{\sfO_{nt}}\left[\Var\left(\sum_{j\in J_s} X'_j\middle| \sfO_{nt}=\an_{nt}\right)\right]\\
\label{eq:tcoinssum3}
&\;\;\;\;\;\;\;+2\sqrt{\bbE_{\sfO_{nt}}\left[\Var\left(\sum_{j\in J_s} Y_j-\sum_{j\in J_s} X'_j \middle| \sfO_{nt}=\an_{nt}\right)\right]\cdot\bbE_{\sfO_{nt}}\left[\Var\left(\sum_{j\in J_s} X'_j\middle| \sfO_{nt}=\an_{nt}\right)\right]}.
\end{align}
Exactly as in the proof of Lemma \ref{cl:asum}, we use Equation~\eqref{eq:tcoinssum2} and prove that as output of the algorithm $\sfO_{nt}$ contains $\Apr^s_{|J_s|}$,
\[\bbE_{\sfO_{nt}}\left[\Var\left(\sum_{j\in J_s} Y_j-\sum_{j\in J_s} X'_j\middle| \sfO_{nt}=\an_{nt}\right)\right]=\bbE_{\sfO_{nt}}\left[\Var\left(\sum_{u=1}^{|J_s|} \Ee^s_u\middle| \sfO_{nt}=\an_{nt}\right)\right]\le \gamma^2 |J_s|=\frac{\ve^2}{100}|J_s|.\]
Similarly, as the output of algorithm $\sfO$ contains $\sfM'_{(k,nt)}$, by the law of total variance
\begin{align*}
\bbE_{\sfO_{nt}}\left[\Var\left(\sum_{j\in J_s} X'_j\middle| \sfO_{nt}=\an_{nt}\right)\right]\le\bbE_{\sfM'_{(k,nt)}}\left[\Var\left(\sum_{j\in J_s} X'_j\middle| \sfM'_{(k,nt)}=m'_{(k,nt)}\right)\right] \le |J_s|\cdot(1-\ve). \tag{implied by Equation~\eqref{eq:tcoinssum1}}
\end{align*}

Therefore, if $\bbE_{\sfM_{(k,nt)}}\left(\bbE\left[\sum_{j\in J_s} X_j\middle| \sfM_{(k,nt)}=m_{(k,nt)}\right]\right)^2\ge \ve |J_s|$, then by Equation \eqref{eq:tcoinssum3},
\[\bbE_{\sfO_{nt}}\left[\Var\left(\sum_{j\in J_s} Y_j\middle| \sfO_{nt}=\an_{nt}\right)\right]\le |J_s|(1-\ve)+\frac{\ve^2}{100}|J_s|+ 2|J_s|\cdot\frac{\ve}{10}\le |J_s|\left(1-\frac{\ve}{2}\right).\]
This finishes the proof of the claim. 
\end{proof}

\section{Omitted Proofs in Section \ref{sec:mainlemma}}\label{Apen:A}

\subsection{Proof of Claim \ref{lem:lmx}}

We recall the statement:
\lemlmx*
\begin{proof}

Given $(i,j)$, we define $$\widetilde{\calA}_{neighbor}:=(\calA_{(\leq i-1,p_{j+1}-1)},\calA_{(\leq i,p_{j}-1)}),$$ and $$\widetilde{\calA}_{non-neigh}:=(\calA_{(\leq i,p_{1}-1)},\calA_{(\leq 
 i,p_{2}-1)},\cdots,\calA_{(\leq i,p_{j-1}-1)},\calA_{(\leq i-1,p_{j+2}-1)},\cdots,\calA_{(\leq i-1,p_{m}-1)}).$$ 
 Notice that $$(\widetilde{\calA}_{neighbor},\widetilde{\calA}_{non-neigh})=(\calA_{(i,p_{1}-1)},\cdots,\calA_{(i,p_{j}-1)},\calA_{(<i,p_{[m]}-1)}),$$
 where $\calA_{(<i,p_{[m]}-1)}$ denotes $(\calA_{(<i,p_{1}-1)},\calA_{(<i,p_{2}-1)},\cdots,\calA_{(<i,p_{m}-1)})$. 
Then, we have the following claim: 
\begin{restatable}[]{claim}{claimblock}
\label{claim:block}
Conditioned on $(\widetilde{\calA}_{neighbor},{X}_{p_{j}})$, $\calA_{i,p_{j+1}-1}$ is independent with $(\widetilde{{X}}_{\neq j},\widetilde{\calA}_{non-neigh})$. That is,
\[
\I(\widetilde{{X}}_{\neq j},\calA_{non-neigh};\calA_{(i,p_{j+1}-1)}\mid \widetilde{\calA}_{neighbor},{X}_{p_j})=0.
\]
As a corollary, 
\[
\I(\widetilde{{X}}_{\neq j};\calA_{(i,p_{j+1}-1)}\mid \calA_{(i,p_{1}-1)},\cdots,\calA_{(i,p_{j}-1)},\calA_{(<i,p_{[m]}-1)},{X}_{p_j})=0,
\]
where $\widetilde{{X}}_{\neq j}$ denotes $({X}_{p_{1}},\cdots,{X}_{p_{j-1}},{X}_{p_{j+1}},\cdots,{X}_{p_{m}})$.
\end{restatable}
\noindent
We defer the proof of this claim to the end of this section. With this property, we have
\begin{align*}
\I(\widetilde{\calA};\widetilde{{X}})=&\sum_{i=1}^{k-1}\sum_{j=1}^m \I(\calA_{(i,p_{j+1}-1)};\widetilde{{X}} \mid \calA_{(i,p_{1}-1)},\cdots,\calA_{(i,p_{j}-1)},\calA_{(<i,p_{[m]}-1)})\\
= & \sum_{i=1}^{k-1}\sum_{j=1}^m \I(\calA_{(i,p_{j+1}-1)};{X}_{p_j} \mid \calA_{(i,p_{1}-1)},\cdots,\calA_{(i,p_{j}-1)},\calA_{(<i,p_{[m]}-1)})
\\ &+\I(\calA_{(i,p_{j+1}-1)};\widetilde{{X}}_{\neq j} \mid \calA_{(i,p_{1}-1)},\cdots,\calA_{(i,p_{j}-1)},\calA_{(<i,p_{[m]}-1)},{X}_{p_j})
\\
=&\sum_{i=1}^{k-1}\sum_{j=1}^m \I(\calA_{(i,p_{j+1}-1)};{X}_{p_j} \mid \calA_{(i,p_{1}-1)},\cdots,\calA_{(i,p_{j}-1)},\calA_{(<i,p_{[m]}-1)})\\
=&\sum_{i=1}^{k-1}\sum_{j=1}^m \I(\calA_{(i,p_{j+1}-1)};{X}_{p_j} \mid \widetilde{\calA}_{neighbor},\widetilde{\calA}_{non-neighbor})\\
\leq &\sum_{i=1}^{k-1}\sum_{j=1}^m \I(\calA_{(i,p_{j+1}-1)};{X}_{p_j} \mid \widetilde{\calA}_{neighbor})
\end{align*}
as desired. The first two equality are by the chain rule. The third equality is due to Claim \ref{claim:block}, and the inequality is by Property \ref{itemmi2} and $\I(\widetilde{\calA}_{non-neigh};\calA_{(i,p_{j+1}-1)} \mid \widetilde{\calA}_{neighbor},{X}_{p_j})=0$.    
\end{proof}
\begin{proof}[\textbf{Proof of Claim \ref{claim:block}}]
We strengthen the statement by proving $$\I\big(( {X} _{[1,{p_{j}-1}]}, {X} _{[p_{j+1},n]},\Ra_{(\leq k, [1,{p_{j}-1}])}, \Ra_{(\leq k,[p_{j+1},n])}); ( {X} _{[{p_{j}},p_{j+1}-1]}, {\Ra} _{(\leq k,[{p_{j}},p_{j+1}-1])}) \mid \widetilde{ \calA}_{neighbor}, {X} _{p_j} \big)=0.$$ This strengthened statement is similar to Claim \ref{claim:micindependence}, and can be proved in a similar way. We omit the details here. Then, the original statement follows by the statement above since $\widetilde{ \calA}_{non-neigh},\widetilde{ {X} }_{\neq j}$  could be fully determined by $( {X} _{[1,{p_{j}-1}]}, {X} _{[p_{j+1},n]},\Ra_{(\leq k, [1,{p_{j}-1}])}, \Ra_{(\leq k,[p_{j+1},n])})$ condition on $\widetilde{ \calA}_{neighbor}$ and $ {X} _{p_j}$, while $\calA_{(i,p_{j+1}-1)}$ could be fully determined by $( {X} _{[{p_{j}},p_{j+1}-1]}, {\Ra} _{(\leq k,[{p_{j}},p_{j+1}-1])})$ condition on $\widetilde{ \calA}_{neighbor}$ and $ {X} _{p_j}$. We  then finish the proof. 
\end{proof}

\subsection{Information Complexity for \textsf{MostlyEq}}\label{sec:ICforMostlyEq}
We first restate the result for the \textsf{MostlyEq} problem: 
\MostlyEq*

Proving Lemma \ref{lem:ICresultforCCproblem} directly is challenging. We use the information complexity lower bound for the $\text{AND}_k$ problem from communication complexity. 

\begin{definition}
In the $\text{AND}_k$ setting, there are $k$ players, the $i$-th player has a single bit input $y_i\in \{0,1\}$ in hand. A good protocol $\Pi$ to resolve $\text{AND}_k$ computes whether the input is $0^k$ or $1^k$ with a low error rate. We formally define the error rate as follows, \[Err_{\Pi}(0^k,1^k):=\Pr[\Pi(1^k)=0]+\Pr[\Pi(0^k)=1]\]
\end{definition}
\noindent In \cite{gronemeier2009asymptotically,Jay09}, the author proved the following lemma, obtaining a lower bound of the information complexity on the $\text{AND}_k$ problem. 

\begin{lemma}[\cite{gronemeier2009asymptotically,Jay09}]\label{lem:AND}
    For any protocol $\Pi$ with that
    \[
    Err_{\Pi}(0^k,1^k)\leq 0.2
    \]
    we have that,
    \[
    \sum_{i=1}^k\mutualent\big(\boldsymbol{U}_i;\Pi(\boldsymbol{U}_i)\big) \geq \Omega(1),
    \]
    where $\Pi(\boldsymbol{U}_i)$ denotes the transcripts of $\Pi$ under the input distribution $\boldsymbol{U}_i$. 
\end{lemma}

Define $\boldsymbol{U}_i$ to be the distribution that the input samples from $\{0^k,e_i\}$ uniformly at random, where $e_i$ denotes the input with $y_i = 1$ and $\forall j\neq i, y_j=0$.

Therefore, we introduce some random variables here and break down the whole information complexity using our newly defined random variables.

Assuming $\boldsymbol{z}=(\boldsymbol{z}_1,\boldsymbol{z}_2,\cdots,\boldsymbol{z}_m)\sim \boldsymbol{P}_U$, then, for each $\boldsymbol{z}_i$, we redefine the sample process of it by a list of random variables $\{\boldsymbol{X}_{i,j},\boldsymbol{Y}_{i,j}\}$, where $i\in[m]$ and $j\in [t
]$. For every pair $(i,j)$, we define $\boldsymbol{X}_{i,j}$ by:

\begin{equation}
\label{eq6}
\boldsymbol{X}_{i,j}=\left\{
\begin{aligned}
1 & ,~~~~  \text{ with probability } \frac{\beta}{n}\\
0 & ,~~~~  \text{ with probability }1-\frac{\beta}{n}, 
\end{aligned}
\right.\nonumber
\end{equation}
where $\beta$ is a constant to be determined later, and define $\boldsymbol{Y}_{i,j}$ to be a uniform bit on $\{0,1\}$.  We sample all $\boldsymbol{X}_{i,j},\boldsymbol{Y}_{i,j}$ independently. We sample $\boldsymbol{z}_i$ as follows:

\begin{framed}
\begin{center}
Sample $\boldsymbol{z}_i$ by auxiliary random variables
\end{center}    
\begin{enumerate}
    \item Randomly sample $\boldsymbol{X}_{i,1},\dots,\boldsymbol{X}_{i,n}$ and $\boldsymbol{Y}_{i,1},\dots,\boldsymbol{Y}_{i,n}$.
    \item Use $T$ to denote the set of such $j\in[t]$ that satisfies $\boldsymbol{X}_{i,j}=\boldsymbol{Y}_{i,j}=1$. 
    \item If $T\neq \emptyset$, uniformly sample (and output) an element from $T$. 
    \item Otherwise, uniformly sample (and output) an element from $[t]$.
\end{enumerate}
\end{framed}

It is easy to verify that this distribution is identical to $\boldsymbol{P}_U$ regardless of the choice of $\beta$, since every $\boldsymbol{z}_i$ is independently and uniformly sampled from $[t]$. There is another important observation about $\{\boldsymbol{X}_{i,j}\}$ and $\{\boldsymbol{Y}_{i,j}\}$: condition on $\boldsymbol{X}_{i,j}=\boldsymbol{Y}_{i,j}=1$, $\boldsymbol{z}_i$ equals $j$ with constant probability; condition on $\boldsymbol{X}_{i,j}=0$ or $\boldsymbol{Y}_{i,j}=0$, $\boldsymbol{z}_i$ is close to the uniform distribution. Intuitively, we can say that $\boldsymbol{X}_{i,j}$ and $\boldsymbol{Y}_{i,j}$ contain a large amount of information about whether $\boldsymbol{z}_i$ equals $j$. 

With those auxiliary random variables, we can prove Lemma \ref{lem:ICresultforCCproblem} by the following steps:
\begin{enumerate}
    \item \label{apenstep:1} Decompose the information complexity via $\{\boldsymbol{X}_{i,j},\boldsymbol{Y}_{i,j}\}$ by showing that
    $$I \big(\Pi(\boldsymbol{P}_U);\boldsymbol{P}_U \big) \geq I\left(\Pi(\boldsymbol{P}_U);\widetilde{\boldsymbol{X}},\widetilde{\boldsymbol{Y}}\right)\geq \sum_{i=1}^nI\left(\Pi(\boldsymbol{P}_U);\widetilde{\boldsymbol{X}}_j,\widetilde{\boldsymbol{Y}}_j\right),$$
    where $\widetilde{\boldsymbol{X}}_j:=(\boldsymbol{X}_{1,j},\boldsymbol{X}_{2,j},\cdots,\boldsymbol{X}_{m,j})$ and $\widetilde{\boldsymbol{Y}}_j:=(\boldsymbol{Y}_{1,j},\boldsymbol{Y}_{2,j},\cdots,\boldsymbol{Y}_{m,j})$. 
    \item \label{apenstep:2} Lower bound the increment, arguing that for most $j\in[t]$, it holds $$I\left(\Pi(\boldsymbol{P}_U); \widetilde{\boldsymbol{X}}_j,\widetilde{\boldsymbol{Y}}_j\right)\geq \Omega(1/n).$$
\end{enumerate}

In step \ref{apenstep:1}, the first inequality comes from the fact that $\widetilde{\boldsymbol{X}},\widetilde{\boldsymbol{Y}} \rightarrow \boldsymbol{z} \rightarrow \Pi(\boldsymbol{P}_U)$ forms a Markov's chain together with the data processing inequality. The second inequality comes from the fact that $$I(\boldsymbol{A};\boldsymbol{B},\boldsymbol{C})\geq I(\boldsymbol{A};\boldsymbol{B})+I(\boldsymbol{A};\boldsymbol{C})$$ if $\boldsymbol{B},\boldsymbol{C}$ are independent. 
\noindent
The proof of step \ref{step:2} is more technical, we prove it via the following two claims. 
\begin{claim}\label{claimmany}
If protocol $\Pi$ distinguishes $\boldsymbol{P}_U$ and $\boldsymbol{P}_{Eq}$ with error rate $Err_{\Pi}(\boldsymbol{P}_U,\boldsymbol{P}_{Eq})\leq 0.1$, there exists more than $n/3$ different $j$s, such that
\[
    Err_{\Pi}(\boldsymbol{P}_U,\boldsymbol{P}_{Eq}^j)\leq 0.15, 
\]
where $\boldsymbol{P}_{Eq}^j$ is defined by the distribution of $\boldsymbol{P}_{Eq}$ condition on the needle equals $j$. 
\end{claim}
\noindent
and, 
\begin{restatable}[]{claim}{claimlast}\label{claim:singleincrement}
If protocol $\Pi$ distinguishes $\boldsymbol{P}_U$ and $\boldsymbol{P}_{Eq}^j$ with error rate $Err_{\Pi}(\boldsymbol{P}_U,\boldsymbol{P}_{Eq}^j)\leq 0.15$,
then 
\[
I\left(\Pi(\boldsymbol{P}_U);\widetilde{\boldsymbol{X}}_j,\widetilde{\boldsymbol{Y}}_j\right)=\Omega(1/t).
\]
\end{restatable}
The first claim can be easily proved by the fact that $\boldsymbol{P}_{Eq}= \sum_{j=1}^t\frac{1}{t}\boldsymbol{P}_{Eq}^j$ and Markov's Inequality. However, the proof for Claim \ref{claim:singleincrement} is more complicated, involving a reduction to the $\text{AND}_m$ communication problem. The detailed proof is to be discussed later. 
With the two claims above, we can easily prove the Lemma \ref{lem:ICresultforCCproblem} by: 
\[I\left(\Pi(\boldsymbol{P}_U);\boldsymbol{P}_U\right)\geq \frac{t}{3} \cdot \Omega(\frac{1}{t})=\Omega(1).\]

\noindent
Now, it suffices to show Claim \ref{claim:singleincrement}, and
we first restate it: 
\claimlast*

We prove this claim by performing a reduction to the $\text{AND}_m$ gadget. We create a communication protocol $\Pi'$ that solves the $\text{AND}_m$ problem using the \textsf{MostlyEq} protocol $\Pi$, and use the information complexity lower bound for $\text{AND}_m$ to establish the claim.
\begin{proof}[\textbf{ Proof of Claim \ref{claim:singleincrement}}]
    Assume $\Pi$ be any protocol which distinguishes $\boldsymbol{P}_U$ from $\boldsymbol{P}_{Eq}^j$ with error rate $Err_{\Pi}(\boldsymbol{P}_U,\boldsymbol{P}_{Eq}^j) \leq 0.15$. We use $\Pi$ to construct a protocol $\Pi'$ for $\text{AND}_{m}$ that distinguishes between $0^m$ and $1^m$ with high accuracy.

      Recall that the input of $\text{AND}_m$ is a vector $y\in \{0,1\}^{m}$, where $y_i\in \{0,1\}$ denotes the input of the $i$-th player. The idea is to construct a random vector $\boldsymbol{z}'$ using $y\in \{0,1\}^{m}$, and feed $\boldsymbol{z}'$ to the protocol $\Pi'$.
      
      The vector $y$ translates $\boldsymbol{z}'$ in the following way:  For each $\boldsymbol{z}_i'$, the corresponding sampling process depends on $y_i$ as follows: 
        \begin{equation}
        \label{eq7}
        \boldsymbol{z}_i':=\left\{
        \begin{aligned}
        \text{$\boldsymbol{z}_i$ conditioned on $\boldsymbol{X}_{p_i,j}=\boldsymbol{Y}_{p_i,j}=0$} & ~~~~,  \text{if $y_i=0$}\\
        \text{$\boldsymbol{z}_i$ conditioned on $\boldsymbol{X}_{p_i,j}=\boldsymbol{Y}_{p_i,j}=1$} & ~~~~,  \text{if $y_i=1$}
        \end{aligned}
        \right.\nonumber
        \end{equation}
    Recall the sampling process of $\boldsymbol{z}_i$ with $\boldsymbol{X}_{i,j},\boldsymbol{Y}_{i,j}$. We know that:
    \begin{itemize}
        \item When $y_i=1$, i.e., $\boldsymbol{X}_{p_i,j}=\boldsymbol{Y}_{p_i,j}=1$, $\boldsymbol{z}_i'$ equals the $j$ with a constant probability. We can set a proper constant value for the parameter $\beta$ here, allowing that $\boldsymbol{z}_i'$ equals the $j$ with probability $1/2$ and otherwise uniformly samples from $[t]$. Notice that  $\beta\in[1,10]$ here. 
        \item When $y_i=0$, i.e., $\boldsymbol{X}_{p_i,j}=\boldsymbol{Y}_{p_i,j}=0$, $\boldsymbol{z}_i'$ is close to the uniform sample from $[t]$. 
    \end{itemize} 
    The two facts are crucial in bounding $Err_{\Pi'}(0^m,1^m)$ later. $\boldsymbol{z}'$ is a random variable here, and we use $z'$ to denote an instance sampled by $\boldsymbol{z}'$. 
    After sampling a vector $z'$ based on $y$, the protocol $\Pi'$ simulates $\Pi$ on the input $z'$.

    First, we would like to bound $Err_{\Pi'}(0^m,1^m)$. When $y = 1^m$, $\boldsymbol{z'}$ shares the same distribution with $\boldsymbol{P}_{Eq}^j$ by definitions, and when $y = 0^m$, $\boldsymbol{z'}$ is somehow close to the uniform distribution $\boldsymbol{P}_U$. We use $\boldsymbol{P}_U'$ to denote its distribution here. We compute the total variance distance between the two distributions: 
    \begin{align*}
        \lVert \boldsymbol{P}_U - \boldsymbol{P}_U'\rVert_{TV} &=\frac{1}{2}\sum_q |\Pr[\boldsymbol{P}_U=q] - \Pr[\boldsymbol{P}_U'=q]|\\
        &\leq 1-(1-\frac{1}{n})^m\\
        &\leq m/t\\
        &\leq 0.01.
    \end{align*}
    The last inequality holds since $m\leq  t/100$. 
    Thus, we know that this protocol $\Pi$ satisfies that $Err_{\Pi'}(0^m,1^m)\leq 0.2$ from the fact that 
    \begin{align*}
        Err_{\Pi'}(AND_{m}) &= Err_{\Pi}(\boldsymbol{P}_U',\boldsymbol{P}_{Eq}^j)\\
        &\leq Err_{\Pi}(\boldsymbol{P}_U,\boldsymbol{P}_{Eq}^j) + 2\lVert \boldsymbol{P}_U - \boldsymbol{P}_U'\rVert_{TV}. 
    \end{align*}
    The last inequality comes from the property of communication protocols.
    Thus, we have
    
    \begin{align*}
        I\left(\Pi(\boldsymbol{P}_U);\widetilde{\boldsymbol{Y}}_j|\widetilde{\boldsymbol{X}}_j\right)&=
        \sum_{x} \Pr[\widetilde{\boldsymbol{X}}_j=x]I\left(\Pi(\boldsymbol{P}_U);\widetilde{\boldsymbol{Y}}_j|\widetilde{\boldsymbol{X}}_j=x\right)\\
        &\geq \sum_{e_i} \Pr[\widetilde{\boldsymbol{X}}_j=e_i] I\left(\Pi(\boldsymbol{P}_U);\widetilde{\boldsymbol{Y}}_j|\widetilde{\boldsymbol{X}}_j=e_i\right)\\
        &\geq \frac{\beta}{t}\cdot(1-\frac{\beta}{t})^{m}\sum_{e_i} I\left(\Pi(\boldsymbol{P}_U);\boldsymbol{Y}_{i,j}|\widetilde{\boldsymbol{X}}_j=e_i\right)
        \\
        &\geq \Omega(1/t).
    \end{align*}
   In the second inequality and third inequality, $e_i$ stands for a special instance of $\widetilde{\boldsymbol{X}}_j$, which satisfies that only $\boldsymbol{X}_{i,j}$ equals $1$ while others equal $0$. 
   The last inequality is from: when we let the input $y$ to $\Pi'$ to be a random variable which is uniformly sampled from $\{0,0,\cdots,0\}$ and $\{0,0,\cdots,y_i=1,\cdots,0\}$ and use $\boldsymbol{U}_i$ to denote its distribution here, we have \[I\bigg(\Pi'(\boldsymbol{U}_i);\boldsymbol{U}_i\bigg)=I\bigg(\Pi(\boldsymbol{P}_U);\boldsymbol{Y}_{i,j}|\widetilde{\boldsymbol{X}}_j=e_i\bigg),\]
   from the construction of $\Pi'$. Then, together with Lemma \ref{lem:AND}, $m\leq t/100$ and $\beta \in[1,10]$, we know that \[\frac{\beta}{t}\cdot(1-\frac{\beta}{t})^{m}\sum_{e_i} I\left(\Pi(\boldsymbol{P}_U);\boldsymbol{Y}_{i,j}|\widetilde{\boldsymbol{X}}_j=e_i\right)\geq \Omega(1/t).\]
    Thus, we have that
    \begin{align*}
        \Omega (1/t)\leq I\left(\Pi(\boldsymbol{P}_U);\widetilde{\boldsymbol{Y}}_j|\widetilde{\boldsymbol{X}}_j\right)\leq I\left(\Pi(\boldsymbol{P}_U);\widetilde{\boldsymbol{X}}_j,\widetilde{\boldsymbol{Y}}_j\right)
    \end{align*}
    as desired. 
\end{proof}

\section{Omitted Proofs in Section \ref{sec:upper_bound}}\label{sec:omittedupperboundproofs}
\subsection{Proof for Claim \ref{claim:expdecay}}
In this subsection, we recall and prove Claim \ref{claim:expdecay}. 
\expdecay*
\begin{proof}
    First, for a counter $(c_1,c_2,c_3)$ after $r$ rounds, $c_3$ can be written in the following form:
    \[c_3 = \sum_{i=1}^r inc_i.\]
    Here, $inc_i$s are independent random variables taking values from $\{0,1\}$, denoting the increase of $c_3$ in $i$-th round. From some easy calculations, we know $\mathbb{E}[inc_i] \leq C_1/C_2$. Thus, we know that $\mathbb{E}[c_3] = \frac{C_1}{C_2}\cdot r$. From Hoeffding's Inequality, 
    \[\Pr[c_3> r/3]=\Pr[c_3-\mathbb{E}[c_3]>  (1/3-C_1/C_2)r]\leq e^{-r/5}.\]
    The last inequality holds since $C_1/C_2\leq 0.01$. 
\end{proof}
\subsection{Analysis of $\calA_2$}
In this subsection, we analyze algorithm $\calA_2$ when $p\leq \frac{1}{\sqrt{n \log^3 n}}$, and prove: 
\generalp*
In the following lemmas and analysis, we always assume $p\leq \frac{1}{\sqrt{n \log^3 n}}$. For simplicity of analysis, we first assume $\calA_2$ runs on an infinitely long data stream sampled by uniform distribution or the needle distribution with probability $p$, where we use $\boldsymbol{D}_0'$ or $\boldsymbol{D}_1'$ to denote them. We analysis the performance of $\calA_2$ on $\boldsymbol{D}_0'$ and $\boldsymbol{D}_1'$, then extend the result to $\boldsymbol{D}_0$ and $\boldsymbol{D}_1$. To begin with, we bound the space complexity of $\calA_2$ on $\boldsymbol{D}_0'$ and $\boldsymbol{D}_1'$: 
\begin{lemma}
    When input data stream follows $\boldsymbol{D}_0'$, the memory of $\calA_2$ exceeds $\frac{C_0'}{p^2n}$ with probability at most $1/n$, where $C'_0$ is a constant.
\end{lemma}
\begin{proof}
    The memory usage of $\calA_2$ comes from two parts: 
    \begin{itemize}
        \item $O(\log n)$ space to store the index of current data group;
        \item space for each $\calA^{\ell}$ to store those counters. 
    \end{itemize}
    Note that when $p\leq \frac{1}{\sqrt{n\log^3 n}}$, the space for the index is $O(\log n)=O(\frac{1}{p^2n})$. It suffices to show that the total memory used by $\calA^{\ell}$s would not exceed $O(\frac{1}{p^2n})$ with high probability. 

    While $\calA^{\ell}$ processing the data group $i$, we define the size of memory used by $\calA^{\ell}$ as $X^{\ell}_i$, and define $X_i$ as $X_i:=\sum_{\ell=1}^{\frac{1}{p^2n}} X^{\ell}_i$. We first calculate $\mathbb{E}[X^{\ell}_i]$ and $\Var[X^{\ell}_i]$, and then give a concentration of $X_i$ by Bernstein's Inequality. To begin with, we first give an important claim. 
    \begin{claim}\label{claim:expdecay2}
        Under $\boldsymbol{D_0}'$, a counter of $\calA^{\ell}$ survives $r>10000$ rounds with probability at most $e^{-10^{-6}r}$. 
    \end{claim}
    The proof is similar to the proof of Claim \ref{claim:expdecay}, and we omit it here. With this claim, we can conclude that: 
    \begin{align*}
        \mathbb{E}[X_i] = \frac{1}{p^2n}\mathbb{E}[X^{\ell}_i] \leq  \frac{1}{p^2n}\Theta\bigg( C_2 \sum_{j=1}^{i} (\log(i-j)+1)e^{10^{-6}(j-i)}\bigg) = \Theta(\frac{1}{p^2n}). 
    \end{align*}
    Here, the first equality comes from the linearity of expectation. The first inequality comes from the calculation of $\mathbb{E}[X_i^{\ell}]$ together with the fact that for a counter with lifespan $i-j$, we only need $O(\log (i-j) +1)$ space to store it. Similarly, we have: 
    \begin{align*}
        \Var[X_i] = \frac{1}{p^2n} \Var[X_i^{\ell}] \leq \frac{1}{p^2n} \mathbb{E}[(X_i^{\ell})^2] \leq O(\frac{1}{p^2n}).
    \end{align*}
    Again, the first equality comes from the independence, and the last inequality comes from $\mathbb{E}[(X_i^{\ell})^2]$, which could be deduced by a similar analysis of $\mathbb{E}[X_i^{\ell}]$. 
    From the definitions of $\calA^{\ell}$, we know that $X_i^{\ell} \leq O(\log^2 n)$ since no counter with lifespan larger than $3\cdot 10^6 \log n$ exists. Then, we apply Bernstein's Inequality on $X_i - \mathbb{E}[X_i]$: 
    \begin{align*}
        \Pr[X_i \geq \mathbb{E}[X_i] + x] \leq \exp\bigg(-\frac{0.5\cdot x^2}{\Var[X_i]+1/3\cdot O(\log^2 n)x}\bigg).
    \end{align*}
    When $p\leq \frac{1}{\sqrt{n \log^3 n}}$, if we take $x=C\log^{1.5} n \sqrt{\mathbb{E}[X_i]}$ for a large enough constant $C$, $\Pr[X_i \geq \mathbb{E}[X_i] + x] \leq 1/n^2$ holds. Also, $\Pr[\exists i, X_i \geq \mathbb{E}[X_i] + x]\leq 1/n$ follows by a union bound. Note that $x=O(\frac{1}{p^2n})$ when $p \leq \frac{1}{\sqrt{n \log^3 n}}$ holds. Thus, we conclude that there exists a constant $C'$ such that: the probability that $\calA_2$'s memory exceeds $\frac{C'}{p^2n}$ is smaller than $1/n$. 
\end{proof}
With similar analysis, we can prove the following lemma, bounding the memory used by $\calA_2$ under distribution $\boldsymbol{D}_1'$. To avoid repetition, we omit the proof here. 
\begin{lemma}
     When input data stream follows $\boldsymbol{D}_1'$, the memory of $\calA_2$ exceeds $\frac{C_1'}{p^2n}$ with probability at most $1/n$, where $C'_1$ is a constant.
\end{lemma}

\noindent
Then, it suffices to bound the error rate of $\calA_2$. Explicitly, we prove the two following lemmas: 
\begin{lemma}
    $\Pr[\calA_2(\boldsymbol{D}_0')=1] = o(1)$.
\end{lemma}
\begin{proof}
   We separately consider the probability $\Pr[\calA^{\ell}(\boldsymbol{D}_0')=1]$, and then have $$\Pr[\calA_2(\boldsymbol{D}_0')=1] \leq \frac{1}{p^2n} \cdot \Pr[\calA^{\ell}(\boldsymbol{D}_0')=1]\leq n \cdot \Pr[\calA^{\ell}(\boldsymbol{D}_0')=1]. $$ by a union bound. 

   By Claim \ref{claim:expdecay2} and a union bound, we can conclude that :
    \[
    \Pr[\calA^{\ell}(\boldsymbol{D}_0')=1] \leq e^{-3\cdot 10^6\log n \cdot 10^{-6}} \cdot C_2 p n \leq C_2/n^2. 
    \]
    Hence, we have $\Pr[\calA_2(\boldsymbol{D}_0' )=1]=o(1)$ as desired. 
\end{proof}
\begin{lemma}
    $\Pr[\calA_2(\boldsymbol{D}_1')=0] \leq e^{-C_1} + 0.1$
\end{lemma}
\begin{proof}
    First, without loss of generality, we assume the needle $\alpha$ lies in $\mathcal{B}^1$, the analysis for other cases is the same out of symmetry. 
    Similarly to Lemma \ref{lem:A11}, we first define two events: 
    \begin{enumerate}
        \item we define $A$ as the event that $\calA_2$ fails to begin tracking the needle, namely for any group $i$, $\alpha\notin h_1^1(i)$;
        \item we define $B$ as the counter of $\calA_2$ to track the potential needle fails to survive $15 \log n$ rounds. 
    \end{enumerate}
    From some calculations, we have
    \begin{align*}
        \Pr[A] \leq (1-\frac{C_1}{pn})^{pn} \leq e^{-C_1}. 
    \end{align*}
    It suffices to bound $\Pr[B]$.  Note that for every group $i$, the probability that at least one needle appears in $i$ is bigger than $$ \Pr[\text{group $i$'s size exceeds $\frac{1}{4p}$}]\cdot \bigg(1-(1-p)^{\frac{1}{4p}}\bigg)\geq (1-4p)^{\frac{1}{4p}} \cdot (1-e^{-1/4})\geq 0.05.$$
    
    Assume the counter for the needle is $(c_1,c_2,c_3)$, we  know that the expectation value of $c_3$ after $r$ rounds, the corresponding random variable denoted by $c_3^r$, is bigger than $$\mathbb{E}[c_3^r]\geq 0.05r.$$
    Together with Hoeffding Inequality, we know that the probability that \[\Pr[c_3^r \leq r/100] \leq e^{-r/1000}.\]
    Summing over $10000 \leq r \leq 3\cdot 10^6 \log n $, we know that 
    \[
    \Pr[B]\leq \sum_{r = 10000}^{ 3\cdot 10^6 \log  n } \Pr[c_3^r \leq r/100] \leq 0.1. 
    \]
    We then conclude the lemma by: 
    \[
    \Pr[A]+\Pr[B]\leq e^{-C_1} +0.1
    \]
\end{proof}

On the other hand, the probability that the number of data used by $\calA_2$ exceeds $n$ can be bounded by Markov's Inequality: $\leq 0.25$ since the expected number of used data is $\frac{1}{4p} \cdot pn = \frac{n}{4}$. Combining those facts together, we conclude that: 
\begin{align*}
    Err_{\calA_2}(\boldsymbol{D}_0,\boldsymbol{D}_1)&:= \Pr[\calA_2(\boldsymbol{D_0})=1]+ \Pr[\calA_2(\boldsymbol{D_1})=0] \\
    &\leq \Pr[\calA_2(\boldsymbol{D_0}')=1]/0.75+ \Pr[\calA_2(\boldsymbol{D_1}')=0]/0.75\\
    &\leq 0.2. 
\end{align*}
The last inequality holds for large enough constant $C_1$. Also, the probability that $\calA_2$'s memory exceeds $\frac{C'}{p^2n}$, where we define $C':=\max\{C_0',C_1'\}$, is at most $o(1)/0.75 = o(1)$.

\end{document}